\documentclass[a4paper,UKenglish,cleveref, autoref, thm-restate]{lipics-v2021}
\title{Taking Bi-Intuitionistic Logic First-Order: A Proof-Theoretic Investigation via Polytree Sequents} %TODO Please add

\titlerunning{Taking Bi-Intuitionistic Logic First-Order} %TODO optional, please use if title is longer than one line

\author{Tim S. Lyon}{Technische Universit\"at Dresden, Germany \and \url{https://sites.google.com/view/timlyon} }{timothy_stephen.lyon@tu-dresden.de}{https://orcid.org/0000-0003-3214-0828}{European Research Council, Consolidator Grant DeciGUT (771779).}%TODO mandatory, please use full name; only 1 author per \author macro; first two parameters are mandatory, other parameters can be empty. Please provide at least the name of the affiliation and the country. The full address is optional. Use additional curly braces to indicate the correct name splitting when the last name consists of multiple name parts.

\author{Ian Shillito}{The Australian National University, Canberra, Ngunnawal \& Ngambri Country, Australia}{ian.shillito@anu.edu.au}{https://orcid.org/0009-0009-1529-2679}{}

\author{Alwen Tiu}{The Australian National University, Canberra, Ngunnawal \& Ngambri Country, Australia}{alwen.tiu@anu.edu.au}{https://orcid.org/0000-0002-2695-5636}{}

\authorrunning{T.S. Lyon, I. Shillito, and A. Tiu} %TODO mandatory. First: Use abbreviated first/middle names. Second (only in severe cases): Use first author plus 'et al.'

\Copyright{Timothy Stephen Lyon, Ian Shillito, and Alwen Tiu} %TODO mandatory, please use full first names. LIPIcs license is "CC-BY";  http://creativecommons.org/licenses/by/3.0/


 %TODO mandatory; please add comma-separated list of keywords

\category{} %optional, e.g. invited paper

\relatedversion{} %optional, e.g. full version hosted on arXiv, HAL, or other respository/website
\nolinenumbers %uncomment to disable line numbering

\EventEditors{J\"{o}rg Endrullis and Sylvain Schmitz}
\EventNoEds{2}
\EventLongTitle{33rd EACSL Annual Conference on Computer Science Logic (CSL 2025)}
\EventShortTitle{CSL 2025}
\EventAcronym{CSL}
\EventYear{2025}
\EventDate{February 10--14, 2025}
\EventLocation{Amsterdam, Netherlands}
\EventLogo{}
\SeriesVolume{326}
\ArticleNo{41}
\usepackage{synttree, bussproofs, stmaryrd, rotating, xcolor}
\usepackage{esvect,pgf, tikz, color}
\usetikzlibrary{arrows, automata}
\usepackage{algorithm2e}
\usepackage[all]{xy}
\usepackage{changepage,enumerate,proof,upgreek,relsize,trfsigns}

\usepackage{multicol}
\usepackage{comment}

\newif\ifshow % toggle true or false based on if want to hide section
\showfalse % hide the sections

\ifshow
  \includecomment{wrap}
\else
  \excludecomment{wrap} % anything wrapped in a {wrap} environment will be excluded
\fi

\usepackage{tikz}
\usetikzlibrary{trees}
\usetikzlibrary{fit}
\usetikzlibrary{positioning,arrows,calc,decorations.markings}
\tikzset{
modal/.style={>=stealth',shorten >=1pt,shorten <=1pt,auto,node distance=1.5cm,semithick},
world/.style={circle,draw,minimum size=0.5cm,fill=gray!15},
point/.style={circle,draw,inner sep=0.5mm,fill=black},
reflexive above/.style={->,loop,looseness=7,in=120,out=60},
reflexive below/.style={->,loop,looseness=7,in=240,out=300},
reflexive left/.style={->,loop,looseness=7,in=150,out=210},
reflexive right/.style={->,loop,looseness=7,in=30,out=330}%,
}

\renewcommand{\phi}{\varphi}

\newtheorem{innercustomlem}{Lemma}
\newenvironment{customlem}[1]
  {\renewcommand\theinnercustomlem{#1}\innercustomlem}
  {\endinnercustomlem}
  
  \newtheorem{innercustomthm}{Theorem}
\newenvironment{customthm}[1]
  {\renewcommand\theinnercustomthm{#1}\innercustomthm}
  {\endinnercustomthm}
  
\newtheorem{innercustomprop}{Proposition}
\newenvironment{customprop}[1]
  {\renewcommand\theinnercustomprop{#1}\innercustomprop}
  {\endinnercustomprop}

\newcommand{\iffi}{\textit{iff} }

\newcommand{\dfn}{Definition}

\newcommand{\thm}{Theorem}
\newcommand{\lem}{Lemma}
\newcommand{\prp}{Proposition}
\newcommand{\sect}{Section}

\newcommand{\fig}{Figure}
\newcommand{\app}{Appendix}
\newcommand{\rmk}{Remark}

\newcommand{\id}{\mathrm{ID}}

\newcommand{\cd}{\mathrm{CD}}

\newcommand{\calc}{\mathsf{LBIQ}}
\newcommand{\icalc}{\mathsf{NIQ}}
\newcommand{\nid}{\mathsf{LBIQ}(\idclass)}

\newcommand{\ncd}{\mathsf{LBIQ}(\cdclass)}

\newcommand{\R}{\mathcal{R}}

\newcommand{\T}{\mathcal{T}}

\newcommand{\lab}{\mathrm{Lab}}
\newcommand{\mint}{\iota}
\newcommand{\seq}{S}
\newcommand{\seqb}{S'}

\newcommand{\pred}{\mathrm{Pred}} %\Upphi}
\newcommand{\var}{\mathrm{Var}}
\newcommand{\VT}[1]{VT(#1)} %Variables appearing in a term
\newcommand{\FV}[1]{FV(#1)} %Free variables of an expression

\newcommand{\func}{\mathrm{Fun}}
\newcommand{\termset}{\mathrm{Ter}}
\newcommand{\parama}{a}

\newcommand{\imp}{\rightarrow} %supset}

\newcommand{\exc}{\mathrel{%
  \hspace{.1ex}
  \begin{tikzpicture}[baseline=-.57ex, line width=.130ex]
    \draw[-] (-0.1ex,0) --(1.5ex,0);
    \draw[-, line width=.01ex, fill=black]
             (1.35ex,0) -- (1.84ex, .48ex)
                       -- (1.91ex ,.418ex)
                       -- (1.55ex,   0ex)
                       -- (1.91ex ,-.418ex)
                       -- (1.84ex,-.48ex)
                       -- (1.35ex,0ex);
  \end{tikzpicture}
\hspace{.1ex}}}

\newcommand{\lang}{\mathcal{L}}
\newcommand{\sub}[1]{(#1)}

\newcommand{\ari}[1]{ar(#1)}

\newcommand{\univ}{U}

\newcommand{\idcond}{\mathrm{(ID)}}

\newcommand{\cdcond}{\mathrm{(CD)}}
\newcommand{\mcond}{\mathrm{(M)}}
\newcommand{\concond}{\mathrm{(C_{1})}}
\newcommand{\funcond}{\mathrm{(C_{2})}}
\newcommand{\funinterp}[0]{I_{1}} %Function interpretation function
\newcommand{\predinterp}[0]{I_{2}} %Predicate interpretation function
\newcommand{\terminterp}[1]{\overline{#1}} %Term interpretation function
\newcommand{\idclass}[0]{\mathcal{ID}} %Class of ID models
\newcommand{\cdclass}[0]{\mathcal{CD}} %Class of CD models
\newcommand{\assign}{\alpha}
\newcommand{\modclass}{\mathcal{C}}

\newcommand{\biqid}{\mathsf{BIQ}(\idclass)}

\newcommand{\biqcd}{\mathsf{BIQ}(\cdclass)}
\newcommand{\il}{\mathsf{IP}}
\newcommand{\ilq}{\mathsf{IQ}}
\newcommand{\ilqc}{\mathsf{IQC}}
\newcommand{\hb}{\mathsf{BIP}}
\newcommand{\hbq}{\biqid} %\mathsf{BIQ}}
\newcommand{\hbqc}{\biqcd} %\mathsf{BIQC}}

\newcommand{\prf}{\pi}
\newcommand{\sr}{R}

\newcommand{\sar}{\vdash}

\newcommand{\inp}{\mathrm{L}}
\newcommand{\outp}{\mathrm{R}}

\newcommand{\lrel}[2]{{#1}R{#2}}

\newcommand{\lseq}[4]{#1,#2,#3\sar#4} %In 1 we have relation atoms, 2 we have labeled terms, 3 and  labeled formulae (antecedent and succedent)
\newcommand{\lterm}[2]{#1\!:\!#2} %In 1 we put the label, 2 put the term
\newcommand{\lform}[2]{#1\!:\!#2} %In 1 we put the label, 2 put the formula

\newcommand{\ax}{(ax)}
\newcommand{\botl}{(\bot{\inp})}
\newcommand{\topr}{(\top{\outp})}

\newcommand{\botr}{(\bot{\outp})}
\newcommand{\topl}{(\top{\inp})}
\newcommand{\disr}{(\lor{\outp})}
\newcommand{\conr}{(\land{\outp})}

\newcommand{\impl}{({\imp}{\inp})}
\newcommand{\impr}{({\imp}{\outp})}
\newcommand{\excl}{({\exc}{\inp})}
\newcommand{\excr}{({\exc}{\outp})}
\newcommand{\existsr}{(\exists{\outp})}
\newcommand{\existsri}{(\exists{\outp})}

\newcommand{\allr}{(\forall{\outp})}
\newcommand{\disl}{(\lor{\inp})}
\newcommand{\conl}{(\land{\inp})}
\newcommand{\existsl}{(\exists{\inp})}
\newcommand{\alll}{(\forall{\inp})}
\newcommand{\allli}{(\forall{\inp})}

\newcommand{\doms}{(ds)}

\newcommand{\psub}{(t/x)} %{(ps)}
\newcommand{\iwk}{(iw)}

\newcommand{\ctrr}{(ctr_{r})}
\newcommand{\ctrl}{(ctr_{l})}

\newcommand{\cut}{(cut)}

\newcommand{\gax}{(gax)}

\newcommand{\mrg}{(mrg)}

\newcommand{\lwr}{(lwr)}
\newcommand{\lft}{(lft)}

\newcommand{\rone}{(r_{1})}
\newcommand{\rtwo}{(r_{2})}

\newcommand{\idr}{(id)}
\newcommand{\cdr}{(cd)}
\newcommand{\wkv}{(wv)}
\newcommand{\brf}{(br_{f})}
\newcommand{\brb}{(br_{b})}

\newcommand{\rable}[2]{#1 \twoheadrightarrow^{*}_{\R} #2}
\newcommand{\notrable}[2]{#1 \begin{normalsize}\not\end{normalsize}\twoheadrightarrow^{*}_{\R} #2}
\newcommand{\rtable}[2]{#1 \twoheadrightarrow^{+}_{\R} #2}

\newcommand{\fsa}{\Gamma}
\newcommand{\fsb}{\Delta}
\newcommand{\fsc}{\Sigma}

\newcommand{\prove}{\mathtt{Prove}}
\newcommand{\success}{\mathtt{True}}
\newcommand{\branch}{\mathcal{B}}

\newcommand{\avail}[3]{\mathsf{A}(#1,#2,#3)}

\newcommand{\lj}{\mathsf{LJ}}

\begin{document}

\maketitle

\begin{abstract}
It is well-known that extending the Hilbert axiomatic system for first-order intuitionistic logic with an exclusion operator, that is dual to implication, collapses the domains of models into a constant domain. This makes it an interesting problem to find a sound and complete proof system for first-order bi-intuitionistic logic with non-constant domains that is also conservative over first-order intuitionistic logic. We solve this problem by presenting the first sound and complete proof system for first-order bi-intuitionistic logic with increasing domains. We formalize our proof system as a polytree sequent calculus (a notational variant of nested sequents), and prove that it enjoys cut-elimination and is conservative over first-order intuitionistic logic. A key feature of our calculus is an explicit eigenvariable context, which allows us to control precisely the scope of free variables in a polytree structure. Semantically this context can be seen as encoding a notion of Scott's existence predicate for intuitionistic logic. This turns out to be crucial to avoid the collapse of domains and to prove the completeness of our proof system. The explicit consideration of the variable context in a formula sheds light on a previously overlooked dependency between the residuation principle and the existence predicate in the first-order setting, which may help to explain the difficulty in designing a sound and complete proof system for first-order bi-intuitionistic logic. 
\end{abstract}

%The Introduction
\section{Introduction}

%%%NOTES
%1. Contributions:
%(i)  We are the first to characterize BiInt with Non-Constant Domains syntactically (i.e. validities recursively enumerable) and 
%(ii) We find that residuation beaks down in increasing domains (and since Hilbert systems usually encode this property, it explains why a Hilbert system has yet to be found).

%General presentation of bi-int
Propositional bi-intuitionistic logic ($\hb$), also referred to as 
\emph{Heyting-Brouwer logic}~\cite{Rau80}, is a conservative extension of propositional intuitionistic logic ($\il$), obtained by adding the binary connective $\exc$ (referred to as \emph{exclusion})\footnote{
%Exclusion $\exc$ is 
Also referred to as \emph{pseudo-difference}~\cite{Rau80}, \emph{subtraction}, and \emph{co-implication}~\cite{GorPosTiu08}.
} among the traditional intuitionistic connectives. This logic has proven relevant in computer science, having a formulae-as-types interpretation in terms of first-class coroutines~\cite{Cro04} and where modal extensions have found import in image processing~\cite{SteSchRyd16}. While in intuitionistic logic the connectives $\land$ and $\imp$ form a residuated pair, 
i.e. $(\phi \land \psi) \imp \chi$ is valid \iffi $\phi \imp (\psi \imp \chi)$ is valid \iffi $\psi \imp (\phi \imp \chi)$ is valid,
% i.e. the validity of $(\phi \land \psi) \imp \chi$, $\phi \imp (\psi \imp \chi)$ and $\psi \imp (\phi \imp \chi)$ are equivalent,
in bi-intuitionistic logic the connectives $\lor$ and $\exc$ also form a residuated pair, 
i.e. $\phi \imp (\psi \lor \chi)$ is valid \iffi $(\phi \exc \psi) \imp \chi$ is valid \iffi $(\phi \exc \chi) \imp \psi$ is valid.\footnote{However, they are not logically equivalent, e.g., $[\phi \imp (\psi \lor \chi)] \imp [(\phi \exc \psi) \imp \chi]$ is not valid.
%, so in general we cannot replace a formula with its residuated counterpart in an arbitrary context while preserving validity.
} To put it succinctly, $\hb$ is a bi-intuitionistic extension of $\il$ that is (1) conservative and (2) has the residuation property, i.e. $(\land, \imp)$ and $(\exc, \lor)$ form residuated pairs. %Motivated by this observation, it is natural to wonder if such a relationship holds in the first-order setting as well, that is, does a bi-intuitionistic extension of first-order intuitionistic logic ($\ilq$) exist that is conservative and has the residuation property?

When extending first-order intuitionistic logic ($\ilq$) to its bi-intuitionistic counterpart, a `natural' axiomatization seems to be one obtained by adding the universal axioms (Ax1) $\forall x\varphi\rightarrow\varphi(t/x)$, (Ax2) $\forall x(\psi\rightarrow\varphi)\rightarrow(\psi\rightarrow\forall x\varphi)$ (where $x$ is not free in $\psi$), and the rule (Gen) $\varphi / \forall x\varphi$ to the axioms of $\hb$. This extension, which we refer to as the logic $\hbqc$, turns out \emph{not} to be conservative over first-order intuitionistic logic $\ilq$, as it allows one to prove  the \emph{quantifier shift axiom} $\forall x (\phi \lor \psi) \imp \forall x \phi \lor \psi$ (where $x$ is not free in $\psi$), which is not valid intuitionistically. A proof of the quantifier shift axiom is given below, where MP stands for modus ponens, Res stands for the residuation property described above, and $\delta:=\forall x ((\forall x (\phi \lor \psi) \exc \psi) \imp \phi) \imp ((\forall x (\phi \lor \psi) \exc \psi) \imp \forall x \phi)$.  
\begin{small}
\begin{center}
\AxiomC{}
\RightLabel{Ax1}
\UnaryInfC{$\forall x (\phi \lor \psi) \imp (\phi \lor \psi)$}
\RightLabel{Res}
\UnaryInfC{$(\forall x (\phi \lor \psi) \exc \psi) \imp \phi$}
\RightLabel{Gen}
\UnaryInfC{$\forall x((\forall x (\phi \lor \psi) \exc \psi) \imp \phi)$}
\AxiomC{}
\RightLabel{Ax2}
\UnaryInfC{$\delta$}
\RightLabel{MP}
\BinaryInfC{$(\forall x (\phi \lor \psi) \exc \psi) \imp \forall x \phi$}
\RightLabel{Res}
\UnaryInfC{$\forall x (\phi \lor \psi) \imp \forall x \phi \lor \psi$}
\DisplayProof
\end{center}
\end{small}
It is well-known that the quantifier shift axiom characterizes the class of first-order intuitionistic Kripke models with \emph{constant domains}~\cite{GabSheSkv09,Grz64}, thus forcing the models for $\hbqc$ to satisfy this constraint. Indeed, various works in the literature (e.g., \cite{Rau77,Res05}) have shown that completeness for $\hbqc$ requires the domain to be constant. These works and the above example strongly suggest that it might not be possible to have a %sound and complete 
proof system for a bi-intuitionistic logic with non-constant domains, at least not as a traditional Hilbert system. As far as we know, there is no prior successful attempt at solving this problem.

In this paper, we provide the first sound and complete proof system for first-order bi-intuitionistic logic with increasing domains, which we refer to here as $\hbq$. With some minor modifications, the proof system for $\hbq$ can be converted into a proof system for $\hbqc.$
%, thereby solving the open problem of providing a syntactic characterization of first-order bi-intuitionistic logic.
 A key insight in avoiding the collapse of domains in $\hbq$ is to consider the universal quantifier as implicitly carrying an assumption about the existence of the quantified variable. Proof theoretically, this could be done by introducing a notion of an {\em existence predicate}, first studied by Scott~\cite{Scott2006identity}. An existence predicate such as $E(x)$ postulates that $x$ exists in the domain under consideration. By insisting that all universally quantified variables be guarded by an existence predicate, i.e. universally quantified formulae would have the form  $\forall x ( E(x) \imp \phi(x))$, the quantifier shift axiom can be rewritten as: 
$\forall x (E(x) \imp (\varphi \lor \psi))  \imp (\forall x (E(x) \imp \varphi) \lor \psi)$.
%, which is not valid. In particular, 
Attempting a bottom-up construction of a derivation similar to our earlier example for this rewritten axiom,
%quantifier shift axiom, 
we get stuck at the the top-most residuation rule, which is in fact not a valid instance of Res:
\begin{small}
\begin{center}
\AxiomC{$E(x) \imp [\forall x( E(x) \imp (\phi \lor \psi)] \imp (\phi \lor \psi)$}
\RightLabel{Res}
\UnaryInfC{$E(x) \imp [ (\forall x( E(x) \imp (\phi \lor \psi)) \exc \psi] \imp \phi$}
\RightLabel{Gen}
\UnaryInfC{$\forall x(E(x) \imp [\forall x( E(x) \imp (\phi \lor \psi)) \exc \psi] \imp \phi)$}
\AxiomC{$\cdots$}
\RightLabel{MP}
\BinaryInfC{$[\forall x(E(x) \imp (\phi \lor \psi)) \exc \psi] \imp \forall x(E(x) \imp \phi)$}
\RightLabel{Res}
\UnaryInfC{$\forall x(E(x) \imp (\phi \lor \psi)) \imp (\forall x ( E(x)\imp \phi) \lor \psi)$}
\DisplayProof
\end{center}
\end{small}
For the proof construction to proceed, we would have to somehow discharge the assumption $E(x)$ in the premise of Gen before applying the residuation rule. In the logic of constant domains $\hbqc$, $E(x)$ is equivalent to $\top$ (i.e. the interpretation of any term in the logic is an object that exists in all worlds in the underlying Kripke model). So the version of the quantifier shift axiom with the existence predicate is provably equivalent to the original one in $\hbqc.$
This is not the case, however, in the logic of increasing domains $\hbq$, since the assumption $E(x)$ cannot always be discharged.  
What this example highlights is that a typical proof-theoretical argument used to show the provability of the quantifier shift axiom (and hence the collapse of domains)
%, like the one we have seen earlier, 
implicitly depends on an existence assumption on objects in the domains in the underlying Kripke model.
What we show here is that by making this dependency explicit and by carefully managing the use of the existence assumptions in proofs, we are able to obtain a sound and complete proof system for $\hbq.$

One issue with %the introduction of 
the existence predicate is that it is not clear how it should interact with the exclusion operator. Semantically, a formula like 
$
\forall x [E(x) \imp ((p(x) \exc \exists y (E(y) \land p(y))) \imp \bot)]
$
asserts that, if an object $x$ exists in the current domain, then postulating that $p(x)$ holds in a predecessor world should imply that $x$ exists as well in that predecessor world. This is valid in our semantics, but it was not at all obvious how a proof system that admits this tautology, and does not also degenerate into a logic with constant domains, should be designed. We shall come back to this example later in Section~\ref{sec:nested-calculi}.  Additionally, the existence predicate poses a problem when proving conservativity over first-order intuitionistic logic that does not feature this predicate. We overcome this remaining hurdle by enriching sequents with an explicit variable context, which can be seen as essentially encoding the existence predicate, while avoiding introducing it explicitly in the language of formulae.

\begin{comment}
\begin{figure}[t]
    \centering

\begin{tikzpicture}
%\node[world,fill=white]
%$\myoverset{Con}{Res}$
\node[] (w0) {$\il$};
\node[] (w1) [above=of w0,yshift=-.25cm] {$ \hb$};

\node[] (w2) [right=of w0,xshift=5em] {$\ilq$};
\node[] (w3) [above left=of w2,yshift=-.35cm,xshift=2em] {$ \hbq$};
\node[] (w4) [above right=of w2,yshift=-.35cm,xshift=-2em] {$ \hbqc$};

\node[] (w5) [right=of w2,xshift=5em] {$\ilqc$};
\node[] (w6) [above=of w5,yshift=-.35cm] {$\hbqc$};

\path[->,draw] (w0) -- (w1) node [midway,left] {C, R};

\path[->,draw] (w2) -- (w3) node [midway,left] {C, $\neg$R};
\path[->,draw] (w2) -- (w4) node [midway,right] {$\neg$C, R};

\path[->,draw] (w5) -- (w6) node [midway,right] {C, R};

\end{tikzpicture}

\caption{Map of (bi-)intuitionistic logics. A directed arrow from one logic to another indicates that the latter is an extension of the former. We use C to indicate a conservative extension and R to indicate }\label{fig:bi-int-map}
\end{figure}
\end{comment}

The proof systems for $\hbq$ and $\hbqc$ are both formalized using \emph{polytree sequents}~\cite{CiaLyoRamTiu21}, which are connected binary graphs whose vertices are traditional Gentzen sequents and which are free of (un)directed cycles. Polytree sequents are a restriction of traditional labeled sequents~\cite{Sim94,Vig00} and are notational variants of nested sequents~\cite{Bul92,Kas94,Bru09}. (NB. For details on the relationship between polytree and nested sequents, see~\cite{CiaLyoRamTiu21}.) Nested sequents were introduced independently by Bull~\cite{Bul92} and Kashima~\cite{Kas94} and employ trees of Gentzen sequents in proofs. Both polytree sequents and nested sequents allow for simple formulations of proof systems for various non-classical logics that enjoy important proof theoretical properties such as cut-elimination and subformula properties. Such systems have also found a range of applications, being used in knowledge integration algorithms~\cite{LyoAlv22}, serving as a basis for constructive interpolation and decidability techniques~\cite{LyoTiuGorClo20,LyoKar24,TiuIanGor12}, and even being used to solve open questions about axiomatizability~\cite{IshKik07}. We make use of polytree sequents in our work as they admit a formula interpretation (at least in the intuitionistic case), which can be leveraged for direct translations of proofs into sequent calculus or Hilbert calculus proofs. %proofs or proofs in a Hilbert system.

% The formalism of nested sequents is rather elegant, yielding proof calculi that require minimal syntactic bureaucracy, have compact proofs, and where fundamental proof-theoretic properties such as the height-preserving admissibility of important structural rules, the height-preserving invertibility of rules, and syntactic cut-elimination typically hold. Due to these nice aesthetic and computational properties, nested sequent systems have found a range of applications, being used in knowledge integration algorithms~\cite{LyoAlv22}, serving as a basis for constructive interpolation and decidability techniques~\cite{FitKuz15,LyoTiuGorClo20,TiuIanGor12}, and even being used to solve open questions about axiomatizability~\cite{IshKik07}.

The calculi for $\hbq$ and $\hbqc$ are based on these richly structured sequents, which internalize the existence predicate into syntactic components, called {\em domain atoms}, present in each node of the sequent. The rich structure of these sequents is exploited by special rules within our calculi called \emph{reachability rules}, which traverse paths in a polytree sequent, propagating and/or consuming data. We demonstrate that our calculi enjoy the height-preserving invertibility of every rule, and show that a wide range of novel and useful structural rules are height-preserving admissible, culminating in a non-trivial proof of cut-elimination. 

\textbf{Outline of Paper.} %The paper is organized as follows: 
 In \sect~\ref{sec:log-prelims-I}, we define a semantics for first-order bi-intuitionistic logic with increasing domains $\hbq$ and constant domains $\hbqc$. In \sect~\ref{sec:nested-calculi}, we define our polytree sequent calculi showing them sound and complete relative to the provided semantics. In \sect~\ref{sec:properties}, we establish %height-preserving
 admissibility and invertibility results as well as prove a non-trivial syntactic cut-elimination theorem. We conclude and discuss future work in \sect~\ref{sec:conclusion}. Due to space constraints, we defer most proofs to the appendix. %Due to space constraints, most proofs have been deferred to the online appended version~\cite{LyoShiTiu24}.

%Introducing Logics, Semantics, Axiomatizations
\section{Logical Preliminaries}\label{sec:log-prelims-I}

In this section, we introduce the language, models, and semantics for first-order bi-intuitionistic logic with increasing domains, dubbed $\biqid$, and with constant domains, dubbed $\biqcd$. Let $\var := \{x, y, z, \ldots\}$ be a countably infinite set of \emph{variables} and $\func = \{f, g, h, \ldots\}$ be a countably infinite set of \emph{function symbols} containing countably many function symbols of each arity $n \in \mathbb{N}$. We let $\ari{f} = n$ denote that the arity of the function symbol $f$ is $n$ and let $a, b, c, \ldots$ denote \emph{constants}, which are function symbols of arity $0$. For a set $X \subseteq \var$, we define the set $\termset(X)$ of \emph{$X$-terms} to be the smallest set satisfying the following two constraints: (1) $X \subseteq \termset(X)$, and (2) if $f \in \func$, $f$ is of arity $n$, and $t_{1}, \ldots, t_{n} \in \termset(X)$, then $f(t_{1}, \ldots, t_{n}) \in \termset(X)$. The complete set of terms $\termset$ is defined to be $\termset(\var)$. We use $t$, $s$, $\ldots$ (potentially annotated) to denote ($X$-)terms and let $\VT{t}$ denote the set of variables occurring in the term $t$. We will often write a list $t_{1}, \ldots, t_{n}$ of terms as $\vec{t}$, and define $\VT{\vec{t}} = \VT{t_{1}} \cup \cdots \cup \VT{t_{n}}$.
 
We let $\pred := \{p, q, \ldots\}$ be a countably infinite set of predicates containing countably many predicates of each arity $n \in \mathbb{N}$. We denote the arity of a predicate $p$ as $\ari{p}$ and refer to predicates of arity $0$ as \emph{propositional atoms}. An \emph{atomic formula} is a formula of the form $p(t_{1}, \ldots, t_{n})$, obtained by prefixing a predicate $p$ of arity $\ari{p} = n$ to a tuple of terms of length $n$. We will often write atomic formulae $p(t_{1}, \ldots, t_{n})$ as $p(\vec{t})$. %Our language is defined as:

\begin{definition}[The Language $\lang$]\label{def:language}  The \emph{language} $\lang$ is defined to be the set of formulae generated via the following grammar in Backus-Naur form: %T: Hi all, I used the circ symbol below so the BNF grammar fits in the column
$$
\phi ::= p(\vec{t}) \ | \ \bot \ | \ \top %\ | \ (\phi \circ \phi) 
\ | \ \phi \land \phi \ | \ \phi \lor \phi \ | \ \phi \exc \phi  \ | \ \phi \imp \phi \ | \ \exists x \phi \ | \ \forall x \phi
$$
where %$\circ \in \{\lor, \land, \imp, \exc\}$, 
$p$ ranges over $\pred$, the terms $\vec{t} = t_{1}, \ldots, t_{n}$ range over $\termset$, and $x$ ranges over the set $\var$. We use $\phi$, $\psi$, $\chi$, $\ldots$ to denote formulae. % and $\Gamma$, $\Delta$, $\ldots$ to denote (multi)sets of formulae \ian{decision notation}.
\end{definition}

The occurrence of a variable $x$ in $\phi$ is defined to be \emph{free} given that $x$ does not occur within the scope of a quantifier binding $x$. We let $\FV{\phi}$ denote the set of all free variables occurring in the formula $\phi$ and use $\phi(x_{1}, \ldots, x_{n})$ to denote that $\FV{\phi} = \{x_{1}, \ldots, x_{n}\}$. We let $\phi\sub{t/x}$ denote the formula obtained by replacing each free occurrence of the variable $x$ in $\phi$ by $t$, potentially renaming bound variables to avoid unwanted variable capture; e.g. $(\forall y p(x,y))\sub{y/x} = \forall z p(y,z)$. The \emph{complexity} of a formula $\phi$, written $|\phi|$, is recursively defined as follows: (1) $|p(t_{1}, \ldots, t_{n})| = |\bot| = |\top| := 0$, (2) $|Q x \phi| := |\phi| + 1$ for $Q \in \{\forall, \exists\}$, and (3) $|\phi \circ \psi| := |\phi| + |\psi| + 1$ for $\circ \in \{\lor, \land, \imp, \exc\}$.
 
Following \cite{Rau77}, we give a Kripke-style semantics for $\biqid$, defining the models used first, and explaining how formulae are evaluated over them second.

\begin{definition}[$\id$-Frame]\label{def:frame} An \emph{$\id$-frame} (or, \emph{frame}) is a tuple $F = (W,\leq,U,D)$ such that:
\begin{itemize}

\item $W$ is a non-empty set $\{w, u, v, \ldots\}$ of \emph{worlds};

\item $\leq \ \subseteq W \times W$ is a reflexive and transitive binary relation;%\footnote{The properties imposed on $\leq$ are defined as follows: (reflexivity) for all $w \in W$, $w \leq w$, and (transitivity) for all $w, u, v \in W$, if $w \leq v$ and $v \leq u$, then $w \leq u$.}

\item $\univ$ is a non-empty set referred to as the \emph{universe};

\item $D : W\to \mathcal{P}(\univ)$ is a \emph{domain function} mapping each $w \in W$ to a non-empty set $D(w) \subseteq \univ$ with $U = \bigcup_{w \in W} D(w)$, which satisfies the \emph{increasing domain condition}: $\idcond$ If $w \leq u$, then $D(w) \subseteq D(u)$.
\end{itemize}
%We will refer to $\id$-frames as \emph{frames} for simplicity.
\end{definition}

\begin{definition}[$\id$-Model]\label{def:model}
We define an \emph{$\id$-Model} (or, \emph{model}) $M$ to be an ordered triple $(F,\funinterp,\predinterp)$ such that:
\begin{itemize}
\item $F = (W,\leq,U,D)$ is a frame;

\item $\funinterp$ is a function interpreting each function symbol $f\in \func$ such that $\ari{f} = n$ by a function $\funinterp(f): \univ^{n} \to \univ$, satisfying two conditions: $\concond$ For each $w \in W$ and constant $a$, $\funinterp(a) \in D(w)$, and $\funcond$ For each $w \in W$, $\vec a \in D(w)^n$ \iffi $\funinterp(f)(\vec a) \in D(w)$.

\item $\predinterp$ is a function interpreting, in each $w\in W$, each predicate $p \in \pred$ such that $\ari{p} = n$ by a set $\predinterp(w,p)\subseteq D(w)^n$, satisfying the following \emph{monotonicity condition}: $\mcond$ If $w \leq u$, then $\predinterp(w,p) \subseteq \predinterp(u,p)$.
\end{itemize}
\end{definition}

%\ian{Explain condition (F) as it is strong.}

%We argue that $\concond$ and $\funcond$ are sensible conditions to impose on the interpretation of functions.

\begin{definition}[$M$-assignment]
Let $M=(F,\funinterp,\predinterp)$ be a model. We define an \emph{$M$-assignment} to be a function $\alpha : \var\rightarrow\univ$. We note $\alpha[\parama/x]$ is the function $\alpha$ modified in $x$ such that $\alpha[\parama/x](x)=\parama$ and $\alpha[\parama/x](y)=\alpha(y)$ if $y\not= x$. Given an $M$-assignment $\alpha$, we define the interpretation of $t$ in $M$ given $\alpha$, denoted $\terminterp{\alpha}(t)$, inductively as follows: $\terminterp{\alpha}(x) := \alpha(x)$ and $\terminterp{\alpha}(f(t_1,...,t_n)) := \funinterp(f)(\terminterp{\alpha}(t_1),...,\terminterp{\alpha}(t_n))$.
\iffalse
\begin{center}
\begin{tabular}{r c l}
$\terminterp{\alpha}(x)$ & $=$ & $\alpha(x)$\\
$\terminterp{\alpha}(f(t_1,...,t_n))$ & $=$ & $\funinterp(f)(\terminterp{\alpha}(t_1),...,\terminterp{\alpha}(t_n))$\\
\end{tabular}
\end{center}
\fi
\end{definition}

%\cite{GabSheSkv09}
\begin{definition}[Semantics]
\label{def:semantic-clauses} Let $M = (W,\leq,U,D,\funinterp,\predinterp)$ be a model with $w \in W$ and $\alpha$ an $M$-assignment. The satisfaction relation $\Vdash$ is defined as follows:
\begin{itemize}

\item $M,w,\alpha \Vdash p(t_{1}, \ldots, t_{n})$ \iffi $(\terminterp{\alpha}(t_{1}), \ldots, \terminterp{\alpha}(t_{n})) \in \predinterp(w,p)$;

\item $M,w,\alpha \not\Vdash \bot$;

\item $M,w,\alpha \Vdash \top$;

\item $M,w,\alpha \Vdash \phi \lor \psi$ \iffi $M,w,\alpha \Vdash \phi$ or $M,w,\alpha \Vdash \psi$;

\item $M,w,\alpha \Vdash \phi \land \psi$ \iffi $M,w,\alpha \Vdash \phi$ and $M,w,\alpha \Vdash \psi$;

\item $M,w,\alpha \Vdash \phi \exc \psi$ \iffi there exists a $u \in W$ such that $u \leq w$, $M,u,\alpha \Vdash \phi$, and $M,u,\alpha \not\Vdash \psi$;

\item $M,w,\alpha \Vdash \phi \imp \psi$ \iffi for all $u \in W$, if $w \leq u$ and $M,u,\alpha \Vdash \phi$, then $M,u,\alpha \Vdash \psi$;

\item $M,w,\alpha \Vdash \exists x \phi$ \iffi there exists an $\parama \in D(w)$ such that $M,w,\alpha[\parama/x] \Vdash \phi$;

\item $M,w, \alpha \Vdash \forall x \phi$ \iffi for all $u \in W$ and all $\parama \in D(u)$, if $w \leq u$, then $M, u, \alpha[\parama/x] \Vdash \phi$;

%\item $M,w \Vdash \phi$ \iffi $M,w,\alpha \Vdash \phi$ for all $M$-assignments $\alpha$ \ian{useful?} \tim{I added this definition as a `stepping stone' to get to the definition of $M \Vdash \phi$ below, the idea being, that we can ultimately define (in the last line of this definition) what is means for a formula to be $\id$-, $\dd$-, or $\cd$-valid};

%\item $M \Vdash \phi$ \iffi $M,w \Vdash \phi$ for all worlds $w \in W$ of $M$.

%\item $M\Vdash \phi$ \iffi $M,w,\alpha \Vdash \phi$ for $M$-assignment $\alpha$ and world $w \in W$ of $M$;

%\item $\fsa\models_g \phi$ \iffi $\forall M.\forall\alpha.(\forall w.M,w,\alpha \Vdash \fsa \Rightarrow \forall w.M,w,\alpha \Vdash \phi)$ \ian{Delete consequence relations? If not, define validity as $\emptyset\models_l\phi$}
\end{itemize}
For a set $\fsa \subseteq \lang$ of formulae, we write $\fsa \Vdash \phi$ \iffi for all models $M$, $M$-assignments $\alpha$, and worlds $w$ in $M$, if $M,w,\alpha \Vdash \psi$ for each $\psi \in \fsa$, then $M,w,\alpha \Vdash \phi$. A formula $\phi$ is \emph{valid} \iffi $\emptyset \Vdash \phi$. Finally, we define the logic $\biqid$ to be the set $\{\phi \ | \ \emptyset\Vdash \phi\}$ of all valid formulae.
\end{definition}

Note that here we define logics as \emph{sets of theorems}, and not \emph{consequence relations}. While this is fit for our purpose, the reader should be warned that historical confusions emerged around this distinction in the case of propositional bi-intuitionistic logic~\cite{GorShi20,Shi23}, notably pertaining to the deduction theorem. %Our work is built in knowledge of these confusions, and hence avoids their pitfalls.

%The following can be shown by a straightforward induction on the complexity of $\phi$.

\begin{proposition}\label{prop:monotonicity} Let $M = (W,\leq,U,D,\funinterp,\predinterp)$ be a model with $\alpha$ an $M$-assignment. For any $\phi \in \lang$, if $M,w,\alpha \Vdash \phi$ and $w \leq u$, then $M,u,\alpha \Vdash \phi$.
\end{proposition}

\begin{remark} We define a \emph{$\cd$-model} to be a model satisfying the \emph{constant domain condition}: $\cdcond$ If $w, u \in W$, then $D(w) = D(u)$. If we impose the $\cdcond$ condition on models, then first-order bi-intuitionistic logic with \emph{constant domains}, dubbed $\biqcd$, can be defined as the set of all valid formulae over the class of $\cd$-models. In what follows, we let $\idclass$ denote the class of $\id$-models and $\cdclass$ denote the class of $\cd$-models.
\end{remark}

\begin{example}
\label{ex:exists-exc}
Consider the formula 
$\forall x( (p(x) \exc \exists y p(y)) \imp \bot)$, discussed in the introduction, but with the existence predicate removed. In the semantics with increasing domains, this formula is valid. To see this, suppose otherwise, i.e. that there exists a world $w$ where the formula is false. Thus, there is a successor $w \leq u$ such that $\bar{\alpha}(x) \in D(u)$ and $p(x) \exc \exists y p(y)$ is true, for some assignment $\alpha$. 
The latter implies that for some $u'$ such that $u' \leq u$, $p(x)$ is true (i.e. $\alpha(x) \in I_P(u', p)$), but $\exists y p(y)$ is false. The former implies that $\alpha(x) \in D(u')$, so by the semantic clause for the $\exists$ quantifier, $\exists y p(y)$ must be true -- contradiction. 
\end{example}

%Introducing labeled Sequent Formalism
\section{Polytree Sequent Systems}\label{sec:nested-calculi}

Let $\lab = \{w, u, v, \ldots\}$ be a countably infinite set of labels. For a formula $\phi \in \lang$ and label $w\in\lab$, we define $\lform{w}{\phi}$ to be a \emph{labeled formula}. We use $\fsa$, $\fsb$, $\fsc$, $\ldots$ to denote finite multisets of labeled formulae, %let $\fsa {\restriction} w$ denote the multiset $\{\phi \ | \ w : \phi \in \fsa\}$, 
and let $\lform{w}{\fsa}$ denote a multiset of labeled formulae all labeled with $w$.
%A \emph{relational atom} is an expression of the form $\lrel{w}{u}$ such that $w,u \in \lab$ and a \emph{domain atom} is an expression of the form $w : x$ such that $w \in \lab$ and $x \in \var$. 
A \emph{relational atom} is an expression of the form $\lrel{w}{u}$ and a \emph{domain atom} is an expression of the form $w : x$, where $w,u \in \lab$ and $x \in \var$. 
Intuitively, the domain atom formalizes an existence predicate: $w : x$ can be interpreted as saying that the interpretation of $x$ exists at world $w$. 
We use $\R$ and $\T$ (and annotated versions thereof) to denote finite multisets of, respectively, relational atoms and domain atoms. 
%We use $\R$ and annotated versions thereof to denote multisets of relational atoms and $\T$ and annotated versions thereof to denote multisets of domain atoms. 
Also, we 
%define $\T {\restriction} w$ to be the multiset $\{x \ | \ w : x \in \T\}$, 
define $\lterm{w}{\VT{t}} = w : x_{1}, \ldots, w : x_{n}$ with $\VT{t} = \{x_{1}, \ldots, x_{n}\}$, define $\lterm{w}{\VT{\vec{t}}} = w : \VT{t_{1}}, \ldots, w : \VT{t_{n}}$ with $\vec{t} = t_{1}, \ldots, t_{n}$, and let $w : \vec{x} = w : x_{1}, \ldots, w : x_{n}$ for $\vec{x} = x_{1}, \ldots, x_{n}$. For multisets $X$ and $Y$ of labeled formulae, relational atoms, and/or domain atoms, we let $X,Y$ denote the multiset union of $X$ and $Y$, and $\lab(X)$ the set of labels occurring in $X$.

\begin{definition}[Polytree Sequent] We define a \emph{polytree sequent} to be an expression of the form $\R,\T,\fsa \sar \fsb$ such that (1) %if $w : \phi \in \fsa, \fsb$ or $w : x \in \T$, then $w \in \lab(\R)$ 
 if $\R \neq \emptyset$, then $\lab(\T, \fsa, \fsb) \subseteq \lab(\R)$ and if $\R = \emptyset$, then $|\lab(\T, \fsa, \fsb)| = 1$, and (2) $\R$ forms a \emph{polytree}, i.e. the graph $G = (V,E)$ such that $V = \lab(\R)$ and $E = \{(w,u) \ | \ \lrel{w}{u} \in \R \}$ is connected and free of both directed and undirected cycles. We refer to $\R,\T,\fsa$ as the \emph{antecedent} and $\fsb$ as the \emph{consequent} of a polytree sequent. We will often refer to \emph{polytree sequents} more simply as \emph{sequents}.
\end{definition}

We sometimes use $S$, $S_0$, $S_1$, $\ldots$ to denote sequents, and for $S = \R,\T,\fsa \sar \fsb$, we define $\lab(S) = \lab(\R,\T,\fsa,\fsb)$. A \emph{flat sequent} is an expression of the form $\T,\fsa \vdash \fsb$ such that $|\lab(\T,\fsa,\fsb)| = 1$, i.e. all labeled formulae and domain atoms share the same label. Polytree sequents encode certain binary graphs whose nodes are flat sequents and such that if you ignore the orientation of the edges, the graph is a tree (cf.~\cite{CiaLyoRamTiu21}). For example, the sequent
%\vspace*{-.5em}
%\begin{flushleft}
$$
S = \underbrace{\lrel{u'}{w}, \lrel{u}{w}, \lrel{w}{v}}_{\R}, \underbrace{u' : x, u : x, u : y, w : z, v : y,}_{\T}
%\end{flushleft}
%\vspace*{-1em}
%\begin{flushright}
\underbrace{w : \phi, w : \psi, v : \theta}_{\fsa} \sar \underbrace{u' : \tau, u : \chi, v : \xi}_{\fsb}
$$
%\end{flushright}
%\vspace*{-.5em}
can be graphically depicted as the polytree $pt(S)$, shown below:
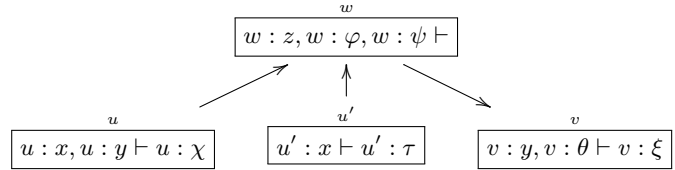
\begin{center}
%\resizebox{\columnwidth}{!}{
%\scalebox{0.9}{
\vspace*{-.25cm}
\begin{tabular}{c}
\xymatrix@C=.15em @R=1.5em{
%\xymatrix{
 & \overset{w}{\boxed{w : z, w : \phi, w : \psi \sar }}\ar@{<-}[dl]\ar@{<-}[d]\ar@{->}[dr]  &  \\
\overset{u}{\boxed{u : x, u : y \sar u : \chi}} & \overset{u'}{\boxed{u' : x \sar u' : \tau}} & \overset{v}{\boxed{v : y, v : \theta \sar v : \xi}}
}
\end{tabular}
%}
%}
\end{center}

\begin{remark}\label{rmk:iso-are-equal} To simplify the proofs of our results in \sect~\ref{sec:properties}, we assume w.l.o.g. that sequents with isomorphic polytree representations are mutually derivable from one another.
\end{remark}

\subsection{Semantics and Proof Systems}

The following definition specifies how to interpret sequents. In essence, we lift the semantics of $\lang$ to sequents by means of `$M$-interpretations', mapping sequents into models.

\begin{definition}[Sequent Semantics]\label{def:sequent-semantics} Let $M = (W,\leq,U,D,\funinterp,\predinterp)$ be a model and $\assign$ an $M$-assignment. We define an \emph{$M$-interpretation} to be a function $\mint$ mapping every label $w \in \lab$ to a world $\iota(w) \in W$. The \emph{satisfaction} of multisets $\R$, $\T$, and $\fsa$ are defined accordingly:
\begin{itemize}

\item $M,\iota,\assign \models \R$ \iffi for all $\lrel{w}{u} \in \R$, $\mint(w) \leq \mint(u)$;

\item $M,\iota,\assign \models \T$ \iffi for all $w : x \in \T$, $\terminterp{\assign}(x) \in D(\mint(w))$;

\item $M,\iota,\assign \models \fsa$ \iffi for all $w : \phi \in \fsa$, $M, \mint(w), \assign \Vdash \phi$.

%\item $M,\iota,\assign \models \R, \T, \fsa \sar \fsb$ \iffi if $M, \iota, \assign \models \R$, $M, \iota, \assign \models \T$, and $M, \iota, \assign \models \fsa$, then there exists a $w : \psi \in \fsb$ such that $M,\iota,\assign \models w : \psi$.

\end{itemize}
We define a sequent $S = \R, \T, \fsa \sar \fsb$ to be \emph{satisfied} on $M$ with $\iota$ and $\assign$, written $M,\iota,\assign \models S$, \iffi if $M, \iota, \assign \models \R$, and $M, \iota, \assign \models \T$, as well as $M, \iota, \assign \models \fsa$, then there exists a $w : \psi \in \fsb$ such that $M,\iota,\assign \models w : \psi$. We write $M,\iota,\assign \not\models S$ when a sequent $S$ is \emph{not satisfied} on $M$ with $\iota$ and $\assign$. A sequent $S$ is defined to be valid \iffi for every model $M$, every $M$-interpretation $\iota$, and every $M$-assignment $\assign$, we have $M,\iota,\assign \models S$; otherwise, we say that $S$ is invalid and write $M,\iota,\assign \not\models S$.
\end{definition}

%labeled version of calculi
\begin{figure*}[t]
%\noindent\hrule

\begin{center}
\begin{tabular}{c c c}
\AxiomC{$\phantom{\fsa}$}
\RightLabel{$\ax^{\dag_{1}}$}
\UnaryInfC{$\lseq{\mathcal R}{\mathcal T}{\fsa,\lform{w}{p(\vec{t})}}{\fsb,\lform{u}{p(\vec{t})}}$}
\DisplayProof

&

\AxiomC{$\phantom{\fsa}$}
\RightLabel{$\botl$}
\UnaryInfC{$\lseq{\mathcal R}{\mathcal T}{\fsa,\lform{w}{\bot}}{\fsb}$}
\DisplayProof 

&

\AxiomC{$\lseq{\mathcal R}{\mathcal T}{\fsa}{\fsb,\lform{w}{\phi},\lform{w}{\psi}}$}
\RightLabel{$\disr$}
\UnaryInfC{$\lseq{\mathcal R}{\mathcal T}{\fsa}{\fsb,\lform{w}{\phi\lor\psi}}$}
\DisplayProof 
\end{tabular}
\end{center}
\begin{center}
\begin{tabular}{c c c}
\AxiomC{$\phantom{\fsa}$}
\RightLabel{$\topr$}
\UnaryInfC{$\lseq{\mathcal R}{\mathcal T}{\fsa}{\fsb,\lform{w}{\top}}$}
\DisplayProof

&

\AxiomC{$\lseq{\mathcal R}{\mathcal T}{\fsa,\lform{w}{\phi},\lform{w}{\psi}}{\fsb}$}
\RightLabel{$\conl$}
\UnaryInfC{$\lseq{\mathcal R}{\mathcal T}{\fsa,\lform{w}{\phi\land\psi}}{\fsb}$}
\DisplayProof

&

\AxiomC{$\lseq{\mathcal R}{\mathcal T,\lterm{w}{y}}{\fsa,\lform{w}{\phi(y/x)}}{\fsb}$}
\RightLabel{$\existsl^{\dag_{2}}$}
\UnaryInfC{$\lseq{\mathcal R}{\mathcal T}{\fsa,\lform{w}{\exists x\phi}}{\fsb}$}
\DisplayProof
\end{tabular}
\end{center}
\begin{center}
\begin{tabular}{c c}
\AxiomC{$\lseq{\mathcal R}{\mathcal T}{\fsa,\lform{w}{\phi}}{\fsb}$}
\AxiomC{$\lseq{\mathcal R}{\mathcal T}{\fsa,\lform{w}{\psi}}{\fsb}$}
\RightLabel{$\disl$}
\BinaryInfC{$\lseq{\mathcal R}{\mathcal T}{\fsa,\lform{w}{\phi\lor\psi}}{\fsb}$}
\DisplayProof

&

\AxiomC{$\lseq{\mathcal R}{\mathcal T}{\fsa}{\fsb,\lform{w}{\phi}}$}
\AxiomC{$\lseq{\mathcal R}{\mathcal T}{\fsa}{\fsb,\lform{w}{\psi}}$}
\RightLabel{$\conr$}
\BinaryInfC{$\lseq{\mathcal R}{\mathcal T}{\fsa}{\fsb,\lform{w}{\phi\land\psi}}$}
\DisplayProof
\end{tabular}
\end{center}

\begin{center}
\begin{tabular}{c c}
\AxiomC{$\lseq{\mathcal R,\lrel{u}{w}}{\mathcal T}{\fsa,\lform{u}{\phi}}{\fsb,\lform{u}{\psi}}$}
\RightLabel{$\excl^{\dag_{3}}$}
\UnaryInfC{$\lseq{\mathcal R}{\mathcal T}{\fsa,\lform{w}{\phi\exc\psi}}{\fsb}$}
\DisplayProof

&
\AxiomC{$\lseq{\mathcal R,\lrel{w}{u}}{\mathcal T}{\fsa,\lform{u}{\phi}}{\fsb,\lform{u}{\psi}}$}
\RightLabel{$\impr^{\dag_{3}}$}
\UnaryInfC{$\lseq{\mathcal R}{\mathcal T}{\fsa}{\fsb,\lform{w}{\phi\imp\psi}}$}
\DisplayProof
\end{tabular}
\end{center}

\begin{center}
\begin{tabular}{c c}
\AxiomC{$\lseq{\mathcal R}{\mathcal T, \lterm{w}{\VT{\vec{t}}}}{\fsa, \lform{w}{p(\vec{t})}}{\fsb}$}
\RightLabel{$\doms$}
\UnaryInfC{$\lseq{\mathcal R}{\mathcal T}{\fsa, \lform{w}{p(\vec{t})}}{\fsb}$}
\DisplayProof

&

\AxiomC{$\lseq{\mathcal R}{\mathcal T}{\fsa}{\fsb,\lform{w}{\exists x\phi},\lform{w}{\phi(t/x)}}$}
\RightLabel{$\existsri^{\dag_{4}}$}
\UnaryInfC{$\lseq{\mathcal R}{\mathcal T}{\fsa}{\fsb,\lform{w}{\exists x\phi}}$}
\DisplayProof
\end{tabular}
\end{center}

\begin{center}
\begin{tabular}{c}
\AxiomC{$\lseq{\mathcal R}{\mathcal T}{\fsa,\lform{w}{\phi\imp\psi}}{\fsb,\lform{u}{\phi}}$}
\AxiomC{$\lseq{\mathcal R}{\mathcal T}{\fsa,\lform{w}{\phi\imp\psi},\lform{u}{\psi}}{\fsb}$}
\RightLabel{$\impl^{\dag_{1}}$}
\BinaryInfC{$\lseq{\mathcal R}{\mathcal T}{\fsa,\lform{w}{\phi\imp\psi}}{\fsb}$}
\DisplayProof
\end{tabular}
\end{center}

\begin{center}
\begin{tabular}{c}
\AxiomC{$\lseq{\mathcal R}{\mathcal T}{\fsa}{\fsb,\lform{u}{\phi\exc\psi},\lform{w}{\phi}}$}
\AxiomC{$\lseq{\mathcal R}{\mathcal T}{\fsa,\lform{w}{\psi}}{\fsb,\lform{u}{\phi\exc\psi}}$}
\RightLabel{$\excr^{\dag_{1}}$}
\BinaryInfC{$\lseq{\mathcal R}{\mathcal T}{\fsa}{\fsb,\lform{u}{\phi\exc\psi}}$}
\DisplayProof
\end{tabular}
\end{center}

\begin{center}
\begin{tabular}{c c}
\AxiomC{$\lseq{\mathcal R}{\mathcal T}{\fsa,\lform{w}{\forall x\phi},\lform{u}{\phi(t/x)}}{\fsb}$}
\RightLabel{$\allli^{\dag_{5}}$}
\UnaryInfC{$\lseq{\mathcal R}{\mathcal T}{\fsa,\lform{w}{\forall x\phi}}{\fsb}$}
\DisplayProof

&

\AxiomC{$\lseq{\mathcal R,\lrel{w}{u}}{\mathcal T,\lterm{u}{y}}{\fsa}{\fsb,\lform{u}{\phi(y/x)}}$}
\RightLabel{$\allr^{\dag_{6}}$}
\UnaryInfC{$\lseq{\mathcal R}{\mathcal T}{\fsa}{\fsb,\lform{w}{\forall x\phi}}$}
\DisplayProof
\end{tabular}
\end{center}

\begin{flushleft}
\textbf{Side Conditions}:
\end{flushleft}

\begin{minipage}{.2\textwidth}
\begin{description}

\item[$\dag_{1} :=$] $\rable{w}{u}$

\item[$\dag_{2} :=$] $y$ is fresh

\end{description}
\end{minipage}
\begin{minipage}{.3\textwidth}
\begin{description}

\item[$\dag_{3} := $] $u$ is fresh

\item[$\dag_{4} := $] $\avail{t}{X_{w}}{\R,\T}$

\end{description}
\end{minipage}
\begin{minipage}{.4\textwidth}
\begin{description}

\item[$\dag_{5} := $] $\rable{w}{u}$ and $\avail{t}{X_{u}}{\R,\T}$

\item[$\dag_{6} := $] $u$ and $y$ are fresh

\end{description}
\end{minipage}

%\hrule
\caption{The System $\nid$.}
\label{fig:nested-calculi}
\end{figure*}

Given a sequent $S = \R, \T, \fsa \sar \fsb$, we define the \emph{term substitution} $S(t/x)$ to be the sequent obtained by replacing (1) every labeled formula $w : \phi$ in $\fsa,\fsb$ by $w : \phi(t/x)$ and (2) $\T$ by $\mathcal T(t/x):= (\mathcal T\setminus\{\lterm{w}{x}\mid\lterm{w}{x}\in\mathcal T\})\cup\{\lterm{w}{y}\mid\lterm{w}{x}\in\mathcal T\text{ and } y \in \VT{t} \}.$ For example, if $S = \lrel{w}{u}, \lterm{w}{x},\lterm{u}{x},\lterm{u}{y}, \lform{w}{p(x)} \sar \lform{u}{\forall y q(x,y)}$, then
$$
S(f(y,z)/x)  = \lrel{w}{u}, \lterm{w}{y},\lterm{w}{z},\lterm{u}{y},\lterm{u}{z},\lterm{u}{y},\lform{w}{p(f(y,z))} \sar \lform{u}{\forall x' q(f(y,z),x')}
$$
where %$\T = \lterm{w}{y},\lterm{w}{z},\lterm{u}{y},\lterm{u}{z},\lterm{u}{y}$ and 
the bound variable $y$ in $\forall y q(x,y)$ was renamed to $x'$ to avoid capture. We now define two \emph{reachability relations} $\rtable{}{}$ and $\rable{}{}$ as well as the notion of \emph{availability}~\cite{Fit14,Lyo23}---all of which are required to properly formulate certain inference rules in our calculi.
 
\begin{definition}[$\rtable{}{}$, $\rable{}{}$] Let $\R$ be a finite multiset of relational atoms such that $w,u \in \lab(\R)$. We say that $u$ is \emph{strictly reachable} from $w$, written $\rtable{w}{u}$, \iffi there exist $v_{1}, \ldots, v_{n} \in \lab(\R)$ such that $\lrel{w}{v_{1}}, \ldots, \lrel{v_{n}}{u} \in \R$ with $n \in \mathbb{N}$. We say that $u$ is \emph{reachable} from $w$, written $\rable{w}{u}$, \iffi $\rtable{w}{u}$ or $w=u$. We write $\notrable{w}{u}$ if $\rable{w}{u}$ does not hold.
\end{definition}

\begin{definition}[Available]\label{def:available} Let $\seq = \R, \T, \fsa \sar \fsb$ be a sequent with $w \in \lab(\seq)$. We define a term $t$ to be \emph{available} for $w$ in $\R,\T$, written $\avail{t}{X_{w}}{\R,\T}$, \iffi $t \in \termset(X_{w})$ such that
$$
X_{w} = \{x \ | \ \lterm{u}{x}\in\mathcal \T \text{ and } \rable{u}{w}\text{ for some $u \in \lab(\seq)$}\}.
$$
\end{definition}

Our polytree calculus $\nid$ for $\biqid$ is shown in \fig~\ref{fig:nested-calculi}. The $\ax$, $\botl$, and $\topr$ rules serve as \emph{initial rules}, the \emph{domain shift rule} $\doms$ encodes the fact that $\predinterp(p,w) \subseteq D(w)^{n}$ in any model. 
%(which is required for completeness; see \app~\ref{app:soundness-completeness}), and the remaining rules are \emph{logical rules}. 
We define the \emph{principal formula} in an inference rule to be the one explicitly mentioned in the conclusion, the \emph{auxiliary formulae} to be the non-principal formulae explicitly mentioned in the premises, and an \emph{active formula} to be either a principal or auxiliary formula. For example, $w : \exists x \phi$ is principal, $w : \phi(t/x)$ is auxiliary, and both are active in $\existsri$. Note that all rules of our calculus preserve the property of being a polytree-structured sequent. We define a \emph{proof} and its \emph{height} as usual~\cite{Tak13}. %We inductively define a \emph{proof} in the usual way: (1) any application of an initial rule is a proof, (2) applying any rule to the conclusion of a proof, or between conclusions of proofs, gives a proof. The \emph{height} of a proof is also defined as usual as the number of sequents occurring along a maximal path in a proof starting from the conclusion and ending at an initial rule; cf.~\cite{Tak13}. \ian{Reduce the above by "We define as usual the notions of proof and height."?} \tim{Nice idea! I made the change above}
 Two unique features of our calculi are the inclusion of {\em reachability rules} and the {\em domain shift rule} ($ds$), which we elaborate on next.

\subsection{Reachability Rules} 

A unique feature of our calculi is the inclusion of \emph{reachability rules} (introduced in~\cite{Lyo21thesis}), a generalization of \emph{propagation rules} (cf.~\cite{CasCerGasHer97,Fit72,GorPosTiu11}), which are not only permitted to propagate formulae throughout a polytree sequent when applied bottom-up, but may also check to see if data exists along certain paths. The rules $\ax$, $\impl$, $\excr$, $\existsr$, and $\alll$ serve as our reachability rules. The side conditions of our reachability rules are listed at the bottom of \fig~\ref{fig:nested-calculi}. Moreover, we define a label $u$ or a variable $y$ to be \emph{fresh} in a rule application (as in the $\existsl$ and $\allr$ rules) \iffi it does not occur in the conclusion of the rule. 

\begin{remark}\label{rmk:ncd} If we set $\dag_{4}$ := `$t \in \termset$', $\dag_{5}$ := `$\rable{w}{u}$ and $t \in \termset$', and remove the $\doms$ rule, then we obtain a polytree calculus, dubbed $\ncd$, for the constant domain version of the logic $\biqcd$. We also note that in the constant domain setting, domain atoms are unnecessary and can be omitted from sequents.
\end{remark}

\begin{figure*}[t]
\begin{small}
\begin{center}
\begin{tabular}{c c}
$\prf \ = $

&

\AxiomC{}
\RightLabel{$\ax$}
\UnaryInfC{$\R, w' : x, u: \forall x( p \lor r(x)), u:  p, v : q \sar u: p ,  w' : r(x)$}
\DisplayProof
\end{tabular}
\end{center}
\begin{center}
%\resizebox{\textwidth}{!}{ 
%\R = w_1 R u, u R v, v R w'
\AxiomC{$\prf$}
\AxiomC{}
\RightLabel{$\ax$}
\UnaryInfC{$\R, w' : x, u: \forall x( p \lor r(x)), u:  r(x), v : q \sar u: p ,  w' : r(x)$}
\RightLabel{$\disl$}
\BinaryInfC{$w R u, u R v, v R w', w' : x, u: \forall x( p \lor r(x)), u:  p \lor r(x), v : q \sar u: p ,  w' : r(x)$}
\RightLabel{$\alll$}
\UnaryInfC{$w R u, u R v, v R w', w' : x, u: \forall x( p \lor r(x)), v : q \sar u: p ,  w' : r(x)$}
\RightLabel{$\allr$}
%\UnaryInfC{$w R u, u R v, v R w', u: \forall x(p \lor r(x)), v : q, w4 : y \sar u: p ,  w' : \forall x r(x)$}
%\RightLabel{$\allr$}
\UnaryInfC{$w R u, u R v, u: \forall x(p \lor r(x)), v : q \sar u: p ,  v : \forall x r(x)$}
\RightLabel{$\impr$}
\UnaryInfC{$w R u, u: \forall x(p \lor r(x)) \sar u: p ,  u : q \imp \forall x r(x)$}
\RightLabel{$\disr$}
\UnaryInfC{$w R u, u: \forall x(p \lor r(x)) \sar u: p \lor (q \imp \forall x r(x))$}
\RightLabel{$\impr$}
\UnaryInfC{$\sar w: \forall x(p \lor r(x)) \imp  (p \lor (q \imp  \forall x r(x)))$}
\DisplayProof
%}
\end{center}
\end{small}
\caption{An example proof in $\ncd$ for bi-intuitionistic logic with constant domains.}
\label{fig:int-const}
\end{figure*}

To provide intuition, we give an example showing the operation of a reachability rule.

\begin{example}\label{ex:propagation-graph-path} Let $S = \R, \T, \Gamma \sar \Delta$ such that $\R = \lrel{u}{w}, \lrel{w}{v}$, $\T = w : x, u : y, v : z$, $\Gamma = w : \forall x p(x), w : p(f(y)), w : p(z)$, and $\Delta = u : q(x) \exc q(x), v : r(y)$. A representation of $S$ as a polytree is shown below. We explain (in)valid applications of the $\alll$ reachability rule.\\
%$$
%\seq := \lrel{u}{w}, \lrel{w}{v}, w : x, u : y, v : z, w : \forall x p(x), w : p(f(y)), w : p(z) \sar u : q(x) \exc q(x), v : r(y)
%$$
%\begin{center}
%\vspace{-1em}
\resizebox{\columnwidth}{!}{
%\begin{tabular}{c}
\xymatrix{
\overset{u}{\boxed{u : y \sar u : q(x) \exc q(x)}} \ar[r] & \overset{w}{\boxed{w : x, w : \forall x p(x), w : p(f(y)), w : p(z) \sar }}\ar[r] & \overset{v}{\boxed{v : z \sar v : r(y)}}
}
%\end{tabular}
}
%\end{center}

\medskip

\noindent
 The term $f(y)$ is available for $w$ in $\seq$ since $\rable{u}{w}$, namely there is an edge from $u$ to $w$, and $f(y) \in \termset(X_{w})$ since $X_{w} = \{x,y\}$. Therefore, we may (top-down) apply the $\alll$ rule to delete $w : p(f(y))$ and derive the sequent $\seq' = \R, \T, \Gamma' \sar \Delta$ with $\Gamma' = w : \forall x p(x), w : p(z)$.
%$$
%\seq' := \lrel{u}{w}, \lrel{w}{v}, w : x, u : y, v : z, w : \forall x p(x), w : p(z) \sar u : q(x) \exc q(x), v : r(y)
%$$
 By contrast, $w : p(z)$ cannot be deleted via an application of $\alll$ because the term $z$ is \emph{not} available for $w$ in $\seq$ (observe that $w$ is not reachable from $v$) meaning $z \not\in \termset(X_{w})$.
\end{example}

\begin{remark}\label{rmk:constant} We note that for any set $X \subseteq \var$, $\termset(X) \neq \emptyset$ since all constants are contained in $\termset(X)$ by definition. This means that bottom-up applications of $\existsri$ and $\allli$ may instantiate existential and universal formulae with any constant.
\end{remark}

The reachability rules $\ax$, $\impl$ and $\excr$ are important to ensure completeness for both $\ncd$ and $\nid.$ The reachability rules for $\existsr$ and $\alll$ are relevant only for $\nid$ to ensure that the domains in the model do not collapse into a constant domain. We illustrate the importance of these reachability rules with a couple of examples. 

\begin{example}[An Intuitionistic Formula Valid in Constant Domain Models] Consider the intuitionistic formula $\forall x(p \lor r(x)) \imp (p \lor (q \imp \forall xr(x)))$. This formula was adapted from an example in \cite{Lopez81}, which was used to illustrate the difficulty of obtaining a sound and complete sequent system for intuitionistic logic with constant domains. A proof of this formula in $\ncd$ is shown in Figure~\ref{fig:int-const} and crucially relies on reachability rules. In the figure, the relational atoms $\R = w R u, u R v, v R w'$ in the instances of $\ax$ allow us to conclude that $\rable{u}{u}$ and $\rable{u}{w'}$, justifying the left and right instances of $\ax$, respectively.
\end{example}

\begin{example}[Non-Provability of the Quantifier Shift Axiom in the Increasing Domain Setting] Let us consider again the quantifier shift axiom $\forall x(\varphi \lor \psi) \imp (\forall x \varphi \lor \psi)$ and an attempt to construct a proof (bottom-up) of one of its instances in $\nid$.
\begin{small}
\begin{center}
\AxiomC{$w R u, u R v, v : x, u : \forall x(p(x) \lor q) \sar  v : p(x),  u : q$}
\RightLabel{$\allr$}
\UnaryInfC{$w R u,  u : \forall x(p(x) \lor q) \sar  u : \forall x p(x),  u : q$}
\RightLabel{$\disr$}
\UnaryInfC{$w R u,  u : \forall x(p(x)\lor q) \sar  u : \forall x p(x) \lor q$}
\RightLabel{$\impr$}
\UnaryInfC{$\sar w : \forall x(p(x) \lor q) \imp (\forall x p(x) \lor q)$}
\DisplayProof
\end{center}
\end{small}
It is obvious that to finish this proof, we would need to instantiate the $\forall x$ quantifier in the labeled formula $u : \forall x(p(x) \lor q)$ with $x$ by applying the $\alll$ rule. 
However, to do so, we would need to demonstrate that the world $u$ is reachable from $v$ where the domain atom $v : x$ resides. Yet, $u$ is not reachable from $v$, so $x$ is not available at $u$ to be used by $\alll$. 
\end{example}

\subsection{The Domain Shift Rule $\doms$}

Although the reachability rules for the quantifiers prevent the quantifier shift axiom from being proved, it turns out that they are not sufficient to ensure the completeness of $\nid$ with respect to the sequent semantics for the logic $\biqid$. Interestingly, this incompleteness only arises when the exclusion connective is involved---if one considers the intuitionistic fragment of $\nid$, these reachability rules are sufficient to prove completeness (see Lemma~\ref{lm:separation} in Section~\ref{sec:int-subsystems}). To see this incompleteness issue, consider the formula in Example~\ref{ex:exists-exc}, which is semantically valid, and the following attempt at a (bottom-up) construction of a proof: 
\begin{small}
\begin{center}
\AxiomC{$w R u, u R v, u' R v, u: x, u' : p(x) \sar u': \exists y p(y), v: \bot$}
\RightLabel{$\excl$}
\UnaryInfC{$w R u, u R v, u: x, v : p(x) \exc \exists y p(y)  \sar v: \bot$}
\RightLabel{$\impr$}
\UnaryInfC{$w R u, u: x \sar u: (p(x) \exc \exists y p(y)) \imp \bot$}
\RightLabel{$\allr$}
\UnaryInfC{$\sar w: \forall x( (p(x) \exc \exists y p(y)) \imp \bot)$}
\DisplayProof
\end{center}
\end{small}
We have so far applied only invertible rules, %\footnote{We remark that invertibility is formally defined in \sect~\ref{sec:properties}.} 
so the original sequent is provable \iffi the top sequent in the above derivation also is. %It is obvious that 
To proceed with the proof construction, one needs to instantiate the existential quantifier $\exists y$ with $x$. However, the only domain atom containing $x$ is located at the world $u$, which is not available to $u'$ where the existential formula is located. 

It is not so obvious how the reachability rules for quantifiers could be amended to allow this example to be proved. Looking at the above derivation, it might be tempting to augment the calculus with a rule that allows a backward reachability condition for domain atoms, e.g., making $u:x$ available to $u'$ for when $\rable{u'}{u}$ under certain admissibility conditions, but this could easily lead to a collapse of the domains if one is not careful. Instead, our approach here is motivated by the semantic clause for predicates: when $p(x)$ holds in a world, its interpretation requires that $x$ is also defined in that world. Proof theoretically, we could think of this as postulating an axiom such as $\forall x (p(x) \imp E(x))$ where $E(x)$ is an existence predicate (which, as we recall, was behind the semantics of the domain atoms). Translated into our calculus, this gives us the $\doms$ rule as shown in Figure~\ref{fig:nested-calculi}. Using the $\doms$ rule, the above derivation can now be completed to a proof: 
\begin{small}
\begin{center}
\AxiomC{}
\RightLabel{$\ax$}
\UnaryInfC{$\R, u: x, u' : x, u' : p(x) \sar u': p(x), u': \exists y p(y), v: \bot$}
\RightLabel{$\existsr$}
\UnaryInfC{$\R, u: x, u' : x, u' : p(x) \sar u': \exists y p(y), v: \bot$}
\RightLabel{$\doms$}
\UnaryInfC{$w R u, u R v, u' R v, u: x, u' : p(x) \sar u': \exists y p(y), v: \bot$}
\DisplayProof
\end{center}
\end{small}
%It is important to 
Note that the $\doms$ rule can be applied only to atomic predicates, but not arbitrary formulae. This rules out unsound instances, e.g., allowing the domain atom $u:x$ to be introduced when $u: p(x) \imp \bot$, which is clearly semantically not valid. It may be possible to relax the restriction to atomic predicates by imposing some positivity conditions on the occurrences of $x$, but we did not find this necessary---neither for completeness, nor for cut-elimination. 

\begin{remark}
\label{rem:logic-without-ds}
The $\doms$ rule can be removed without affecting the cut-elimination result for $\nid$. This raises the possibility of defining a first-order bi-intuitionistic logic strictly weaker than $\biqid$. It is unclear what the semantics for such a logic would look like. 
\end{remark}

% \begin{example}[The need for the $\doms$ rule in $\nid$]
% \label{exa:ds-rule}
% The formula:
% $\forall x( (p(x) \exc \exists y. p(y)) \imp \bot)$
% can be proved in $\nid$ only when the $ds$ rule is used. An example of such a proof is as follows, where $\R = \{w R u, u R v, u' R v\}.$ 
% \begin{small}
% \begin{center}
% \AxiomC{}
% \RightLabel{$\ax$}
% \UnaryInfC{$\R, u: x, u' : x, u' : p(x) \sar u': p(x), u': \exists y. p(y), v: \bot$}
% \RightLabel{$\existsr$}
% \UnaryInfC{$\R, u: x, u' : x, u' : p(x) \sar u': \exists y. p(y), v: \bot$}
% \RightLabel{$ds$}
% \UnaryInfC{$w R u, u R v, u' R v, u: x, u' : p(x) \sar u': \exists y. p(y), v: \bot$}
% \RightLabel{$\excl$}
% \UnaryInfC{$w R u, u R v, u: x, v : p(x) \exc \exists y. p(y)  \sar v: \bot$}
% \RightLabel{$\impr$}
% \UnaryInfC{$w R u, u: x \sar u: (p(x) \exc \exists y. p(y)) \imp \bot$}
% \RightLabel{$\allr$}
% \UnaryInfC{$\sar w: \forall x( (p(x) \exc \exists y. p(y)) \imp \bot)$}
% \DisplayProof
% \end{center}
% \end{small}
% Note that $u'$ is not reachable from $u$, so without the $ds$ rule, the variable $x$ will not be available at $u'$.
% \end{example}

\subsection{Soundness and Completeness}

\begin{theorem}[Soundness]\label{thm:soundness} Let $\seq$ be a sequent. If $\seq$ is provable in $\nid$ ($\ncd$), then $\seq$ is ($\cd$-)valid.
%If $\seq$ is derivable in $\calc(\modclass)$, then $\seq$ is valid.
\end{theorem}

\begin{proof} By induction on the height of the given proof; %see the online appended version for details~\cite{LyoShiTiu24}.
see \app~\ref{app:soundness-completeness} for details.
%APPREF
\end{proof}

%The completeness of our polytree calculi (stated below) is shown by taking a sequent of the form $w : \vec{x} \sar w : \phi(\vec{x})$ as input and showing that if the sequent is not provable, then the calculus can be used to construct an infinite derivation from which a counter-model of the end sequent can be extracted. As the proof is rather involved, we defer the proof to \app~\ref{app:soundness-completeness}. %As the proof is rather involved, we defer the proof to the online appended version~\cite{LyoShiTiu24}. 

The completeness of our polytree calculi (see \thm~\ref{thm:completeness} below) is shown by taking a sequent of the form $w : \vec{x} \sar w : \phi(\vec{x})$ as input and showing that if the sequent is not provable, then the calculus can be used to construct an infinite derivation from which a counter-model of the end sequent can be extracted. We note that completeness only holds relative to sequents of the form $w : \vec{x} \sar w : \phi(\vec{x})$, which includes a domain atom for each free variable in $\phi(\vec{x})$. This restriction is needed because quantifier rules can only (bottom-up) instantiate quantified formulae with the free variables $\vec{x}$ of $\phi(\vec{x})$ if such free variables occur as domain atoms, and such free variables must be accessible to quantifier rules to properly extract a counter-model of the end sequent (see~\cite{Lyo23arxiv} for a relevant discussion).

%%%APPREF
Below, we outline the cut-free completeness proof for $\nid$ as the proof for $\ncd$ is similar; the complete proof can be found in Appendix~\ref{app:soundness-completeness}. Our proof outline makes use of various new notions, which we now define. A \emph{pseudo-derivation} is defined to be a (potentially infinite) tree whose nodes are sequents and where every parent node corresponds to the conclusion of a rule in $\nid$ with the children nodes corresponding to the premises. We remark that a proof in $\nid$ is a finite pseudo-derivation where all top sequents are instances of $\ax$, $\botl$, or $\topr$. A \emph{branch} $\branch$ is defined to be a maximal path of sequents through a pseudo-derivation, starting from the conclusion. The following lemma is useful in proving completeness.

\begin{lemma}\label{lem:rable-preserved-up} Let $\modclass \in \{\idclass, \cdclass\}$. For each $i \in \{0,1,2\}$, let $\seq_{i} = \R_{i}, \T_{i}, \fsa_{i} \sar \fsb_{i}$ be a sequent.
\begin{enumerate}

\item If $\rable{w}{u}$ holds for the conclusion of a rule $(r)$ in $\calc(\modclass)$, then $\rable{w}{u}$ holds for the premises of $(r)$;

\item If $w : p(\vec{t}) \in \fsa_{0}, \fsb_{0}$ and $\seq_{0}$ is the conclusion of a rule $(r)$ in $\calc(\modclass)$ with $\seq_{1}$ (and $\seq_{2}$) the premise(s) of $(r)$, then $w : p(\vec{t}) \in \fsa_{1}, \fsb_{1}$ (and $w : p(\vec{t}) \in \fsa_{2}, \fsb_{2}$, resp.);

\item If $w : x \in \T_{0}$ and $\seq_{0}$ is the conclusion of a rule $(r)$ in $\calc(\modclass)$ with $\seq_{1}$ (and $\seq_{2}$) the premise(s) of $(r)$, then $w : x\in \T_{1}$ (and $w : x \in \T_{2}$, resp.).

\end{enumerate}
\end{lemma}

% \begin{proof} Each claim can be seen to hold by inspecting the rules of $\calc(\modclass)$.
% \end{proof}

The lemma tells us that propagation paths, the position of atomic formulae, and the position of terms are bottom-up preserved in rule applications.
 
\begin{theorem}[Completeness]\label{thm:completeness} If $w : \vec{x} \sar w : \phi(\vec{x})$ is ($\cd$-)valid, then $w : \vec{x} \sar w : \phi(\vec{x})$ is provable in $\nid$ ($\ncd$).
\end{theorem}

\begin{proof} (Outline.) We assume that $\seq = w : \vec{x} \sar w : \phi(\vec{x})$ is not provable in $\nid$ and show that a model $M$ can be defined which witnesses that $\seq$ is invalid. To prove this, we first define a proof-search procedure $\prove$ that bottom-up applies rules from $\nid$ to $w : \vec{x} \sar w : \phi(\vec{x})$. Second, we show how a model $M$ can be extracted from failed proof-search. We now describe the proof-search procedure $\prove$ and let $\prec$ be a well-founded, strict linear order over the set $\termset$ of terms.\\

\noindent
$\prove$. Let us take $w : \vec{x} \sar w : \phi(\vec{x})$ as input and continue to the next step. We show here some key selected steps; the complete $\prove$ procedure can be found in Appendix~\ref{app:soundness-completeness}. %the online appended version~\cite{LyoShiTiu24}.

$\ax$, $\botl$, and $\topr$. Suppose $\branch_{1}, \ldots, \branch_{n}$ are all branches occurring in the current pseudo-derivation and let $\seq_{1}, \ldots, \seq_{n}$ be the top sequents of each respective branch. For each $1 \leq i \leq n$, we halt the computation of $\prove$ on each branch $\branch_{i}$ where $\seq_{i}$ is of the form $\ax$, $\botl$, or $\topr$. If $\prove$ is halted on each branch $\branch_{i}$, then $\prove$ returns $\success$ because a proof of the input has been constructed. However, if $\prove$ did not halt on each branch $\branch_{i}$ with $1 \leq i \leq n$, then let $\branch_{j_{1}}, \ldots, \branch_{j_{k}}$ be the remaining branches for which $\prove$ did not halt. For each such branch, copy the top sequent above itself, and continue to the next step.

\medskip

$\doms$. Suppose $\branch_{1}, \ldots, \branch_{n}$ are all branches occurring in the current pseudo-derivation and let $\seq_{1}, \ldots, \seq_{n}$ be the top sequents of each respective branch. For each $1 \leq i \leq n$, we consider $\branch_{i}$ and extend the branch with bottom-up applications of $\doms$ rules. Let $\branch_{k+1}$ be the current branch under consideration, and assume that $\branch_{1},\ldots,\branch_{k}$ have already been considered. We assume that the top sequent in $\branch_{k+1}$ is of the form
$$
\seq_{k+1} = \lseq{\mathcal R}{\mathcal T}{\fsa, w : p_{1}(\vec{t}_{1}), \ldots, w_{\ell} : p_{\ell}(\vec{t}_{\ell})}{\fsb}
$$
where all atomic input formulae are displayed in $\seq_{k+1}$ above. We successively consider each atomic input formula and bottom-up apply $\doms$, yielding a branch extending $\branch_{k+1}$ with a top sequent saturated under $\doms$ applications. After these operations have been performed for each branch $\branch_{i}$ with $1 \leq i \leq n$, we continue to the next step.

\medskip

$\existsl$. Suppose $\branch_{1}, \ldots, \branch_{n}$ are all branches occurring in the current pseudo-derivation and let $\seq_{1}, \ldots, \seq_{n}$ be the top sequents of each respective branch. For each $1 \leq i \leq n$, we consider $\branch_{i}$ and extend the branch with bottom-up applications of $\existsl$ rules. Let $\branch_{k+1}$ be the current branch under consideration, and assume that $\branch_{1},\ldots,\branch_{k}$ have already been considered. We assume that the top sequent in $\branch_{k+1}$ is of the form
$$
\seq_{k+1} = \R, \T, \fsa, w_{1} : \exists x_{1} \phi_{1}, \ldots, w_{m} : \exists x_{m} \phi_{m} \vdash \fsb
$$
where all existential input formulae $w_{i} : \exists x_{i} \phi_{i}$ are displayed in $\seq_{k+1}$ above. We consider each formula $w_{i} : \exists x_{i} \phi_{i}$ in turn, and bottom-up apply the $\existsl$ rule. These rule applications extend $\branch_{k+1}$ such that
$$
\R, \T', \fsa, w_{1} : \phi_{1}(y_{1}/x_{1}), \ldots, w_{n} : \phi_{m}(y_{m}/x_{m}) \vdash \fsb
$$
is now the top sequent of the branch with $y_{1},\ldots,y_{m}$ fresh variables and $\T' = \T, w_{1} : y_{1}, \ldots, w_{m} : y_{m}$. After these operations have been performed for each branch $\branch_{i}$ with $1 \leq i \leq n$, we continue to the next step.

\medskip

$\existsri$. Suppose $\branch_{1}, \ldots, \branch_{n}$ are all branches occurring in the current pseudo-derivation and let $\seq_{1}, \ldots, \seq_{n}$ be the top sequents of each respective branch. For each $1 \leq i \leq n$, we consider $\branch_{i}$ and extend the branch with bottom-up applications of $\existsri$ rules. Let $\branch_{k+1}$ be the current branch under consideration, and assume that $\branch_{1},\ldots,\branch_{k}$ have already been considered. We assume that the top sequent in $\branch_{k+1}$ is of the form
$$
\seq_{k+1} = \R, \T, \fsa \vdash w_{1} : \exists x_{1} \phi_{1}, \ldots, w_{m} : \exists x_{m} \phi_{m}, \fsb
$$
 where all existential formulae $w_{i} : \exists x_{i} \phi_{i}$ are displayed in $\seq_{k+1}$ above. We consider each labeled formula $w_{m} : \exists x_{m} \phi_{i}$ in turn, and bottom-up apply the $\existsri$ rule. Let $w_{\ell+1} : \exists x_{\ell+1} \phi_{\ell+1}$ be the current formula under consideration, and assume that $w_{1} : \exists x_{1} \phi_{1}, \ldots, w_{\ell} : \exists x_{\ell} \phi_{\ell}$ have already been considered. Recall that $\prec$ is a well-founded, strict linear order over the set $\termset$ of terms. Choose the $\prec$-minimal term $t \in \termset(X_{w_{\ell+1}})$ that has yet to be picked to instantiate $w_{\ell+1} : \exists x_{\ell+1} \phi_{\ell+1}$ and bottom-up apply the $\existsri$ rule, thus adding $w_{\ell+1} : \phi_{\ell+1}(t/x_{\ell+1} )$. We perform these operations for each branch $\branch_{i}$ with $1 \leq i \leq n$. %After these operations have been performed for each branch $\branch_{i}$ with $1 \leq i \leq n$, we continue to the next step.
 
\medskip

The remaining rules of $\nid$ are processed in a similar fashion. The $\prove$ procedure will saturate open branches of the pseudo-derivation that is under construction by repeatedly (bottom-up) applying rules from $\nid$ in a roundabout fashion. % which will construct an infinite pseudo-derivation of $\seq = w : \vec{x} \sar w : \phi(\vec{x})$.

\medskip

Next, we aim to show that if $\prove$ does not return $\success$, then a model $M$, $M$-interpretation $\iota$, and $M$-assignment $\assign$ can be defined such that $M, \iota, \assign \not\models \seq$. If $\prove$ halts, i.e. $\prove$ returns $\success$, then a proof of $\seq$ may be obtained by `contracting'  all redundant inferences from the `$\ax$, $\botl$, and $\topr$' step of $\prove$. Therefore, in this case, since a proof exists, we have obtained a contradiction to our assumption. As a consequence, we have that $\prove$ does not halt, that is, $\prove$ generates an infinite tree with finite branching. By K\"onig's lemma, an infinite branch must exist in this infinite tree, which we denote by $\branch$. We define a model $M = (W,\leq,U,D,\funinterp,\predinterp)$ by means of this branch as follows: 
Let us define the following sets, all of which are obtained by taking the union of each multiset of relational atoms, domain atoms, antecedent labeled formulae, and consequent labeled formulae (resp.) occurring within a sequent in $\branch$:
$$
\R^{\branch} = \!\!\!\!\!\!\!\! \bigcup_{(\R, \T, \fsa \sar \fsb) \in \branch} \!\!\!\!\!\!\!\! \R 
\qquad 
\T^{\branch} = \!\!\!\!\!\!\!\! \bigcup_{(\R, \T, \fsa \sar \fsb) \in \branch} \!\!\!\!\!\!\!\! \T
\qquad
\fsa^{\branch} = \!\!\!\!\!\!\!\! \bigcup_{(\R, \T, \fsa \sar \fsb) \in \branch} \!\!\!\!\!\!\!\! \fsa
\qquad
\fsb^{\branch} = \!\!\!\!\!\!\!\! \bigcup_{(\R, \T, \fsa \sar \fsb) \in \branch} \!\!\!\!\!\!\!\! \fsb
$$
We now define: (1) $u \in W$ \iffi $u \in \lab(\R^{\branch},\T^{\branch}, \fsa^{\branch}, \fsb^{\branch})$, (2) $\leq \ = \{(u,v) \ | \ uRv \in \R\}^{*}$ where $*$ denotes the reflexive-transitive closure, (3) $t \in U$ \iffi there exists a label $u \in \lab(\R^{\branch},\T^{\branch}, \fsa^{\branch}, \fsb^{\branch})$ such that $t \in \termset(X_{u})$, (4) $t \in D(u)$ \iffi $t \in \termset(X_{u})$, and (5) $(t_{1}, \ldots, t_{n}) \in \predinterp(u,p)$ \iffi $v,u \in \lab(\R^{\branch},\T^{\branch}, \fsa^{\branch}, \fsb^{\branch})$, $v \twoheadrightarrow^{*}_{\R^{\branch}} u$, and $v : p(t_{1}, \ldots, t_{n}) \in \fsa^{\branch}$.

It can be shown that $M$ is indeed a model. % (see Appendix~\ref{app:soundness-completeness}). %(see the online appended version~\cite{LyoShiTiu24}).
Let us define $\assign$ to be the $M$-assignment mapping every variable in $U$ to itself and every variable in $\var \setminus U$ arbitrarily. To finish the proof of completeness, we now argue the following by mutual induction on the complexity of the formula $\psi$: (1) if $u : \psi \in \fsa^{\branch}$, then $M,u,\assign \Vdash \psi$, and (2) if $u : \psi \in \fsb^{\branch}$, then $M,u,\assign \not\Vdash \psi$. 
 Let $\iota$ to be the $M$-interpretation such that $\iota(u) = u$ for $u \in W$ and $\iota(v) \in W$ for $v \not\in W$. By the proof above, $M,\iota,\assign \not\models w : \vec{x} \sar w : \phi(\vec{x})$, showing that if a sequent of the form $w : \vec{x} \sar w : \phi(\vec{x})$ is not provable in $\nid$, then it is invalid, that is, every valid sequent of the form $w : \vec{x} \sar w : \phi(\vec{x})$ is provable in $\nid$.
\end{proof}

\begin{remark} We remark that cut admissibility follows from the soundness of the $\cut$ rule (see \fig~\ref{fig:admiss-rules}) and the completeness theorem above. However, this method of proof has two downsides: first, the restriction in the completeness theorem above implies that cut admissibility only holds for proofs with an end sequent of the form $w : \vec{x} \sar w : \phi(\vec{x})$. Second, this (semantic) method of proof does not define an algorithm showing how instances of $\cut$ can be permuted upward and eliminated from a given proof. In \sect~\ref{sec:properties}, we will prove that cut admissibility holds for all proofs and will provide such an algorithm (see \thm~\ref{thm:cut-elim-int}).
\end{remark}

% =====================================================================================================

\subsection{Intuitionistic Subsystems}
\label{sec:int-subsystems}

We end this section by discussing two subsystems of $\nid$ and $\ncd$ arising from restricting the connectives to the intuitionistic fragment. In the former case, we obtain a proof system for the usual first-order intuitionistic logic (with non-constant domains) $\ilq$, and in the latter, we obtain a proof system for intuitionistic logic with constant domains $\ilqc$.

% \begin{corollary} $\biqid$ is conservative over $\ilq$ and $\biqcd$ is conservative over $\ilqc$.
% \end{corollary}

\begin{corollary}[Conservativity]
\label{cor:conservativity}
Let $\varphi$ be an intuitionistic formula (i.e.~a formula with no occurrences of $\exc$). Then, $\varphi$ is valid in $\ilq$ ($\ilqc$) iff $\sar w:\varphi$ is provable in $\nid$ (respectively, $\ncd$). 
\end{corollary}
%In fact, if we omit the $\excl$ and $\excr$ rules from $\nid$ (or, $\ncd$) we obtain the system $\mathrm{N}_{\mathrm{ND}}$ ($\mathrm{N}_{\mathrm{CD}}$, resp.) first-order intuitionistic logic with non-constant (constant, resp.) domains defined in~\cite{Lyo23}.
%\alw{This is not so obvious. Even if we omit the exclusion operator from a sequent, there could be implicit exclusion operators encoded in the relational atoms. We are in a similar situation as in the display calculus --- the structures of a sequent may hide exclusion operators. We need to sharpen this statement slightly, i.e. defining what is an intuitionistic sequent explicitly. See the next subsection.}

%By omitting the quantifier rules and the domain shift rule from either $\nid$ or $\ncd$, as well as restricting ourselves to the use of propositional formulae, we obtain a nested calculi for propositional bi-intuitionistic logic which is an equivalent variant of the system \textbf{L-LBiI} defined by Pinto and Uustalu~\cite{PinUus18}. Moreover, by removing the exclusion rules from this calculus, we obtain a structural-rule-free nested calculus for propositional intuitionistic logic; cf.~\cite[\sect~5]{Lyo21thesis} and~\cite{Lyo23}.

The proof of Corollary~\ref{cor:conservativity} is straightforward from 
% the definition of the semantics of $\nid$ (resp.~$\ncd$) in 
Definition~\ref{def:sequent-semantics}. 
%However, neither $\nid$ nor $\ncd$ has a direct interpretation in the semantics in Definition~\ref{def:semantic-clauses}, which does not feature domain atoms. 
%(i.e. existence predicates). 
We show here a stronger {\em proof-theoretic} conservativity result: we can in fact extract a purely intuitionistic fragment out of $\nid$, where every sequent in the fragment is interpretable in the semantics without the existence predicate. 
We prove this via syntactic means, by showing how we can translate intuitionistic proofs in $\calc(\idclass)$ to proofs in Gentzen's $\lj$~\cite{Gen35a,Gen35b}. A key idea is to first define a formula interpretation of a polytree sequent, and then show that every inference rule corresponds to a valid implication in $\lj$. We start by defining a notion of intuitionistic (polytree) sequent. 

% The completeness for both systems can be via the \emph{separation property}, which was first discussed in the context of tense logics~\cite{GorPosTiu11}. In short, a nested calculus is in possession of the separation property when omitting inference rules for certain logical connectives yields a sound and complete nested system for the fragment of the logic without such connectives. More specifically, in our context, the nested systems $\nid$ and $\ncd$ exhibit the separation property in the following manner: if we exclude the $\excl$ and $\excr$ rules from $\nid$ or $\ncd$, then the resulting calculus is sound and complete relative to first-order intuitionistic logic $\ilq$ or first-order intuitionistic logic with constant domains $\ilqc$, respectively. 
%Therefore, we have the following as a consequence: 

%the Hilbert system for first-order intuitionistic logic. In the following, we shall assume the Hilbert system for first-order intuitionistic logic $\mathrm{H\!-\!IQC}$ given in \cite{TroDal88}.  

\begin{definition}
\label{def:int-seq}
    A sequent $S = \R, \T, \Gamma \sar \Delta$ is an {\em intuitionistic sequent} iff $\R$ is a tree rooted at node $u$ such that 
    \begin{itemize}
    \item every formula in $S$ is an intuitionistic formula (i.e. it contains no occurrences of $\exc$),   
    \item for every labeled formula $w : \phi$ in $S$ and variable $x \in VT(\phi)$, $x$ is available for $w$, and
    \item if $w: x$ and $z: x$ are in $\T$, then $w = z.$ 
    \end{itemize}
\end{definition}

%Notice that in an intuitionistic sequent, for every label $w$, there is a unique path from the root of the sequent; thus each label can be seen as a prefix in Fitting's prefixed tableaux. 

%Restricted to intuitionistic sequents, the proof system $\biqid$, minus the $ds$ rule, can be seen as a syntactic variant of Fitting's first-order prefixed tableaux for intuitionistic logic (with increasing domain). 

By $\icalc(\idclass)$ we denote the restriction to intuitionistic sequents of the proof system $\calc(\idclass)$ without the $\doms$ rule.
%Let $\icalc(\idclass)$ denote the proof system $\calc(\idclass)$ without the $\doms$ rule, and where the sequents are restricted to intuitionistic sequents. 
The next lemma states an important property of $\nid$, called the \emph{separation property}, which was first discussed in the context of tense logics~\cite{GorPosTiu11}.

\begin{lemma}[Separation]
\label{lm:separation}
An intuitionistic sequent $S$ is provable in $\icalc(\idclass)$ \iffi it is provable in $\calc(\idclass).$
\end{lemma}
\begin{proof} (Outline.)
One direction, from $\icalc(\idclass)$ to $\calc(\idclass)$ is trivial. For the other direction, suppose $\prf$ is a proof of $S$ in $\calc(\idclass).$ 
By induction on the structure of $\prf$, it can be shown that there is a proof $\prf'$ in $\calc(\idclass)$ in which every sequent in $\prf'$ is \emph{almost} intuitionistic---it satisfies all the requirements in Definition~\ref{def:int-seq} except possibly the last condition (due to the possible use of the $\doms$ rule). Then, from $\prf'$ we can construct another proof $\prf''$ of $S$ that does not use $\doms$, by showing that one can always permute the rule $\doms$ up until it disappears. Since all the rules of $\calc(\idclass)$, other than $\doms$, preserve the property of being an intuitionistic sequent, it then follows that $\prf''$ is a proof in $\icalc(\idclass).$
\end{proof}

To translate a proof in $\icalc(\idclass)$ to $\lj$, we need to interpret a polytree sequent as a formula. This turns out to be quite difficult, due to the difficulty in interpreting the scopes of domain atoms, when interpreting them as universally quantified variables. Fortunately, in the case of intuitionistic sequents, the scopes of such variables follow a straightforward lexical scoping (i.e.~their scopes are over formulae in the subtrees).
To define the translation, we first relax the requirement on the domain atoms in intuitionistic sequents: a {\em quasi-intuitionistic sequent} is defined as in Definition~\ref{def:int-seq}, except that in the second clause, $x$ is either available for $w$, or it does not occur in $\T$. Obviously an intutionistic sequent is also a quasi-intuitionistic sequent. 
Given a quasi-intuitionistic sequent $S$ and a label $w$, we write $S_w$ to denote the quasi-intuitionistic sub-sequent of $S$ that is rooted in $w$, i.e.~the sequent obtained from $S$ by removing any relational atoms, domain atoms, and labeled formulae that mention a world $v$ not reachable from $w.$ Given a multiset of labeled formulae $\Gamma$, we denote with $\Gamma_u$ the labeled formulae in $\Gamma$ that are labeled with $u$.

\begin{definition}
\label{def:formula-interp}
Let  $S=\R, \T, \Gamma \sar \Delta$ be a quasi-intuitionistic sequent. We define its formula interpretation $F(S)$ recursively on the height of the sequent tree and suppose $S$ is rooted at $u$.
\begin{itemize}
\item If $S$ is a flat sequent, then $F(X) = \forall \vec x (\bigwedge \Gamma \imp \bigvee \Delta)$ where $\vec x$ are all the variables in $\T$;
\item otherwise, if $u$ has $n$ successors $w_1, \ldots, w_n$, then
$$
F(S) = \forall \vec{x} (\bigwedge \Gamma_u \imp (\bigvee \Delta_u \lor F(S_{w_1}) \lor \cdots \lor F(S_{w_n}))).  
$$

\end{itemize}

\end{definition}

The following proof-theoretic conservativity result can then be proved using a standard translation technique for relating nested sequents and traditional Gentzen sequent calculi~\cite{CloustonDGT13}. % (see Appendix~\ref{app:conservativity} for a proof outline). %(see the online appended version for a proof outline~\cite{LyoShiTiu24}.)
\begin{proposition}
\label{lm:hilbert}
Let $S$ be an intuitionistic sequent. $S$ is provable in $\icalc(\idclass)$ \iffi $F(S)$ is provable in $\lj$. 
\end{proposition}

\begin{proof} (Outline.) The proof is tedious, but not difficult and follows a general strategy to translate nested sequent proofs (which, recall, are notational variants of polytree sequent proofs) to traditional sequent proofs (with cuts) from the literature, see e.g., the translation from nested sequent to traditional sequent proofs for full intuitionistic linear logic~\cite{CloustonDGT13}. 
For every inference rule in $\icalc(\idclass)$ of the form:
%\begin{small}
\begin{center}
\AxiomC{$S_1 \quad \cdots \quad S_n$}
    \UnaryInfC{$S$}
\DisplayProof
\end{center}
%\end{small}
we show that the formula $F(S_1) \land \cdots \land F(S_n) \imp F(S)$ is provable in $\lj$. Then, given any proof in $\icalc(\idclass)$, we simulate every inference step with its corresponding implication, followed by a cut. %A detailed proof will be available in a forthcoming extended version of this paper. 
\end{proof}

As far as we know, for intuitionistic logic with constant domains $\ilqc$, there is no formalization in the traditional Gentzen sequent calculus that admits cut-elimination. There is, however, a formalization in prefixed tableaux by Fitting~\cite{Fit14}, which happens to be a syntactic variant of the intuitionistic fragment of $\ncd$ (shown in~\cite{Lyo21a}).

%The conservativity of $\biqid$ over first-order intuitionistic logic is then an easy consequence of Lemma~\ref{lm:separation} and Lemma~\ref{lm:hilbert}. 

%Invertiblity and admissibility properties
\section{Cut-Elimination}\label{sec:properties}

\begin{figure*}[t]
%\noindent\hrule

\begin{center}
\begin{tabular}{c c c}
\AxiomC{$\lseq{\mathcal R}{\mathcal T}{\Gamma}{\Delta}$}
\RightLabel{$\wkv$}
\UnaryInfC{$\lseq{\mathcal R}{\mathcal T,\lterm{w}{x}}{\Gamma}{\Delta}$}
\DisplayProof

%&

%\AxiomC{$\lseq{\mathcal R}{\mathcal T,\lterm{w}{x},\lterm{w}{x}}{\Gamma}{\Delta}$}
%\RightLabel{$\ctrv$}
%\UnaryInfC{$\lseq{\mathcal R}{\mathcal T,\lterm{w}{x}}{\Gamma}{\Delta}$}
%\DisplayProof

&

\AxiomC{$\lseq{\mathcal R}{\mathcal T,\lterm{w}{x},\lterm{u}{x}}{\Gamma}{\Delta}$}
\RightLabel{$\idr^{\dag_{1}}$}
\UnaryInfC{$\lseq{\mathcal R}{\mathcal T,\lterm{w}{x}}{\Gamma}{\Delta}$}
\DisplayProof

&

\AxiomC{$\lseq{\mathcal R}{\mathcal T}{\Gamma}{\Delta}$}
\RightLabel{$\iwk$}
\UnaryInfC{$\lseq{\mathcal R}{\mathcal T}{\Gamma,\Sigma}{\Delta,\Pi}$}
\DisplayProof
\end{tabular}
\end{center}

\begin{center}
\begin{tabular}{c c c}
\AxiomC{$\phantom{\fsa}$}
\RightLabel{$\gax^{\dag_{1}}$}
\UnaryInfC{$\lseq{\mathcal R}{\mathcal T}{\Gamma,\lform{w}{\phi}}{\Delta,\lform{u}{\phi}}$}
\DisplayProof

&

\AxiomC{$\lseq{\mathcal R,\lrel{w}{v}}{\mathcal T}{\Gamma}{\Delta}$}
\RightLabel{$\brf^{\dag_{2}}$}
\UnaryInfC{$\lseq{\mathcal R,\lrel{u}{v}}{\mathcal T}{\Gamma}{\Delta}$}
\DisplayProof

&

\AxiomC{$\lseq{\mathcal R,\lrel{v}{u}}{\mathcal T}{\Gamma}{\Delta}$}
\RightLabel{$\brb^{\dag_{3}}$}
\UnaryInfC{$\lseq{\mathcal R,\lrel{v}{w}}{\mathcal T}{\Gamma}{\Delta}$}
\DisplayProof
\end{tabular}
\end{center}

\begin{center}
\begin{tabular}{c c c}
\AxiomC{$\lseq{\mathcal R}{\mathcal T,\lterm{w}{x}}{\Gamma}{\Delta}$}
\RightLabel{$\cdr$}
\UnaryInfC{$\lseq{\mathcal R}{\mathcal T}{\Gamma}{\Delta}$}
\DisplayProof

&

\AxiomC{$\lseq{\mathcal R}{\mathcal T}{\Gamma,\lform{w}{\phi},\lform{w}{\phi}}{\Delta}$}
\RightLabel{$\ctrl$}
\UnaryInfC{$\lseq{\mathcal R}{\mathcal T}{\Gamma,\lform{w}{\phi}}{\Delta}$}
\DisplayProof

&

\AxiomC{$\lseq{\mathcal R}{\mathcal T}{\Gamma}{\Delta,\lform{w}{\phi},\lform{w}{\phi}}$}
\RightLabel{$\ctrr$}
\UnaryInfC{$\lseq{\mathcal R}{\mathcal T}{\Gamma}{\Delta,\lform{w}{\phi}}$}
\DisplayProof
\end{tabular}
\end{center}

\begin{center}
\begin{tabular}{c c}
\AxiomC{$\lseq{\mathcal R,\lrel{w}{u}}{\mathcal T}{\Gamma}{\Delta}$}
\RightLabel{$\mrg$}
\UnaryInfC{$\lseq{\mathcal R(w/u)}{\mathcal T(w/u)}{\Gamma(w/u)}{\Delta(w/u)}$}
\DisplayProof

&

\AxiomC{$\lseq{\mathcal R}{\mathcal T}{\Gamma}{\Delta,\lform{w}{\phi}}$}
\AxiomC{$\lseq{\mathcal R}{\mathcal T}{\Gamma,\lform{u}{\phi}}{\Delta}$}
\RightLabel{$\cut^{\dag_{1}}$}
\BinaryInfC{$\lseq{\mathcal R}{\mathcal T}{\Gamma}{\Delta}$}
\DisplayProof
\end{tabular}
\end{center}

\begin{center}
\begin{tabular}{c c c}
\AxiomC{$\lseq{\mathcal R}{\mathcal T}{\Gamma}{\Delta}$}
\RightLabel{$\psub$}
\UnaryInfC{$\lseq{\mathcal R}{\mathcal T(t/x)}{\Gamma(t/x)}{\Delta(t/x)}$}
\DisplayProof

&

\AxiomC{$\lseq{\mathcal R}{\mathcal T}{\Gamma}{\Delta,\lform{w}{\bot}}$}
\RightLabel{$\botr$}
\UnaryInfC{$\lseq{\mathcal R}{\mathcal T}{\Gamma}{\Delta}$}
\DisplayProof

&

\AxiomC{$\lseq{\mathcal R}{\mathcal T}{\Gamma,\lform{w}{\top}}{\Delta}$}
\RightLabel{$\topl$}
\UnaryInfC{$\lseq{\mathcal R}{\mathcal T}{\Gamma}{\Delta}$}
\DisplayProof
\end{tabular}
\end{center}

\begin{center}
\begin{tabular}{c c}
\AxiomC{$\lseq{\mathcal R}{\mathcal T}{\Gamma}{\Delta,\lform{w}{\Pi}}$}
\RightLabel{$\lwr^{\dag_{1}}$}
\UnaryInfC{$\lseq{\mathcal R}{\mathcal T}{\Gamma}{\Delta,\lform{u}{\Pi}}$}
\DisplayProof

&

\AxiomC{$\lseq{\mathcal R}{\mathcal T}{\Gamma,\lform{u}{\Sigma}}{\Delta}$}
\RightLabel{$\lft^{\dag_{1}}$}
\UnaryInfC{$\lseq{\mathcal R}{\mathcal T}{\Gamma,\lform{w}{\Sigma}}{\Delta}$}
\DisplayProof
\end{tabular}
\end{center}

\begin{flushleft}
\textbf{Side Conditions}:
\end{flushleft}
%\begin{minipage}{.25\textwidth}
\begin{description}

\item[$\dag_{1} :=$] $\rable{w}{u}$

\item[$\dag_{2} :=$] $\rable{w}{u}$ and $\notrable{u}{v}$

\item[$\dag_{3} :=$] $\rable{w}{u}$ and $\notrable{w}{v}$

\end{description}
%\end{minipage}

%\hrule
\caption{Admissible rules.}
\label{fig:admiss-rules}
\end{figure*}

In this section, we show that $\nid$ and $\ncd$ satisfy a sizable number of favorable properties culminating in syntactic cut-elimination. We explain here some key steps; %the full details are available in the online appended version~\cite{LyoShiTiu24}.
the full details are available in Appendix~\ref{app:proofs}. 

$\nid$ and $\ncd$ can be seen as first-order extensions of Postniece's deep-nested sequent calculus for bi-intuitionistic logic $\mathrm{DBiInt}$~\cite{Pos09,GorPosTiu08}. 
Cut-elimination for $\mathrm{DBiInt}$~\cite{GorPosTiu08} was proven in two stages. First, cut-elimination was proven for a ``shallow'' version of the nested sequent calculus $\mathrm{LBiInt}$, which can be seen as a variant of a display calculus~\cite{Bel82}. The cut-elimination proof for this shallow calculus follows from Belnap's generic cut-elimination for display calculi~\cite{Bel82}. Second, cut-free proofs in the shallow calculus are shown to be translatable to proofs in the deep-nested calculus. We do not have the corresponding shallow versions of $\nid$ and $\ncd$, so we cannot rely on Belnap's generic cut-elimination. It may be possible to define shallow versions of our calculi, and then follow the same methodology outlined in \cite{GorPosTiu08} to prove cut-elimination, but we find that a direct cut-elimination proof is simpler, e.g., it avoids the need for proving the admissibility of certain structural rules called the display postulates~\cite{Bel82}, which lets one transition from shallow to deep inference systems. 

Since our polytree sequents are a restriction of ordinary labeled sequents, another possible approach to cut-elimination is to apply the methodology for labeled sequent calculi~\cite{Neg05}. A main issue in adapting this methodology is ensuring that the proof transformations needed in cut-elimination preserve the polytree structure of sequents. A key proof transformation in a typical cut-elimination proof for labeled calculi is label substitution: given a proof $\pi_1$ and labels $u$ and $w$, one constructs another proof $\pi_2$ by replacing $u$ with $w$ everywhere in $\pi_1$ and adjusts the inference rules accordingly. This is typically needed in the reduction of a cut where the last rules in both branches of the cut apply to the cut formula, and where one of the rules introduces (reading the rule bottom up) a new label and a new relational atom (e.g., $\impr$). Such a substitution operation may not preserve polytree structures. Another notable difference between our calculi and traditional labeled calculi is the absence of structural rules manipulating relational atoms. These differences mean that cut-elimination techniques for labeled sequent calculi cannot be immediately applied in our setting. 

Instead, our cut-elimination proof builds on an approach by Pinto and Uustalu \cite{PinUus10,PinUus18}, which deals with a polytree sequent calculus for \emph{propositional} bi-intuitionistic logic. We thus provide a series of proof transformations, culminating in the elimination of cuts, which shares similarities with their work in the propositional case and expands in the first-order direction. These transformations are captured in proofs of the admissibility of rules shown in Figure~\ref{fig:admiss-rules}. We illustrate some key transformations and why they are needed, through an example of a cut where $\impl$ and $\impr$ are applied to the cut formula. The formal details are available in the proof of Theorem~\ref{thm:cut-elim-int}. 

Suppose we have the instance of cut shown below left, where $\prf_{1}$ is shown below right and $\prf_{2}$ is shown below bottom with $\rable{w}{u}.$
% The cut rule (Figure~\ref{fig:admiss-rules}) requires that $\rable{w}{u}.$ This is needed to allow one to permute this cut over rules that rely on reachability assumptions. 
% An obvious case is the following:
\begin{center}
\begin{tabular}{@{\hskip 0cm} c @{\hskip .05cm} c}
\AxiomC{$\pi_1$}
\UnaryInfC{$\R, \T, \Gamma \sar w: \varphi \imp \psi, \Delta$}
\AxiomC{$\pi_2$}
\UnaryInfC{$\R, \T, \Gamma, w: \varphi \imp \psi \sar \Delta$}
\RightLabel{$cut$}
\BinaryInfC{$\R, \T, \Gamma \sar \Delta$}
\DisplayProof

&

\AxiomC{$\pi_1'$}
\UnaryInfC{$\R, \T, w R w', \Gamma, w':\varphi \sar w': \psi, \Delta$}
\RightLabel{{$\impr$}}
\UnaryInfC{$\R, \T, \Gamma \sar w: \varphi \imp \psi, \Delta$}
\DisplayProof
\end{tabular}
\end{center}
\begin{center}
\AxiomC{$\pi_3$}
\UnaryInfC{$\R, \T, w : \varphi \imp \psi, \Gamma \sar \Delta, u : \varphi$}
\AxiomC{$\pi_4$}
\UnaryInfC{$\R, \T, \Gamma, w : \varphi \imp \psi, u : \psi \sar \Delta$}
\RightLabel{$\impl$}
\BinaryInfC{$\R, \T, \Gamma, w : \varphi \imp \psi \sar \Delta$}
\DisplayProof
\end{center}
A typical cut reduction strategy would be to cut $\pi_{1}$ with $\pi_{3}$ and $\pi_{4}$ (both with cut formula $\phi \imp \psi$), producing the proofs $\pi_{5}$ and $\pi_{6}$ of $\R, \Gamma \sar \Delta, u : \phi$ and $\R, \Gamma, u : \psi \sar \Delta$, respectively. Next, one would cut $\pi_{5}$ with $\pi_{1}'$ (with cut formula $\phi$), producing a proof $\pi_{7}$, and then cut $\pi_{7}$ with $\pi_{6}$ (with cut formula $\psi$). There are a couple of issues with this strategy: (1) the cut formulae in the last two instances of cut have mismatched labels, i.e., $w'$ on one side and $u$ on the other ; (2) the label $w'$ and the relational atom $w R w'$ are not present in the conclusions of $\pi_5$ and $\pi_6$, so the contexts of the premises of the cuts do not match. 

To fix these issues, we first need to transform the proof $\pi_1'$ into two proofs $\prf_{5}$ and $\prf_{6}$, shown below left and right, respectively. As shown in the cut-elimination proof, a transformation that we use in this case is one that is represented by the rule $\iwk$. This ensures that the contexts match the contexts of the concluding sequents in $\prf_{3}$ and $\prf_{4}$.
\begin{center}
\begin{tabular}{c c}
\AxiomC{$\pi_1'$}
\UnaryInfC{$\R, w R w', \T, \Gamma, w':\varphi \sar w': \psi, \Delta$}
\RightLabel{{$\impr$}}
\UnaryInfC{$\R, \T, \Gamma \sar w: \varphi \imp \psi, \Delta$}
\RightLabel{$\iwk$}
\UnaryInfC{$\R, \T, \Gamma \sar w: \varphi \imp \psi, u : \phi, \Delta$}
\DisplayProof

&

\AxiomC{$\pi_1'$}
\UnaryInfC{$\R, w R w', \T, \Gamma, w':\varphi \sar w': \psi, \Delta$}
\RightLabel{{$\impr$}}
\UnaryInfC{$\R, \T, \Gamma \sar w: \varphi \imp \psi, \Delta$}
\RightLabel{$\iwk$}
\UnaryInfC{$\R, \T, \Gamma, w:\varphi \imp \psi, u: \psi \sar \Delta$}
\DisplayProof
\end{tabular}
\end{center}
We then cut $\prf_{3}$ and $\prf_{4}$ with $\prf_{5}$ and $\prf_{6}$, respectively, which yields proofs $\prf_{7}$ and $\prf_{8}$ of $\R, \Gamma \sar u : \phi, \Delta$ and $\R, \Gamma, u: \psi \sar \Delta$, respectively. At this stage, we want to cut $\prf_{7}$ with $\prf_{1}'$, and then cut the resulting proof with $\prf_{8}$. However, this cut cannot be performed until the label $w'$ and its associated relational atom are removed from the conclusion of $\prf_{1}'$. Simply substituting $u$ for $w'$ may break the polytree shape of the sequent, e.g., if there is a $v$ such that $w R v$ and $v R u$ are in $\R$, then replacing $u$ for $w'$ in $w R w'$ would break the polytree shape of the sequent. So the relational atoms in the sequent also need to be modified. A transformation that we use in this case is represented by the rule $\brf$, which shifts the relational atom $wRw'$ `forward' from $w$ to $u$ given that $\rable{w}{u}$. We also need another transformation to `merge' the label $u$ with the label $w'$, deleting the relational atom in the process, represented as the $\mrg$ rule.
\begin{center}
\AxiomC{$\pi_7$}
\RightLabel{$\iwk$}
\UnaryInfC{$\R, \T, \Gamma \sar u : \phi, u : \psi, \Delta$}

\AxiomC{$\pi_1'$}
\UnaryInfC{$\R, w R w', \T, \Gamma, w':\varphi \sar w': \psi, \Delta$}
\RightLabel{$\brf$}
\UnaryInfC{$\R, u R w', \T, \Gamma, w':\varphi \sar w': \psi, \Delta$}
\RightLabel{$\mrg$}
\UnaryInfC{$\R, \T, \Gamma, u: \varphi \sar u : \psi, \Delta$}

\RightLabel{$cut$}
\BinaryInfC{$\R, \T, \Gamma \sar u : \psi, \Delta$}
\DisplayProof
\end{center}
%To allow cut to be applied to the above proof and $\pi_8$, we relax the cut rule to allow mismatched labels in the cut formula, as long as $\rable{w}{u}$. This lets us cut $w : \psi$ with $u : \psi$ in the conclusion of $\prf_{8}$ to obtain a proof of $\R, \Gamma \sar \Delta.$ Here we gloss over the termination arguments, but the details are available in the proof of Theorem~\ref{thm:cut-elim-int}.
 If we cut the above proof with $\prf_{8}$, we obtain a proof of $\R, \T, \Gamma \sar \Delta.$ Here we gloss over the termination arguments, but the details are available in the proof of Theorem~\ref{thm:cut-elim-int}.

The above example illustrates one among several proof transformations needed in cut-elimination. These transformations make use of the auxiliary rules in Figure~\ref{fig:admiss-rules}. The bulk of the cut-elimination proof consists of showing these rules (height-preserving/hp-) admissible and the rules of $\nid$ and $\ncd$ height-preserving invertible (i.e., hp-invertible). (NB. We take the notions of (hp-)admissibility and (hp-)invertibility to be defined as usual.) %The bulk of the cut-elimination proof consists in showing that these rules are \emph{(height-preserving) admissible, i.e. (hp-)admissible}, meaning if the premises of the rule have proofs (of heights $h_{1}, \ldots, h_{n}$), then the conclusion of the rule has a proof (of height $h \leq \max\{h_{1}, \ldots, h_{n}\}$). If we let $(r^{-1})$ be the inverse of the rule $(r)$ whose premise is the conclusion of $(r)$ and conclusion is the premises of $(r)$, then we say that $(r)$ is \emph{(hp-)invertible} \iffi $(r^{-1})$ is (hp-)admissible. Whenever a rule is (hp-)admissible in both calculi, we only focus on the $\nid$ case and 
 All (hp-)admissible rules preserve the polytree structure of sequents, and with the exception of $\gax$, $\cut$, and $\cdr$, all rules in \fig~\ref{fig:admiss-rules} are \emph{hp-admissible} in \emph{both} calculi. The $\gax$ and $\cut$ rules are strictly \emph{admissible} in both calculi, while $\cdr$ is hp-admissible in only $\calc(\cdclass)$ as the availability conditions are not imposed on rules, rendering domain atoms unnecessary (see Remark~\ref{rmk:ncd}). We now discuss %provide intuition for %the functionality 
 some of the most interesting rules of Figure~\ref{fig:admiss-rules}.

Let us first explain the rules $\brf$ and $\brb$. For some labels $w$, $u$, and $v$, assume $\rable{w}{u}$ and $\notrable{u}{v}$ for $\mathcal R := \R', \lrel{w}{v}$. Then, we know that (1) $u$ and $v$ are on two different paths passing through $w$ of the polytree generated from $\mathcal R$, and (2) there is no vertex between $v$ and $w$ since otherwise a cycle would be present in $\R$. In this scenario, the rule $\brf$ (for \emph{br}anch \emph{f}orward) allows one to move the polytree `rooted' at $v$ forward by connecting it to $u$ instead of $w$ as shown left in the below figure. The rule $\brb$ has a similar functionality; for some labels $w$, $u$, and $v$, assume $\rable{w}{u}$ and $\notrable{w}{v}$ for $\mathcal R := \R', \lrel{v}{u}$. Then, the rule $\brb$ (for \emph{br}anch \emph{b}ackward) lets one move the polytree `rooted' at $v$ backward by connecting it to $w$ instead of $u$ as shown right in the below figure.

\begin{figure}[h]
    \centering

\begin{minipage}{0.45\textwidth}
\begin{tikzpicture}
\draw (0,0) node[point] (w0) [label=below:{$w$}]{};
\draw (1.5,0) node[point](w1)[label=below:{$u$}]{};
\draw (0.5,0.4) node[point](w2)[label=left:{$v$}]{};
%\draw (2.25,0) node {$\Rightarrow$};

\path[->] (w0) edge[dashed] (w1);
\path[->] (w0) edge (w2);
\path[-] (w2) edge[dashed] (1.1,0.3);
\path[-] (w2) edge[dashed] (1,0.8);
\path[-] (1.1,0.3) edge[dashed] (1,0.8);

\draw (2.25,0.2125) node {$\Rightarrow$};

\draw (3,0) node[point] (w0) [label=below:{$w$}]{};
\draw (4.5,0) node[point](w1)[label=below:{$u$}]{};
\draw (5,0.4) node[point](w2)[label=left:{$v$}]{};

\path[->] (w0) edge[dashed] (w1);
\path[->] (w1) edge (w2);
\path[-] (w2) edge[dashed] (5.6,0.3);
\path[-] (w2) edge[dashed] (5.5,0.8);
\path[-] (5.6,0.3) edge[dashed] (5.5,0.8);
\end{tikzpicture}
\end{minipage}
\begin{minipage}{0.05\textwidth}
\ \\
\end{minipage}
\begin{minipage}{0.45\textwidth}
\begin{tikzpicture}
\draw (0,0) node[point] (w0) [label=above:{$w$}]{};
\draw (1.5,0) node[point](w1)[label=above:{$u$}]{};
\draw (1,-0.4) node[point](w2)[label=below:{$v$}]{};

\path[->] (w0) edge[dashed] (w1);
\path[->] (w2) edge (w1);
\path[-] (w2) edge[dashed] (0.4,-0.3);
\path[-] (w2) edge[dashed] (0.5,-0.8);
\path[-] (0.4,-0.3) edge[dashed] (0.5,-0.8);

\draw (2.25,-0.1875) node {$\Rightarrow$};

\draw (4,0) node[point] (w0) [label=above:{$w$}]{};
\draw (5.5,0) node[point](w1)[label=above:{$u$}]{};
\draw (3.5,-0.4) node[point](w2)[label=below:{$v$}]{};

\path[->] (w0) edge[dashed] (w1);
\path[->] (w2) edge (w0);
\path[-] (w2) edge[dashed] (2.9,-0.3);
\path[-] (w2) edge[dashed] (3,-0.8);
\path[-] (2.9,-0.3) edge[dashed] (3,-0.8);
\end{tikzpicture}
\end{minipage}
    
\caption{The left and right diagrams demonstrate the functionality of $\brf$ and $\brb$, respectively.\label{fig:branching-for-back-rules}}
\end{figure}
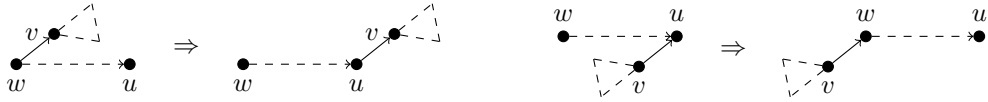

The $\mrg$ rule merges a label and its direct successor and %, i.e., if $\lrel{w}{u}$ or $\lrel{u}{w}$ occurs as a relational atom, then the rule deletes 
%Note that these two rules 
corresponds to the rules $nodemergeD$ and $nodemergeU$ of Pinto and Uustalu~\cite{PinUus10}. The rule $\idr$ reflects the redundancy of a variable labeled by two labels such that one is reachable by the other. Model-theoretically, this redundancy follows from the fact that if $x$ is interpreted at $w$, and $u$ is reachable from $w$ (in a model), then $x$ is interpreted at $u$ as well, showing the domain atom $u : x$ superfluous in the premise. Note that when the labels $w$ and $u$ are identical, then the rule represents a contraction on domain atoms; as $\rable{w}{w}$ always holds, we have that identical domain atoms can always be contracted in sequents. The rules $\lwr$ and $\lft$ allow us to modify the labels of formulae in a sequent by looking at its underlying polytree. More precisely, reading $\lwr$ and $\lft$ bottom-up, if $\rable{w}{u}$ we can both \emph{lower} the label of $\lform{u}{\Pi}$ on the right of the sequent to $w$, and \emph{lift} the label of $\lform{w}{\Sigma}$ on the left of the sequent to $u$.

\begin{lemma}\label{lem:inv}
%Let $\modclass \in \{\idclass, \cdclass\}$. 
%The rules $\conl,\conr,\disl,\disr,\impl,\impr,\excl,\excr,\existsl,\existsr,\alll$ and $\allr$ are hp-invertible in $\calc(\modclass)$.
All non-initial rules in $\calc(\modclass)$ are hp-invertible.
\end{lemma}

Finally, we can prove the admissibility of the $\cut$ rule. As our proof proceeds via local transformations of proofs, the proof is constructive and yields a cut-elimination algorithm.

%\begin{theorem}[Cut-elimination]\label{thm:cut-elim-int}
%The following rule
%\begin{center}
%\AxiomC{$\lseq{\mathcal R}{\mathcal T}{\Gamma}{\Delta,\lform{w}{\phi}}$}
%\AxiomC{$\lseq{\mathcal R}{\mathcal T}{\Gamma,\lform{u}{\phi}}{\Delta}$}
%\RightLabel{$\cut$} %^{\dag}$}
%\BinaryInfC{$\lseq{\mathcal R}{\mathcal T}{\Gamma}{\Delta}$}
%\DisplayProof
%\end{center}
%is admissible in $\nid$, subject to the side condition $\rable{w}{u}$.
%\end{theorem}

\begin{theorem}[Cut-elimination]\label{thm:cut-elim-int}
The $\cut$ rule is admissible in $\calc(\modclass)$. %$\nid$ and $\ncd$.
\end{theorem}

%UPDATED PROOF
%\begin{wrap}
\begin{proof} 
We proceed by a primary induction (PIH) on the complexity of the cut formula, and a secondary induction (SIH) on the sum of the heights of the proofs of the premises of $\cut$. Assume that we have proofs of the following form, with $\rable{w}{u}$.
\begin{center}
\begin{tabular}{c@{\hspace{1cm}} c}
\AxiomC{$\pi_1$}
\RightLabel{$\rone$}
\UnaryInfC{$\lseq{\mathcal R}{\mathcal T}{\fsa}{\fsb,\lform{w}{\phi}}$}
\DisplayProof
&
\AxiomC{$\pi_2$}
\RightLabel{$\rtwo$}
\UnaryInfC{$\lseq{\mathcal R}{\mathcal T}{\fsa,\lform{u}{\phi}}{\fsb}$}
\DisplayProof \\
\end{tabular}
\end{center}
We argue that there is a proof of $\lseq{\mathcal R}{\mathcal T}{\fsa}{\fsb}$ by a case distinction on $\rone$ and $\rtwo$, the last rules applied in the above proofs. We focus on some interesting cases; %the remaining cases can be found in the online appended version~\cite{LyoShiTiu24}.
the remaining cases can be found in Appendix~\ref{app:proofs}.

\medskip

%\begin{wrap}
\noindent $\rone=\ax:$ Then $\lseq{\mathcal R}{\mathcal T}{\fsa}{\fsb,\lform{w}{\phi}}$ is of the form $\lseq{\mathcal R}{\mathcal T}{\fsa_0,\lform{v_0}{p(\vec{t})}}{\fsb_0,\lform{v_1}{p(\vec{t})}}$ where $\rable{v_0}{v_1}$. 
If $\lform{v_1}{p(\vec{t})}$ is $\lform{w}{\phi}$, then we have that $\lseq{\mathcal R}{\mathcal T}{\fsa,\lform{u}{\phi}}{\fsb}$ is of the form $\lseq{\mathcal R}{\mathcal T}{\fsa_0,\lform{v_0}{p(\vec{t})},\lform{u}{p(\vec{t})}}{\fsb}$ where $\fsa=\fsa_0,\lform{v_0}{p(\vec{t})}$. Given that $\rable{v_0}{v_1}$ and $\rable{v_1}{u}$, we can apply the hp-admissibility of $\lft$ on the latter to obtain a proof of $\lseq{\mathcal R}{\mathcal T}{\fsa_0,\lform{v_0}{p(\vec{t})},\lform{v_0}{p(\vec{t})}}{\fsb}$. Consequently, we obtain a proof of $\lseq{\mathcal R}{\mathcal T}{\fsa_0,\lform{v_0}{p(\vec{t})}}{\fsb}$, i.e.~of $\lseq{\mathcal R}{\mathcal T}{\fsa}{\fsb}$, using the hp-admissibility of $\ctrl$. If $\lform{v_1}{p(\vec{t})}$ is not $\lform{w}{\phi}$, then we have that $\lseq{\mathcal R}{\mathcal T}{\fsa}{\fsb}$ is of the form $\lseq{\mathcal R}{\mathcal T}{\fsa_0,\lform{v_0}{p(\vec{t})}}{\fsb_0,\lform{v_1}{p(\vec{t})}}$ where $\rable{v_0}{v_1}$. The latter is provable by the admissibility of $\gax$.\\
%\end{wrap}

%\noindent \textbf{(II)} $\mathbf{\rone=\botl:}$ Then $\lseq{\mathcal R}{\mathcal T}{\fsa}{\fsb,\lform{w}{\phi}}$ is of the form $\lseq{\mathcal R}{\mathcal T}{\fsa_0,\lform{v}{\bot}}{\fsb,\lform{w}{\phi}}$ where $\fsa=\fsa_0,\lform{v}{\bot}$. Consequently, we know that that $\lseq{\mathcal R}{\mathcal T}{\fsa}{\fsb}$ is of the form $\lseq{\mathcal R}{\mathcal T}{\fsa_0,\lform{v}{\bot}}{\fsb}$. We straightforwardly prove the latter via a single application of $\botl$.

%\noindent \textbf{(III)} $\mathbf{\rone=\topr:}$ Then $\lseq{\mathcal R}{\mathcal T}{\fsa}{\fsb,\lform{w}{\phi}}$ is of the form $\lseq{\mathcal R}{\mathcal T}{\fsa}{\fsb_0,\lform{v}{\top}}$. If $\lform{w}{\phi}$ is $\lform{v}{\top}$, then we have that $\lseq{\mathcal R}{\mathcal T}{\fsa,\lform{u}{\phi}}{\fsb}$ is of the form $\lseq{\mathcal R}{\mathcal T}{\fsa,\lform{u}{\top}}{\fsb}$. Thus, we obtain a proof of $\lseq{\mathcal R}{\mathcal T}{\fsa}{\fsb}$ using Lemma \ref{lem:topl-admiss}. If $\lform{w}{\phi}$ is not $\lform{v}{\top}$, then we have that $\lseq{\mathcal R}{\mathcal T}{\fsa}{\fsb}$ is of the form $\lseq{\mathcal R}{\mathcal T}{\fsa}{\fsb_0,\lform{v}{\top}}$. A single application of $\topr$ gets us to our goal.

\noindent $\rone= \doms:$ Then $\lseq{\mathcal R}{\mathcal T}{\fsa}{\fsb,\lform{w}{\phi}}$ is of the form $\lseq{\mathcal R}{\mathcal T}{\fsa_0, \lform{v}{p(\vec{t})}}{\fsb,\lform{w}{\phi}}$ and we have a proof of $\lseq{\mathcal R}{\mathcal T, \lterm{v}{\VT{\vec{t}}}}{\fsa_0, \lform{v}{p(\vec{t})}}{\fsb,\lform{w}{\phi}}$. 
Consequently, we know that $\lseq{\mathcal R}{\mathcal T}{\fsa}{\fsb}$ is of the form $\lseq{\mathcal R}{\mathcal T}{\fsa_0, \lform{v}{p(\vec{t})}}{\fsb}$. We also have that $\lseq{\mathcal R}{\mathcal T}{\fsa,\lform{u}{\phi}}{\fsb}$ is of the form $\lseq{\mathcal R}{\mathcal T}{\fsa_0, \lform{v}{p(\vec{t})},\lform{u}{\phi}}{\fsb}$.
We can repeatedly apply the hp-admissibility of $\wkv$ on the proof of the latter to obtain a proof of 
$S := \lseq{\mathcal R}{\mathcal T, \lterm{v}{\VT{\vec{t}}}}{\fsa_0, \lform{v}{p(\vec{t})},\lform{u}{\phi}}{\fsb}$. Then, we proceed as follows:
%\begin{small}
\begin{center}
\AxiomC{$\lseq{\mathcal R}{\mathcal T, \lterm{v}{\VT{\vec{t}}}}{\fsa_0, \lform{v}{p(\vec{t})}}{\fsb,\lform{w}{\phi}}$}
\AxiomC{$S$}
\RightLabel{SIH}
\dashedLine
\BinaryInfC{$\lseq{\mathcal R}{\mathcal T, \lterm{v}{\VT{\vec{t}}}}{\fsa_0, \lform{v}{p(\vec{t})}}{\fsb}$}
\RightLabel{$\doms$}
\UnaryInfC{$\lseq{\mathcal R}{\mathcal T}{\fsa_0, \lform{v}{p(\vec{t})}}{\fsb}$}
\DisplayProof
\end{center}
%\end{small}
Note that the instance of SIH is justified because the sum of the heights of the proofs of the premises has decreased.\\

\noindent $\rone=\excl:$ Then $\lseq{\mathcal R}{\mathcal T}{\fsa}{\fsb,\lform{w}{\phi}}$ is of the form $\lseq{\mathcal R}{\mathcal T}{\fsa_0,\lform{v}{\psi\exc\chi}}{\fsb,\lform{w}{\phi}}$ and we a have proof of $\lseq{\mathcal R,\lrel{v_0}{v}}{\mathcal T}{\fsa_0,\lform{v_0}{\psi}}{\fsb,\lform{w}{\phi},\lform{v_0}{\chi}}$. Consequently, we know that $\lseq{\mathcal R}{\mathcal T}{\fsa}{\fsb}$ is of the form $\lseq{\mathcal R}{\mathcal T}{\fsa_0,\lform{v}{\psi\exc\chi}}{\fsb}$. We also have that $\lseq{\mathcal R}{\mathcal T}{\fsa,\lform{u}{\phi}}{\fsb}$ is of the form $\lseq{\mathcal R}{\mathcal T}{\fsa_0,\lform{v}{\psi\exc\chi},\lform{u}{\phi}}{\fsb}$. 
We apply Lemma \ref{lem:inv} %\ref{lem:impr-excl-inv} 
on the proof of the latter sequent to obtain a proof of $\lseq{\mathcal R,\lrel{v_0}{v}}{\mathcal T}{\fsa_0,\lform{v_0}{\psi},\lform{u}{\phi}}{\fsb,\lform{v_0}{\chi}}$, which we call $S$. Thus, we proceed as shown below. Note that the instance of SIH is justified as the sum of the heights of the proofs of the premises is smaller than that of the original cut.
%\begin{small}
\begin{center}
\AxiomC{$\lseq{\mathcal R,\lrel{v_0}{v}}{\mathcal T}{\fsa_0,\lform{v_0}{\psi}}{\fsb,\lform{w}{\phi},\lform{v_0}{\chi}}$}

\AxiomC{$S$}

\RightLabel{SIH}
\dashedLine
\BinaryInfC{$\lseq{\mathcal R,\lrel{v_0}{v}}{\mathcal T}{\fsa_0,\lform{v_0}{\psi}}{\fsb,\lform{v_0}{\chi}}$}
\RightLabel{$\excl$}
\UnaryInfC{$\lseq{\mathcal R}{\mathcal T}{\fsa_0,\lform{v}{\psi\exc\chi}}{\fsb}$}
\DisplayProof
\end{center}
%\end{small}

\begin{wrap}
\noindent \textbf{(XII)} $\mathbf{\rone=\excr:}$ Then $\lseq{\mathcal R}{\mathcal T}{\fsa}{\fsb,\lform{w}{\phi}}$ is of the form $\lseq{\mathcal R}{\mathcal T}{\fsa}{\fsb_0,\lform{v}{\psi\exc\chi}}$ and we have proofs of $\lseq{\mathcal R}{\mathcal T}{\fsa}{\fsb_0,\lform{v}{\psi\exc\chi},\lform{v_0}{\psi}}$ and $\lseq{\mathcal R}{\mathcal T}{\fsa,\lform{v_0}{\chi}}{\fsb_0,\lform{v}{\psi\exc\chi}}$ where $\rable{v_0}{v}$.

If $\lform{w}{\phi}$ is not $\lform{v}{\psi\exc\chi}$, then we have proofs of $\lseq{\mathcal R}{\mathcal T}{\fsa}{\fsb_1,\lform{v}{\psi\exc\chi},\lform{v_0}{\psi},\lform{w}{\phi}}$, which we call $S_0$, and $\lseq{\mathcal R}{\mathcal T}{\fsa,\lform{v_0}{\chi}}{\fsb_1,\lform{v}{\psi\exc\chi},\lform{w}{\phi}}$, which we call $S_1$, and $\lseq{\mathcal R}{\mathcal T}{\fsa,\lform{u}{\phi}}{\fsb}$ is of the form $\lseq{\mathcal R}{\mathcal T}{\fsa,\lform{u}{\phi}}{\fsb_1,\lform{v}{\psi\exc\chi}}$. Then, we proceed as follows where $\pi$ is the first proof displayed.

%\begin{small}
\begin{center}
\AxiomC{$S_0$}
\AxiomC{$\lseq{\mathcal R}{\mathcal T}{\fsa,\lform{u}{\phi}}{\fsb_1,\lform{v}{\psi\exc\chi}}$}
\RightLabel{Lem.\ref{lem:iwk-admiss}}
\dashedLine
\UnaryInfC{$\lseq{\mathcal R}{\mathcal T}{\fsa,\lform{u}{\phi}}{\fsb_1,\lform{v}{\psi\exc\chi},\lform{v_0}{\psi}}$}
\RightLabel{SIH}
\dashedLine
\BinaryInfC{$\lseq{\mathcal R}{\mathcal T}{\fsa}{\fsb_1,\lform{v}{\psi\exc\chi},\lform{v_0}{\psi}}$}
\DisplayProof
\end{center}
\begin{center}
\AxiomC{$\pi$}

\AxiomC{$S_1$}
\AxiomC{$\lseq{\mathcal R}{\mathcal T}{\fsa,\lform{u}{\phi}}{\fsb_1,\lform{v}{\psi\exc\chi}}$}
\RightLabel{Lem.\ref{lem:iwk-admiss}}
\dashedLine
\UnaryInfC{$\lseq{\mathcal R}{\mathcal T}{\fsa,\lform{v_0}{\chi},\lform{u}{\phi}}{\fsb_1,\lform{v}{\psi\exc\chi}}$}
\RightLabel{SIH}
\dashedLine
\BinaryInfC{$\lseq{\mathcal R}{\mathcal T}{\fsa,\lform{v_0}{\chi}}{\fsb_1,\lform{v}{\psi\exc\chi}}$}

\RightLabel{$\excr$}
\BinaryInfC{$\lseq{\mathcal R}{\mathcal T}{\fsa}{\fsb_1,\lform{v}{\psi\exc\chi}}$}
\DisplayProof
\end{center}
%\end{small}

If $\lform{w}{\phi}$ is $\lform{v}{\psi\exc\chi}$, then we have proofs of $\lseq{\mathcal R}{\mathcal T}{\fsa}{\fsb,\lform{v}{\psi\exc\chi},\lform{v_0}{\psi}}$ and $\lseq{\mathcal R}{\mathcal T}{\fsa,\lform{v_0}{\chi}}{\fsb,\lform{v}{\psi\exc\chi}}$, and $\lseq{\mathcal R}{\mathcal T}{\fsa,\lform{u}{\phi}}{\fsb}$ is of the form $\lseq{\mathcal R}{\mathcal T}{\fsa,\lform{u}{\psi\exc\chi}}{\fsb}$. In this case, we need to consider the shape of $\rtwo$. If $\lform{u}{\psi\exc\chi}$ is not principal in $\rtwo$, then we use the proof of $\lseq{\mathcal R}{\mathcal T}{\fsa}{\fsb,\lform{v}{\psi\exc\chi}}$ with SIH to cut $\lform{u}{\psi\exc\chi}$ from the premises of $\rtwo$, and then reapply $\rtwo$ to reach our goal. If $\lform{u}{\psi\exc\chi}$ is principal in $\rtwo$, then the premise of $\rtwo$ is of the shape $\lseq{\mathcal R,\lrel{v_1}{u}}{\mathcal T}{\fsa,\lform{v_1}{\psi}}{\fsb,\lform{v_1}{\chi}}$. Then, we proceed as follows where $\pi_0$ and $\pi_1$ are (in this order) the first proofs given.
%\begin{small}
\begin{center}
\AxiomC{$\lseq{\mathcal R}{\mathcal T}{\fsa}{\fsb,\lform{v}{\psi\exc\chi}}$}
\RightLabel{Lem.\ref{lem:iwk-admiss}}
\dashedLine
\UnaryInfC{$\lseq{\mathcal R}{\mathcal T}{\fsa}{\fsb,\lform{v_0}{\psi},\lform{v}{\psi\exc\chi}}$}
\AxiomC{$\lseq{\mathcal R}{\mathcal T}{\fsa,\lform{u}{\psi\exc\chi}}{\fsb,\lform{v_0}{\psi}}$}
\RightLabel{SIH}
\dashedLine
\BinaryInfC{$\lseq{\mathcal R}{\mathcal T}{\fsa}{\fsb,\lform{v_0}{\psi}}$}
\DisplayProof
\end{center}
%\end{small}
%\begin{small}
\begin{center}
\AxiomC{$\lseq{\mathcal R}{\mathcal T}{\fsa}{\fsb,\lform{v}{\psi\exc\chi}}$}
\RightLabel{Lem.\ref{lem:iwk-admiss}}
\dashedLine
\UnaryInfC{$\lseq{\mathcal R}{\mathcal T}{\fsa,\lform{v_0}{\chi}}{\fsb,\lform{v}{\psi\exc\chi}}$}
\AxiomC{$\lseq{\mathcal R}{\mathcal T}{\fsa,\lform{u}{\psi\exc\chi},\lform{v_0}{\chi}}{\fsb}$}
\RightLabel{SIH}
\dashedLine
\BinaryInfC{$\lseq{\mathcal R}{\mathcal T}{\fsa,\lform{v_0}{\chi}}{\fsb}$}
\RightLabel{Lem.\ref{lem:iwk-admiss}}
\dashedLine
\UnaryInfC{$\lseq{\mathcal R}{\mathcal T}{\fsa,\lform{v_0}{\psi},\lform{v_0}{\chi}}{\fsb}$}
\DisplayProof
\end{center}
%\end{small}

\begin{center}
\AxiomC{$\pi_0$}

\AxiomC{$\lseq{\mathcal R,\lrel{v_1}{u}}{\mathcal T}{\fsa,\lform{v_1}{\psi}}{\fsb,\lform{v_1}{\chi}}$}
\RightLabel{Lem.\ref{lem:brf-brb-admiss}}
\dashedLine
\UnaryInfC{$\lseq{\mathcal R,\lrel{v_1}{v_0}}{\mathcal T}{\fsa,\lform{v_1}{\psi}}{\fsb,\lform{v_1}{\chi}}$}
\RightLabel{Lem.\ref{lem:mrgf-mrgb-admiss}}
\dashedLine
\UnaryInfC{$\lseq{\mathcal R}{\mathcal T}{\fsa,\lform{v_0}{\psi}}{\fsb,\lform{v_0}{\chi}}$}
\AxiomC{$\pi_1$}
\RightLabel{PIH}
\dashedLine
\BinaryInfC{$\lseq{\mathcal R}{\mathcal T}{\fsa,\lform{v_0}{\psi}}{\fsb}$}

\RightLabel{PIH}
\dashedLine
\BinaryInfC{$\lseq{\mathcal R}{\mathcal T}{\fsa}{\fsb}$}
\DisplayProof
\end{center}
\end{wrap}

\noindent $\rone=\existsr:$ Then there are two cases to consider: in the left premise $\lseq{\mathcal R}{\mathcal T}{\fsa}{\fsb,\lform{w}{\phi}}$ of $\cut$, either (1) $\lform{w}{\phi}$ is not the principal formula $\lform{v}{\exists x\psi}$, or (2) $\lform{w}{\phi}$ is the principal formula. %Either, (1) $\lseq{\mathcal R}{\mathcal T}{\fsa}{\fsb,\lform{w}{\phi}}$ is of the form $\lseq{\mathcal R}{\mathcal T}{\fsa}{\fsb_0,\lform{v}{\exists x\psi}}$ and we have a proof of $\lseq{\mathcal R}{\mathcal T}{\fsa}{\fsb_0,\lform{v}{\exists x\psi},\lform{v}{\psi(t/x)}}$ where $t$ is available for $v$. Or, (2) if $\lform{w}{\phi}$ is not $\lform{v}{\exists x\psi}$, then 
 In case (1), we have a proof of $\lseq{\mathcal R}{\mathcal T}{\fsa}{\fsb_1,\lform{v}{\exists x\psi},\lform{v}{\psi(t/x)},\lform{w}{\phi}}$, which we call $S$, and $\lseq{\mathcal R}{\mathcal T}{\fsa,\lform{u}{\phi}}{\fsb}$ is of the form $\lseq{\mathcal R}{\mathcal T}{\fsa,\lform{u}{\phi}}{\fsb_1,\lform{v}{\exists x\psi}}$. We proceed as follows.
%\begin{small}
\begin{center}
\AxiomC{$S$}

\AxiomC{$\lseq{\mathcal R}{\mathcal T}{\fsa,\lform{u}{\phi}}{\fsb_1,\lform{v}{\exists x\psi}}$}
\RightLabel{Lem.%\ref{lem:impl-excr-existsri-alll-inv}
\ref{lem:inv}}
\dashedLine
\UnaryInfC{$\lseq{\mathcal R}{\mathcal T}{\fsa,\lform{u}{\phi}}{\fsb_1,\lform{v}{\exists x\psi},\lform{v}{\psi(t/x)}}$}

\RightLabel{SIH}
\dashedLine
\BinaryInfC{$\lseq{\mathcal R}{\mathcal T}{\fsa}{\fsb_1,\lform{v}{\exists x\psi},\lform{v}{\psi(t/x)}}$}
\RightLabel{$\existsr$}
\UnaryInfC{$\lseq{\mathcal R}{\mathcal T}{\fsa}{\fsb_1,\lform{v}{\exists x\psi}}$}
\DisplayProof
\end{center}
%\end{small}
%If $\lform{w}{\phi}$ is the principal formula $\lform{v}{\exists x\psi}$, then 
In case (2), we have proof a of $\lseq{\mathcal R}{\mathcal T}{\fsa}{\fsb,\lform{v}{\exists x\psi},\lform{v}{\psi(t/x)}}$, and $\lseq{\mathcal R}{\mathcal T}{\fsa,\lform{u}{\phi}}{\fsb}$ is of the form $\lseq{\mathcal R}{\mathcal T}{\fsa,\lform{u}{\exists x\psi}}{\fsb}$. In this case, we need to consider the shape of $\rtwo$. If $\lform{u}{\exists x\psi}$ is not principal in $\rtwo$, then we apply the hp-invertibility of $\rtwo$ (\cref{lem:inv}) to the left premise of $\cut$ and use SIH to cut the result with the premise of $\rtwo$, applying $\rtwo$ afterward to reach our goal. If $\lform{u}{\exists x\psi}$ is principal in $\rtwo$, then the premise of $\rtwo$ is of the shape $\lseq{\mathcal R}{\mathcal T,\lterm{v}{y}}{\fsa,\lform{v}{\psi(y/x)}}{\fsb}$ where $y$ is fresh. Then, we proceed as follows where $\pi$ is the first proof given and $x_0,\dots,x_n$ are all the variables appearing in $t$.
%\begin{small}
\begin{center}
\AxiomC{$\lseq{\mathcal R}{\mathcal T}{\fsa}{\fsb,\lform{v}{\exists x\psi},\lform{v}{\psi(t/x)}}$}
\AxiomC{$\lseq{\mathcal R}{\mathcal T}{\fsa,\lform{u}{\exists x\psi}}{\fsb}$}
\RightLabel{$\iwk$}
%Lem.\ref{lem:iwk-admiss}}
\dashedLine
\UnaryInfC{$\lseq{\mathcal R}{\mathcal T}{\fsa,\lform{u}{\exists x\psi}}{\fsb,\lform{v}{\psi(t/x)}}$}
\RightLabel{SIH}
\dashedLine
\BinaryInfC{$\lseq{\mathcal R}{\mathcal T}{\fsa}{\fsb,\lform{v}{\psi(t/x)}}$}
\DisplayProof
\end{center}
%\end{small}
%\begin{small}
\begin{center}
\AxiomC{$\pi$}

\AxiomC{$\lseq{\mathcal R}{\mathcal T,\lterm{v}{y}}{\fsa,\lform{v}{\psi(y/x)}}{\fsb}$}
\RightLabel{$(t/y)$}
%Lem.\ref{lem:psub-admiss}}
\dashedLine
\UnaryInfC{$\lseq{\mathcal R}{\mathcal T,\lterm{v}{x_0},\dots,\lterm{v}{x_n}}{\fsa,\lform{v}{\psi(t/x)}}{\fsb}$}
\RightLabel{$\idr$}
%Lem.\ref{lem:ndr-admiss}}
\dashedLine
\UnaryInfC{$\lseq{\mathcal R}{\mathcal T}{\fsa,\lform{v}{\psi(t/x)}}{\fsb}$}

\RightLabel{PIH}
\dashedLine
\BinaryInfC{$\lseq{\mathcal R}{\mathcal T}{\fsa}{\fsb}$}
\DisplayProof
\end{center}
%\end{small}
Note that the step involving $\idr$ %Lemma \ref{lem:ndr-admiss}
is justified as $t$ is available for $v$, meaning for each $x_{i} \in \VT{t}$, there exists a domain atom $u_{i} : x_{i}$ such that $\rable{u_{i}}{v}$, showing $\idr$ applicable.
\end{proof}

\section{Concluding Remarks} %Conclusion and Future Work}
\label{sec:conclusion}

Our analysis indicates that there may be two interesting and possibly distinct first-order extensions of bi-intuitionistic logic that may be worth exploring. The first is to consider a logic with decreasing domains, i.e., if $w \leq u$ then $D(u) \subseteq D(w)$ in the Kripke model. Semantically, this logic is easy to define, but its proof theory is not at all obvious. We are looking into the possibility of formalizing a notion of ``non-existence predicate,'' that is dual to the existence predicate, suggested by Restall~\cite{Res05}. This non-existence predicate may play a similar (but dual) role to the existence predicate in $\nid.$ The other extension is motivated from the proof theoretical perspective. As mentioned in Remark~\ref{rem:logic-without-ds}, it seems that one can obtain a subsystem of $\nid$ without the domain-shift rule $\doms$ that satisfies cut-elimination. As discussed in Section~\ref{sec:nested-calculi}, the $\doms$ rule is crucial to ensure the completeness of $\biqid$ in the presence of the exclusion operator, and so, a natural question to ask is what the semantics of such a logic would look like.

%%
%% Bibliography
%%

%% Please use bibtex, 

\bibliography{bibfile}

\appendix

%%Appendix A
\section{Soundness and Completeness Proofs}\label{app:soundness-completeness}

\subsection{Soundness}

\begin{customthm}{\ref{thm:soundness}}[Soundness] Let $\seq$ be a sequent. If $\seq$ is provable in $\nid$ ($\ncd$), then $\seq$ is ($\cd$-)valid.
%If $\seq$ is derivable in $\nid$, then $\seq$ is valid.
\end{customthm}

\begin{proof} We argue the claim by induction on the height of the given derivation and consider the $\nid$ case as the $\ncd$ case is similar.

\textit{Base case.} It is straightforward to show that any instance of $\botl$ or $\topr$ is valid; hence, we focus on $\ax$ and show that any instance thereof is valid. Let us consider the following instance of $\ax$, where $\rable{w}{u}$ due to the side condition imposed on $\ax$, that is, there exist $v_{1}, \ldots, v_{n} \in \lab(\R)$ such that $wRv_{1}, \ldots, v_{n}Ru \in \R$.
\begin{center}
\AxiomC{}
\RightLabel{$\ax$}
\UnaryInfC{$\lseq{\mathcal R}{\mathcal T}{\Gamma,\lform{w}{p(\vec{t})}}{\Delta,\lform{u}{p(\vec{t})}}$}
\DisplayProof
\end{center}
Let us suppose $\lseq{\mathcal R}{\mathcal T}{\Gamma,\lform{w}{p(\vec{t})}}{\Delta,\lform{u}{p(\vec{t})}}$ is invalid, i.e., a model $M = (W,\leq,U,D,\funinterp,\predinterp)$, $M$-interpretation $\iota$, and $M$-assignment $\assign$ exist such that the following hold: $\iota(w) \leq \iota(v_{1}), \ldots, \iota(v_{n}) \leq \iota(u)$, $M, \iota(w), \assign \Vdash p(\vec{t})$, and $M, \iota(u), \assign \not\Vdash p(\vec{t})$. By the monotonicity condition $\mcond$ (see \dfn~\ref{def:model}), it must be that $M, \iota(u), \assign \Vdash p(\vec{t})$, giving a contradiction. Thus, every instance of $\ax$ must be valid.

\textit{Inductive step.} We prove the inductive step by contraposition, showing that if the conclusion of the last inference in the given proof is invalid, then at least one premise of the final inference must be invalid. We make a case distinction based on the last rule applied in the given derivation.

$\doms$. Suppose $\lseq{\mathcal R}{\mathcal T}{\fsa, \lform{w}{p(\vec{t})}}{\fsb}$ is invalid with $\vec{t} = t_{1}, \ldots, t_{n}$. Then, there exists a model $M$, $M$-interpretation $\iota$, and $M$-assignment $\assign$ such that $M, \iota(w), \assign \Vdash p(\vec{t})$. Therefore, $(\overline{\assign}(t_{1}),\ldots,\overline{\assign}(t_{n})) \in \predinterp(w,p)$, and since $\predinterp(w,p) \subseteq D(w)^{n}$, we have that $\overline{\assign}(t_{i}) \in D(\iota(w))$ for $1 \leq i \leq n$. By the $\concond$ and $\funcond$ conditions, we know that $M, \iota, \assign \models \lterm{w}{\VT{\vec{t}}}$. Therefore, $\lseq{\mathcal R}{\mathcal T, \lterm{w}{\VT{\vec{t}}}}{\fsa, \lform{w}{p(\vec{t})}}{\fsb}$ is invalid as well.

$\conl$. If we assume that $\lseq{\mathcal R}{\mathcal T}{\Gamma,\lform{w}{\phi\land\psi}}{\Delta}$ is invalid, then there exists a model $M$, $M$-interpretation $\iota$, and $M$-assignment $\assign$ such that $M, \iota(w), \assign \Vdash \phi \land \psi$, implying that $M, \iota(w), \assign \Vdash \phi$ and $M, \iota(w), \assign \Vdash \psi$, showing that the premise $\lseq{\mathcal R}{\mathcal T}{\Gamma,\lform{w}{\phi},\lform{w}{\psi}}{\Delta}$ is invalid as well.

$\conr$. Let us suppose that $\lseq{\mathcal R}{\mathcal T}{\Gamma}{\Delta,\lform{w}{\phi\land\psi}}$ is invalid. Then, there exists an model $M$, $M$-interpretation $\iota$, and $M$-assignment $\assign$ such that $M, \iota(w), \assign \not\Vdash \phi \land \psi$. Hence, either $M, \iota(w), \assign \not\Vdash \phi$ or $M, \iota(w), \assign \not\Vdash \psi$. In the first case, the left premise of $\conr$ is invalid, and in the second case, the right premise of $\conr$ is invalid.

$\disl$. Similar to the $\conr$ case.

$\disr$. Similar to the $\conl$ case.

$\impl$. Assume $\lseq{\mathcal R}{\mathcal T}{\Gamma,\lform{w}{\phi\imp\psi}}{\Delta}$ is invalid and $\rable{w}{u}$, i.e. a sequence $w\sr v_{1}, \ldots, v_{n}\sr u$ of relational atoms exist in $\R$. By our assumption, there exists a model $M$, $M$-interpretation $\iota$, and $M$-assignment such that $\iota(w) \leq \iota(v_{1}), \ldots, \iota(v_{n}) \leq \iota(u)$ and $M, \iota(w), \assign \Vdash \phi \imp \psi$. Because $M, \iota(w), \assign \Vdash \phi \imp \psi$ and $\leq$ is transitive, we know that either $M, \iota(u), \assign \not\Vdash \phi$ or $M,\iota(u),\assign \Vdash \psi$. In the first case, the left premise of $\impl$ is invalid, and in the second case, the right premise of $\impl$ is invalid.

$\impr$. Assume that $\lseq{\mathcal R}{\mathcal T}{\Gamma}{\Delta,\lform{w}{\phi\imp\psi}}$ is invalid. Then, there exists a model $M$, an $M$-interpretation $\iota$, and an $M$-assignment $\assign$ such that $M, \iota(w), \assign \not\Vdash \phi \imp \psi$. Hence, there exists a world $u$ such that $\iota(w) \leq u$, $M, u, \assign \Vdash \phi$, and $M, u, \assign \not\Vdash \psi$. Let $\iota'(v) = \iota(v)$ for all labels $v \neq u$ and $\iota'(u) = u$ otherwise. Then, $M$, $\iota'$, and $\assign$ falsify the premise of $\impr$, showing it invalid.

$\excl$. Similar to the $\impr$ case above. %Suppose that $\seq \{\Gamma, \phi \exc \psi^{\inp}\}_{w}$ is invalid. Then, there exists a model $M$ and interpretation $\iota$ such that $M, \iota(w) \not\Vdash \phi \exc \psi$. Hence, there exists a world $u$ such that $u \leq \iota(w)$, $M, u \Vdash \phi$, and $M, u \not\Vdash \psi$. Let $\iota'(v) = \iota(v)$ for all labels $v \neq u$ and $\iota'(u) = u$ otherwise. Then, $M$ and $\iota'$ falsify the premise of $\excl$, showing it invalid.

$\excr$. Similar to the $\impl$ case above. %Let us assume that $\seq \{\Gamma\}_{w}\{\phi \exc \psi^{\outp}\}_{u}$ is invalid and $\exists \ppath(w,u) \in \prgr{\seq} (\stra_{\ppath}(w,u) \in \langsfour(\fd))$. We let $\ppath(w,u) = w, \fd, v_{1}, \ldots, v_{n}, \fd, u$ as the case where $\ppath(w,u) = \emppath(w,w)$ is shown similarly. By our assumption, there exists a model $M$ and interpretation $\iota$ such that $\iota(w) \leq \iota(v_{1}), \ldots, \iota(v_{n}) \leq \iota(u)$ and $M, \iota(w) \not\Vdash \phi \exc \psi$. Hence, either $M, \iota(u) \not\Vdash \phi$ or $M,\iota(u) \Vdash \psi$. In the first case, the left premise of $\excr$ is invalid, and in the second case, the right premise of $\excr$ is invalid.

$\existsl$. Suppose that $\seq = \lseq{\mathcal R}{\mathcal T}{\Gamma,\lform{w}{\exists x\phi}}{\Delta}$ is invalid. Then, there exists a model $M$, an $M$-interpretation $\iota$, and an $M$-assignment $\assign$ such that $M, \iota(w), \assign \Vdash \exists x \phi$. Therefore, there exists an $\parama \in D(\iota(w))$ such that $M, \iota(w), \assign[\parama/y] \Vdash \phi(y/x)$ with $y$ not occurring in $\seq$. Then, as $y$ is fresh, $M$, $\iota$, and $\assign[\parama/y]$ falsify the premise of $\existsl$, showing it invalid.

$\existsri$. Suppose that $\seq = \lseq{\mathcal R}{\mathcal T}{\Gamma}{\Delta,\lform{w}{\exists x\phi}}$ is invalid. Then, there exists a model $M$, and $M$-interpretation $\iota$, and $M$-assignment such that $M, \iota(w), \assign \not\Vdash \exists x \phi$. By the side condition on $\existsri$, we know that $\avail{t}{X_{w}}{\R,\T}$, meaning there exist labels $u_{1}, \ldots, u_{n} \in \lab(\seq)$ such that $u_{1} : x_{1}, \ldots, u_{n} : x_{n} \in \T$, $\VT{t} = \{x_{1}, \ldots, x_{n}\}$, and $\rable{u_{1}}{w}, \ldots, \rable{u_{n}}{w}$. It follows that $\iota(u_{1}) \leq \iota(w), \ldots, \iota(u_{n}) \leq \iota(w)$ and $\terminterp{\assign}(x_{1}) \in D(\iota(u_{1})), \ldots, \terminterp{\assign}(x_{n}) \in D(\iota(u_{n}))$. By the increasing domain condition $\idcond$, we have that $\terminterp{\assign}(x_{1}) \in D(\iota(w)), \ldots, \terminterp{\assign}(x_{n}) \in D(\iota(w))$. Therefore, by the $\concond$ and $\funcond$ conditions, we know that $\terminterp{\assign}(t) \in D(\iota(w))$, showing that $M, \iota(w), \assign \not\Vdash \phi(t/x)$, and thus, the premise is invalid.

$\allli$. Suppose that $\seq = \lseq{\mathcal R}{\mathcal T}{\fsa,\lform{w}{\forall x\phi}}{\fsb}$ is invalid. Then, there exists a model $M$, and $M$-interpretation $\iota$, and $M$-assignment $\assign$ such that $M, \iota(w), \assign \Vdash \forall x \phi$. By the side condition on $\allli$, we know that $\rable{w}{u}$ and $\avail{t}{X_{w}}{\R,\T}$. By the latter fact, there exist labels $v_{1}, \ldots, v_{n} \in \lab(\seq)$ such that $v_{1} : x_{1}, \ldots, v_{n} : x_{n} \in \T$, $\VT{t} = \{x_{1}, \ldots, x_{n}\}$, and $\rable{v_{1}}{w}, \ldots, \rable{v_{n}}{w}$. It follows that $\iota(v_{1}) \leq \iota(w), \ldots, \iota(v_{n}) \leq \iota(w)$ and $\terminterp{\assign}(x_{1}) \in D(\iota(v_{1})), \ldots, \terminterp{\assign}(x_{n}) \in D(\iota(v_{n}))$. By the increasing domain condition $\idcond$, we have that $\terminterp{\assign}(x_{1}) \in D(\iota(w)), \ldots, \terminterp{\assign}(x_{n}) \in D(\iota(w))$. Therefore, by the $\concond$ and $\funcond$ conditions and our assumption, we know that $\terminterp{\assign}(t) \in D(\iota(w))$, showing that $M, \iota(w), \assign \Vdash \phi(t/x)$. By the fact that $\rable{w}{u}$, we know $\iota(w) \leq \iota(u)$ and $\terminterp{\assign}(t) \in D(\iota(u))$, showing that $M, \iota(u), \assign \Vdash \phi(t/x)$ by \prp~\ref{prop:monotonicity}. Thus, the premise is invalid.

$\allr$. Let us assume that $\lseq{\mathcal R}{\mathcal T}{\Gamma}{\Delta,\lform{w}{\forall x\phi}}$ is invalid. Then, there exists a model $M$, an $M$-interpretation $\iota$, and an $M$-assignment $\assign$ such that $M, \iota(w), \assign \not\Vdash \forall x \phi$. Thus, there exists a world $u \in W$ such that $\iota(w) \leq u$, $\parama \in D(u)$, and $M, u, \assign[\parama/y] \not\Vdash \phi(y/x)$. We define $\mint'(v) = \mint(v)$ if $v \neq u$ and $\mint'(u) = u$. Then, $M$, $\iota'$, and $\assign[\parama/y]$ falsify the premise $\lseq{\mathcal R,w\leq u}{\mathcal T,\lterm{u}{y}}{\Gamma}{\Delta,\lform{u}{\phi(y/x)}}$, showing it invalid.
\end{proof}

\subsection{Completeness}

We prove the cut-free completeness of $\nid$ by showing that if a sequent of the form $w : \vec{x} \sar w : \phi(\vec{x})$ is not provable, then a counter-model can be constructed witnessing the invalidity of the sequent. We focus on the proof for $\nid$ as it is more involved than the similar proof for $\ncd$. Our proof makes use of various new notions, which we now define. A \emph{pseudo-derivation} is defined to be a (potentially infinite) tree whose nodes are sequents and where every parent node corresponds to the conclusion of a rule in $\nid$ with the children nodes corresponding to the premises. We remark that a proof in $\nid$ is a finite pseudo-derivation where all sequents at leaves are instances of $\ax$, $\botl$, or $\topr$. A \emph{branch} $\branch$ is defined to be a maximal path of sequents through a pseudo-derivation, starting from the conclusion. The following lemma is useful in our proof of completeness.

%\begin{customlem}{\ref{lem:rable-preserved-up}} 
\begin{customlem}{\ref{lem:rable-preserved-up}} Let $\modclass \in \{\idclass, \cdclass\}$. For each $i \in \{0,1,2\}$, let $\seq_{i} = \R_{i}, \T_{i}, \fsa_{i} \sar \fsb_{i}$ be a sequent.
\begin{enumerate}

\item If $\rable{w}{u}$ holds for the conclusion of a rule $(r)$ in $\calc(\modclass)$, then $\rable{w}{u}$ holds for the premises of $(r)$;

\item If $w : p(\vec{t}) \in \fsa_{0}, \fsb_{0}$ and $\seq_{0}$ is the conclusion of a rule $(r)$ in $\calc(\modclass)$ with $\seq_{1}$ (and $\seq_{2}$) the premise(s) of $(r)$, then $w : p(\vec{t}) \in \fsa_{1}, \fsb_{1}$ (and $w : p(\vec{t}) \in \fsa_{2}, \fsb_{2}$, resp.);

\item If $w : x \in \T_{0}$ and $\seq_{0}$ is the conclusion of a rule $(r)$ in $\calc(\modclass)$ with $\seq_{1}$ (and $\seq_{2}$) the premise(s) of $(r)$, then $w : x\in \T_{1}$ (and $w : x \in \T_{2}$, resp.).

\end{enumerate}
\end{customlem}

\begin{proof} Each claim can be seen to hold by inspecting the rules of $\calc(\modclass)$.
\end{proof}

The lemma tells us that propagation paths, the position of atomic formulae, and the position of terms are bottom-up preserved in rule applications.

\begin{customthm}{\ref{thm:completeness}}  If $w : \vec{x} \sar w : \phi(\vec{x})$ is ($\cd$-)valid, then $w : \vec{x} \sar w : \phi(\vec{x})$ is provable in $\nid$ ($\ncd$).
%If $w : \vec{x} \sar w : \phi(\vec{x})$ is valid, then $w : \vec{x} \sar w : \phi(\vec{x})$ is derivable in $\nid$.
\end{customthm}

\begin{proof} We assume that $\seq = w : \vec{x} \sar w : \phi(\vec{x})$ is not derivable in $\nid$ and show that a model $M$ can be defined which witnesses that $\seq$ is invalid. To prove this, we first define a proof-search procedure $\prove$ that bottom-up applies rules from $\nid$ to $w : \vec{x} \sar w : \phi(\vec{x})$. Second, we show how an $M$ can be extracted from failed proof-search. We now describe the proof-search procedure $\prove$ and let $\prec$ be a well-founded, strict linear order over the set $\termset$ of terms.\\

\noindent
$\prove$. Let us take $w : \vec{x} \sar w : \phi(\vec{x})$ as input and continue to the next step.\\

$\ax$, $\botl$, and $\topr$. Suppose $\branch_{1}, \ldots, \branch_{n}$ are all branches occurring in the current pseudo-derivation and let $\seq_{1}, \ldots, \seq_{n}$ be the top sequents of each respective branch. For each $1 \leq i \leq n$, we halt the computation of $\prove$ on each branch $\branch_{i}$ where $\seq_{i}$ is of the form $\ax$, $\botl$, or $\topr$. If $\prove$ is halted on each branch $\branch_{i}$, then $\prove$ returns $\success$ because a proof of the input has been constructed. However, if $\prove$ did not halt on each branch $\branch_{i}$ with $1 \leq i \leq n$, then let $\branch_{j_{1}}, \ldots, \branch_{j_{k}}$ be the remaining branches for which $\prove$ did not halt. For each such branch, copy the top sequent above itself, and continue to the next step.

\medskip

$\doms$. Suppose $\branch_{1}, \ldots, \branch_{n}$ are all branches occurring in the current pseudo-derivation and let $\seq_{1}, \ldots, \seq_{n}$ be the top sequents of each respective branch. For each $1 \leq i \leq n$, we consider $\branch_{i}$ and extend the branch with bottom-up applications of $\doms$ rules. Let $\branch_{k+1}$ be the current branch under consideration, and assume that $\branch_{1},\ldots,\branch_{k}$ have already been considered. We assume that the top sequent in $\branch_{k+1}$ is of the form
$$
\seq_{k+1} = \lseq{\mathcal R}{\mathcal T}{\fsa, w_{1} : p_{1}(\vec{t}_{1}), \ldots, w_{\ell} : p_{\ell}(\vec{t}_{\ell})}{\fsb}
$$
where all atomic input formulae are displayed in $\seq_{k+1}$ above. We successively consider each atomic input formula and bottom-up apply $\doms$, yielding a branch extending $\branch_{k+1}$ with a top sequent saturated under $\doms$ applications. After these operations have been performed for each branch $\branch_{i}$ with $1 \leq i \leq n$, we continue to the next step.

\medskip

$\disl$. Suppose $\branch_{1}, \ldots, \branch_{n}$ are all branches occurring in the current pseudo-derivation and let $\seq_{1}, \ldots, \seq_{n}$ be the top sequents of each respective branch. For each $1 \leq i \leq n$, we consider $\branch_{i}$ and extend the branch with bottom-up applications of $\disl$ rules. Let $\branch_{k+1}$ be the current branch under consideration, and assume that $\branch_{1},\ldots,\branch_{k}$ have already been considered. We assume that the top sequent in $\branch_{k+1}$ is of the form
$$
\seq_{k+1} = \R, \T, \fsa, w_{1} : \phi_{1} \lor \psi_{1}, \ldots, w_{m} : \phi_{m} \lor \psi_{m} \sar \fsb
$$
 where all disjunctive formulae $w_{i} : \phi_{i} \lor \psi_{i}$ are displayed in $\seq_{k+1}$ above. We consider each labeled formula $w_{i} : \phi_{i} \lor \psi_{i}$ in turn, and bottom-up apply the $\disl$ rule, which gives $2^{m}$ new branches extending $\branch_{k+1}$ such that each branch has a top sequent of the form $\seq_{k+1} = \R, \T, \fsa, w_{1} : \chi_{1} \ldots, w_{m} : \chi_{m} \sar \fsb$ with $\chi_{i} \in \{\phi_{i},\psi_{i}\}$ and $1 \leq i \leq m$. After these operations have been performed for each branch $\branch_{i}$ with $1 \leq i \leq n$, we continue to the next step.
 
\medskip 
 
$\disr$. Suppose $\branch_{1}, \ldots, \branch_{n}$ are all branches occurring in the current pseudo-derivation and let $\seq_{1}, \ldots, \seq_{n}$ be the top sequents of each respective branch. For each $1 \leq i \leq n$, we consider $\branch_{i}$ and extend the branch with bottom-up applications of $\disr$ rules. Let $\branch_{k+1}$ be the current branch under consideration, and assume that $\branch_{1},\ldots,\branch_{k}$ have already been considered. We assume that the top sequent in $\branch_{k+1}$ is of the form
$$
\seq_{k+1} = \R, \T, \fsa \sar w_{1} : \phi_{1} \lor \psi_{1}, \ldots, w_{m} : \phi_{m} \lor \psi_{m}, \fsb
$$
 where all disjunctive formulae $w_{i} : \phi_{i} \lor \psi_{i}$ are displayed in $\seq_{k+1}$ above. We consider each labeled formula $w_{i} : \phi_{i} \lor \psi_{i}$ in turn, and bottom-up apply the $\disr$ rule. These $\disr$ rule applications extend $\branch_{k+1}$ such that $\R, \T, \fsa \sar w_{1} : \phi_{1}, w_{1} :  \psi_{1}, \ldots, w_{m} : \phi_{m}, w_{m} : \psi_{m}, \fsb$ is now the top sequent of the branch. After these operations have been performed for each branch $\branch_{i}$ with $1 \leq i \leq n$, we continue to the next step.
 
\medskip

$\conl$. Similar to the $\disr$ case above.

\medskip

$\conr$. Similar to the $\disl$ case above.

\medskip

$\impl$. Suppose $\branch_{1}, \ldots, \branch_{n}$ are all branches occurring in the current pseudo-derivation and let $\seq_{1}, \ldots, \seq_{n}$ be the top sequents of each respective branch. For each $1 \leq i \leq n$, we consider $\branch_{i}$ and extend the branch with bottom-up applications of $\impl$ rules. Let $\branch_{k+1}$ be the current branch under consideration, and assume that $\branch_{1},\ldots,\branch_{k}$ have already been considered. We assume that the top sequent in $\branch_{k+1}$ is of the form
$$
\seq_{k+1} = \R, \T, \fsa, w_{1} : \phi_{1} \imp \psi_{1}, \ldots, w_{m} : \phi_{m} \imp \psi_{m} \sar \fsb
$$
 where all implicational formulae $w_{i} : \phi_{i} \imp \psi_{i}$ are displayed in $\seq_{k+1}$ above. We consider each formula $w_{i} : \phi_{i} \imp \psi_{i}$ in turn, and bottom-up apply the $\impl$ rule. Let $w_{\ell + 1} : \phi_{\ell+1} \imp \psi_{\ell+1}$ be the current formula under consideration, and assume that $w_{1} : \phi_{1} \imp \psi_{1}, \ldots, w_{\ell} : \phi_{\ell} \imp \psi_{\ell}$ have already been considered. For every label $u$ such that $\rable{w_{\ell + 1}}{u}$, bottom-up apply the $\impl$ rule. After these operations have been performed for each branch $\branch_{i}$ with $1 \leq i \leq n$, we continue to the next step.

\medskip

$\impr$. Suppose $\branch_{1}, \ldots, \branch_{n}$ are all branches occurring in the current pseudo-derivation and let $\seq_{1}, \ldots, \seq_{n}$ be the top sequents of each respective branch. For each $1 \leq i \leq n$, we consider $\branch_{i}$ and extend the branch with bottom-up applications of $\impr$ rules. Let $\branch_{k+1}$ be the current branch under consideration, and assume that $\branch_{1},\ldots,\branch_{k}$ have already been considered. We assume that the top sequent in $\branch_{k+1}$ is of the form
$$
\seq_{k+1} = \R, \T, \fsa \sar w_{1} : \phi_{1} \imp \psi_{1}, \ldots, w_{m} : \phi_{m} \imp \psi_{m}, \fsb
$$
where all implicational formulae $w_{i} : \phi_{i} \imp \psi_{i}$ are displayed in $\seq_{k+1}$ above. We consider each formula $w_{i} : \phi_{i} \imp \psi_{i}$ in turn, and bottom-up apply the $\impr$ rule. These $\impr$ rule applications extend $\branch_{k+1}$ such that
$$
\R, w_{1} \sr u_{1}, \ldots, w_{m} \sr u_{m}, \T, \fsa, u_{1} : \phi_{1}, \ldots, u_{m} : \phi_{m} \sar u_{1} : \psi_{1}, \ldots, u_{m} : \psi_{m}, \fsb
$$
is now the top sequent of the branch with $u_{1}, \ldots, u_{m}$ fresh. After these operations have been performed for each branch $\branch_{i}$ with $1 \leq i \leq n$, we continue to the next step.

\medskip

$\excl$. Similar to the $\impr$ case above.

\medskip

$\excr$. Similar to the $\impl$ case above.

\medskip

$\existsl$. Suppose $\branch_{1}, \ldots, \branch_{n}$ are all branches occurring in the current pseudo-derivation and let $\seq_{1}, \ldots, \seq_{n}$ be the top sequents of each respective branch. For each $1 \leq i \leq n$, we consider $\branch_{i}$ and extend the branch with bottom-up applications of $\existsl$ rules. Let $\branch_{k+1}$ be the current branch under consideration, and assume that $\branch_{1},\ldots,\branch_{k}$ have already been considered. We assume that the top sequent in $\branch_{k+1}$ is of the form
$$
\seq_{k+1} = \R, \T, \fsa, w_{1} : \exists x_{1} \phi_{1}, \ldots, w_{m} : \exists x_{m} \phi_{m} \vdash \fsb
$$
where all existential input formulae $w_{i} : \exists x_{i} \phi_{i}$ are displayed in $\seq_{k+1}$ above. We consider each formula $w_{i} : \exists x_{i} \phi_{i}$ in turn, and bottom-up apply the $\existsl$ rule. These $\existsl$ rule applications extend $\branch_{k+1}$ such that
$$
\R, \T, w_{1} : y_{1}, \ldots, w_{m} : y_{m}, \fsa, w_{1} : \phi_{1}(y_{1}/x_{1}), \ldots, w_{n} : \phi_{m}(y_{m}/x_{m}) \vdash \fsb
$$
is now the top sequent of the branch with $y_{1},\ldots,y_{m}$ fresh variables. After these operations have been performed for each branch $\branch_{i}$ with $1 \leq i \leq n$, we continue to the next step.

\medskip

$\existsri$. Suppose $\branch_{1}, \ldots, \branch_{n}$ are all branches occurring in the current pseudo-derivation and let $\seq_{1}, \ldots, \seq_{n}$ be the top sequents of each respective branch. For each $1 \leq i \leq n$, we consider $\branch_{i}$ and extend the branch with bottom-up applications of $\existsri$ rules. Let $\branch_{k+1}$ be the current branch under consideration, and assume that $\branch_{1},\ldots,\branch_{k}$ have already been considered. We assume that the top sequent in $\branch_{k+1}$ is of the form
$$
\seq_{k+1} = \R, \T, \fsa \vdash w_{1} : \exists x_{1} \phi_{1}, \ldots, w_{m} : \exists x_{m} \phi_{m}, \fsb
$$
 where all existential formulae $w_{i} : \exists x_{i} \phi_{i}$ are displayed in $\seq_{k+1}$ above. We consider each labeled formula $w_{m} : \exists x_{m} \phi_{i}$ in turn, and bottom-up apply the $\existsri$ rule. Let $w_{\ell+1} : \exists x_{\ell+1} \phi_{\ell+1}$ be the current formula under consideration, and assume that $w_{1} : \exists x_{1} \phi_{1}, \ldots, w_{\ell} : \exists x_{\ell} \phi_{\ell}$ have already been considered. Recall that $\prec$ is a well-founded, strict linear order over the set $\termset$ of terms. Choose the $\prec$-minimal term $t \in \termset(X_{w_{\ell+1}})$ that has yet to be picked to instantiate $w_{\ell+1} : \exists x_{\ell+1} \phi_{\ell+1}$ and bottom-up apply the $\existsri$ rule, thus adding $w_{\ell+1} : \phi_{\ell+1}(t/x_{\ell+1} )$. After these operations have been performed for each branch $\branch_{i}$ with $1 \leq i \leq n$, we continue to the next step.
 
 \medskip

$\allli$. Similar to the $\existsri$ case above.

\medskip

$\allr$. Similar to the $\impr$ and $\existsl$ cases above. We now loop back to the `$\ax$, $\botl$, and $\topr$' step.

\medskip

\noindent
This concludes the description of $\prove$.\\
 
We now argue that if $\prove$ does not return $\success$, then a model $M$, $M$-interpretation $\iota$, and $M$-assignment $\assign$ can be defined such that $M, \iota, \assign \not\models \seq$. If $\prove$ halts, i.e. $\prove$ returns $\success$, then a proof of $\seq$ may be obtained by `contracting'  all redundant inferences from the `$\ax$, $\botl$, and $\topr$' step of $\prove$, which contradicts our assumption. Therefore, $\prove$ does not halt, that is, $\prove$ generates an infinite tree with finite branching. By K\"onig's lemma, an infinite branch must exist in this infinite tree, which we denote by $\branch$. We define a model $M = (W,\leq,U,D,\funinterp,\predinterp)$ by means of this branch.

First, let us define the following multisets, all of which are obtained by taking the union of each multiset of relational atoms, domain atoms, antecedent labeled formulae, and consequent labeled formulae (resp.) occurring within a sequent in $\branch$:
$$
\R^{\branch} = \!\!\!\!\!\!\!\! \bigcup_{(\R, \T, \fsa \sar \fsb) \in \branch} \!\!\!\!\!\!\!\! \R 
\qquad 
\T^{\branch} = \!\!\!\!\!\!\!\! \bigcup_{(\R, \T, \fsa \sar \fsb) \in \branch} \!\!\!\!\!\!\!\! \T
\qquad
\fsa^{\branch} = \!\!\!\!\!\!\!\! \bigcup_{(\R, \T, \fsa \sar \fsb) \in \branch} \!\!\!\!\!\!\!\! \fsa
\qquad
\fsb^{\branch} = \!\!\!\!\!\!\!\! \bigcup_{(\R, \T, \fsa \sar \fsb) \in \branch} \!\!\!\!\!\!\!\! \fsb
$$

\begin{itemize}

\item $u \in W$ \iffi $u \in \lab(\R^{\branch},\T^{\branch}, \fsa^{\branch}, \fsb^{\branch})$;

\item $\leq \ = \ \{(u,v) \ | \ uRv \in \R\}^{*}$ where $*$ denotes the reflexive-transitive closure;

\item $t \in U$ \iffi there exists a label $u \in \lab(\R^{\branch},\T^{\branch}, \fsa^{\branch}, \fsb^{\branch})$ such that $t \in \termset(X_{u})$;

\item $t \in D(u)$ \iffi $t \in \termset(X_{u})$;

%\item $\funinterp(f)(t_{1}, \ldots, t_{n}) := f(t_{1}, \ldots, t_{n})$;

\item $(t_{1}, \ldots, t_{n}) \in \predinterp(u,p)$ \iffi $v,u \in \lab(\R^{\branch}\!, \T^{\branch}\!, \fsa^{\branch}\!, \fsb^{\branch})$, $v \twoheadrightarrow^{*}_{\R^{\branch}} u$, and $v : p(t_{1}, \ldots, t_{n}) \in \fsa^{\branch}\!.$

\end{itemize}

We now verify that $M$ is indeed a model. Observe that $W \neq \emptyset$ since $w \in W$ by definition and the relation $\leq$ is reflexive and transitive by definition. Furthermore, by definition, $D(u) \subseteq U$ for each $u \in W$ and $U = \bigcup_{u \in W} D(u)$. Also, since our language contains at least one constant symbol $a$ (see \rmk~\ref{rmk:constant}), we know that $a \in \termset(X_{u})$ for each $u \in W$, i.e. for each $u \in W$, $D(u) \neq \emptyset$. Let us now argue that $M$ satisfies the increasing domain condition $\idcond$, and assume $u,v \in W$, $t \in D(u)$, and $u \leq v$. Since $t \in D(u)$, $t \in \termset(X_{u})$, and since $u \leq v$, we know that $u \twoheadrightarrow^{*}_{\R^{\branch}} v$. It follows that $\termset(X_{u}) \subseteq \termset(X_{v})$ by \dfn~\ref{def:available}, showing that $t \in \termset(X_{v})$, and thus, $t \in D(v)$. It is simple to confirm that $\funinterp$ satisfies the $\concond$ and $\funcond$ conditions as $t \in D(u)$ \iffi $t \in \termset(X_{u})$, and $\termset(X_{u})$ contains every constant and is closed under the formation of terms by definition. For each $n$-ary predicate $p$ and world $u \in W$, $\predinterp(p,u) \subseteq D(u)^{n}$. To show this, suppose that $(t_{1}, \ldots, t_{n}) \in \predinterp(p,u)$. Then, there exists a label $v \in \lab(\R^{\branch},\T^{\branch}, \fsa^{\branch}, \fsb^{\branch})$ such that $v \twoheadrightarrow^{*}_{\R^{\branch}} u$, and $v : p(t_{1}, \ldots, t_{n}) = v : p(\vec{t}) \in \fsa^{\branch}$. By the $\doms$ step of $\prove$, we know that $w : \VT{\vec{t}} \in \T^{\branch}$. It follows that $t_{1}, \ldots, t_{n} \in \termset(X_{u})$, implying that $(t_{1}, \ldots, t_{n}) \in D(u)^{n}$. Finally, we argue that $M$ satisfies the monotonicity condition $\mcond$, and therefore, we assume $u,v \in W$, $u \leq v$, and $(t_{1}, \ldots, t_{n}) \in \predinterp(p,u)$. Since $u \leq v$, we know that there exist $w_{1}, \ldots, w_{n} \in \lab(\R^{\branch},\T^{\branch}, \fsa^{\branch}, \fsb^{\branch})$ such that $uRw_{1}, \ldots, w_{n}Rv$, implying that $u \twoheadrightarrow^{*}_{\R^{\branch}} v$. Since $(t_{1}, \ldots, t_{n}) \in \predinterp(p,u)$, there exists a $v'$ such that $v' \twoheadrightarrow^{*}_{\R^{\branch}} u$ and $v' : p(t_{1}, \ldots, t_{n}) \in \fsa^{\branch}$. Hence, $v' \twoheadrightarrow^{*}_{\R^{\branch}} v$ because $v' \twoheadrightarrow^{*}_{\R^{\branch}} u$ and $u \twoheadrightarrow^{*}_{\R^{\branch}} v$, which shows that $(t_{1}, \ldots, t_{n}) \in \predinterp(p,v)$.
 
We have now confirmed that $M$ is indeed a model. Let us define $\assign$ to be the $M$-assignment mapping every variable in $U$ to itself and every variable in $\var \setminus U$ arbitrarily. To finish the proof of completeness, we now argue the following by mutual induction on the complexity of $\psi$: (1) if $u : \psi \in \fsa^{\branch}$, then $M,u,\assign \Vdash \psi$, and (2) if $u : \psi \in \fsb^{\branch}$, then $M,u,\assign \not\Vdash \psi$. We argue the cases where $\psi$ is of the form $p(t_{1}, \ldots, t_{n})$ or $\forall x \chi$, and omit the remaining cases as they are straightforward or similar.
 
\begin{itemize}

\item $u : p(t_{1},\ldots,t_{n}) \in \fsa^{\branch}$. In this case, $(t_{1},\ldots,t_{n}) \in \predinterp(p,u)$ by the definition of $\predinterp$, implying that $M,u,\assign \Vdash p(t_{1},\ldots,t_{n})$.

\item $u : p(t_{1},\ldots,t_{n}) \in \fsb^{\branch}$. Observe that if a label $v \in \lab(\R^{\branch},\T^{\branch}, \fsa^{\branch}, \fsb^{\branch})$ exists such that $v \twoheadrightarrow^{*}_{\R^{\branch}} u$ and $v : p(t_{1},\ldots,t_{n}) \in \fsa^{\branch}$, then due to the `$\ax$, $\botl$, and $\topr$' step of $\prove$, $\branch$ would be finite. However, as this is not the case, it must be that no label $v  \in \lab(\R^{\branch},\T^{\branch}, \fsa^{\branch}, \fsb^{\branch})$ exists such that $v \twoheadrightarrow^{*}_{\R^{\branch}} u$ and $v : p(t_{1},\ldots,t_{n}) \in \fsa^{\branch}$, showing that $(t_{1}, \ldots, t_{n}) \not\in \predinterp(u,p)$, i.e. $M,u,\assign \not\Vdash p(t_{1},\ldots,t_{n})$.

\item $u : \forall x \chi \in \fsa^{\branch}$. Suppose $v \in W$, $t \in D(v)$, and $u \leq v$. By the assumption that $t \in D(v)$, we know that $\avail{t}{X_{u}}{\R,\T}$ holds for some sequent $\seq = \R,\T,\fsa \sar \fsb$ in $\branch$. Let us suppose w.l.o.g. that $\seq$ is the first such sequent in $\branch$ for which this holds. By \lem~\ref{lem:rable-preserved-up}, it follows that $\avail{t}{X_{w}}{\R',\T'}$ holds for every sequent $\seqb = \R', \T', \fsa' \sar \fsb'$ above $\seq$ in $\branch$. By the assumption that $u \leq v$, we know that $u \twoheadrightarrow^{*}_{\R^{\branch}} v$. Hence, at some point the $\allli$ step of $\prove$ will be applicable at or above $\seq$ in $\branch$, meaning $v : \chi(t/x) \in \fsa^{\branch}$. By IH, we have that $M,v,\assign \Vdash \chi(t/x)$, from which it follows that $M,w,\assign \Vdash \forall x \chi$ by our assumptions.

\item $u : \forall x \chi \in \fsb^{\branch}$. Due to the $\allr$ step of $\prove$, a  sequent of the form $\lseq{\mathcal R,\lrel{u}{v}}{\mathcal T,\lterm{v}{y}}{\fsa}{\fsb,\lform{v}{\chi(y/x)}}$ with $v$ and $y$ fresh must occur in $\branch$. By the definition of $\leq$ and $D(v)$, as well as \lem~\ref{lem:rable-preserved-up}, it follows that $u \leq v$ and $y \in D(v)$. By IH and the definition of $\assign$, $M,v,\assign \not\Vdash \chi(y/x)$, and so, $M,u,\assign \not\Vdash \forall x \chi$.

\end{itemize}

 Let $\iota$ to be the $M$-interpretation such that $\iota(u) = u$ for $u \in W$ and $\iota(v) \in W$ for $v \not\in W$. By the proof above, $M,\iota,\assign \not\models w : \vec{x} \sar w : \phi(\vec{x})$, showing that if a  sequent of the form $w : \vec{x} \sar w : \phi(\vec{x})$ is not derivable in $\nid$, then it is invalid, that is, every valid  sequent of the form $w : \vec{x} \sar w : \phi(\vec{x})$ is provable in $\nid$.
\end{proof}

\begin{comment}
\section{Proof-theoretic conservativity}
\label{app:conservativity}
\begin{customprop}
{\ref{lm:hilbert}}
Let $S$ be an intuitionistic sequent. $S$ is provable in $\icalc(\idclass)$ \iffi $F(S)$ is provable in $LJ$. 
\end{customprop}
\begin{proof}
{\em Outline}. The proof is tedious, but not difficult and follows a general strategy to translate nested sequent proofs (which, recall, are notational variants of polytree sequent proofs) to traditional sequent proofs (with cuts) from the literature, see e.g., the translation from nested sequent to traditional sequent proofs for full intuitionistic linear logic~\cite{CloustonDGT13}. 
For every inference rule in $\icalc(\idclass)$ of the form:
%\begin{small}
\begin{center}
\AxiomC{$S_1 \quad \cdots \quad S_n$}
    \UnaryInfC{$S$}
\DisplayProof
\end{center}
%\end{small}
we show that the formula $F(S_1) \land \cdots \land F(S_n) \imp F(S)$ is provable in $LJ.$ Then given any proof in $\icalc(\idclass)$, we simulate every inference step with its corresponding implication, followed by a cut. 
A detailed proof will be available in a forthcoming extended version of this paper. 
\end{proof}
\end{comment}

%%%Appendix B
\section{Admissibility and Invertibility Proofs}\label{app:proofs}

\begin{lemma}\label{lem:gax-admiss}
%Let $\modclass \in \{\idclass, \cdclass\}$. 
 The $\gax$ rule is derivable (and therefore, admissible) in $\calc(\modclass)$.
\end{lemma}

\begin{proof} By induction on the complexity of the principal formula $\phi$.
\end{proof}

%\begin{customlem}{\ref{lem:brf-brb-admiss}}
\begin{lemma}\label{lem:brf-brb-admiss}
%Let $\modclass \in \{\idclass, \cdclass\}$. 
 The rules $\brf$ and $\brb$ are hp-admissible in $\calc(\modclass)$.
\end{lemma}
%\end{customlem}

\begin{proof}
By changing $\lrel{w}{v}$ to $\lrel{u}{v}$, we are effectively moving the polytree rooted at $v$ from $w$ to $u$ as indicated in the diagram below. % we move from the left diagram to the right one below.
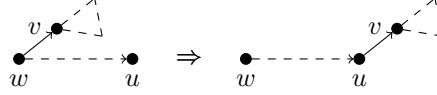
\begin{center}
\begin{tikzpicture}
\draw (0,0) node[point] (w0) [label=below:{$w$}]{};
\draw (1.5,0) node[point](w1)[label=below:{$u$}]{};
\draw (0.5,0.4) node[point](w2)[label=left:{$v$}]{};
\draw (2.25,0) node {$\begin{Huge}\Rightarrow\end{Huge}$};

\path[->] (w0) edge[dashed] (w1);
\path[->] (w0) edge (w2);
\path[-] (w2) edge[dashed] (1.1,0.3);
\path[-] (w2) edge[dashed] (1,0.8);
\path[-] (1.1,0.3) edge[dashed] (1,0.8);

\draw (2.25,0) node {$\begin{Huge}\Rightarrow\end{Huge}$};

\draw (3,0) node[point] (w0) [label=below:{$w$}]{};
\draw (4.5,0) node[point](w1)[label=below:{$u$}]{};
\draw (5,0.4) node[point](w2)[label=left:{$v$}]{};

\path[->] (w0) edge[dashed] (w1);
\path[->] (w1) edge (w2);
\path[-] (w2) edge[dashed] (5.6,0.3);
\path[-] (w2) edge[dashed] (5.5,0.8);
\path[-] (5.6,0.3) edge[dashed] (5.5,0.8);
\end{tikzpicture}
\end{center}
The crucial observation is that given that $\rable{w}{u}$, this operation only extends the reachability relation: we have that for all $w'$ and $w''$, if $\rable{w'}{w''}$ in $\mathcal{R},\lrel{w}{v}$ then $\rable{w'}{w''}$ in $\mathcal{R},\lrel{u}{v}$. Consequently, we can see that $\brf$ only expands reachability, and thus does not violate the reachability or availability conditions of any of the rules. Freshness is also not impacted. So, in all cases we can simply apply the induction hypothesis on the premises, and then the rule.

A similar argument can be provided for $\brb$, where we move from the left diagram to the right one below.
\begin{center}
\begin{tikzpicture}
\draw (0,0) node[point] (w0) [label=above:{$w$}]{};
\draw (1.5,0) node[point](w1)[label=above:{$u$}]{};
\draw (1,-0.4) node[point](w2)[label=below:{$v$}]{};

\path[->] (w0) edge[dashed] (w1);
\path[->] (w2) edge (w1);
\path[-] (w2) edge[dashed] (0.4,-0.3);
\path[-] (w2) edge[dashed] (0.5,-0.8);
\path[-] (0.4,-0.3) edge[dashed] (0.5,-0.8);

\draw (2.25,0) node {$\begin{Huge}\Rightarrow\end{Huge}$};

\draw (4,0) node[point] (w0) [label=above:{$w$}]{};
\draw (5.5,0) node[point](w1)[label=above:{$u$}]{};
\draw (3.5,-0.4) node[point](w2)[label=below:{$v$}]{};

\path[->] (w0) edge[dashed] (w1);
\path[->] (w2) edge (w0);
\path[-] (w2) edge[dashed] (2.9,-0.3);
\path[-] (w2) edge[dashed] (3,-0.8);
\path[-] (2.9,-0.3) edge[dashed] (3,-0.8);
\end{tikzpicture}
\end{center}

%\ian{For all local propositional rules it goes through. 

%For $\ax$ we need to make a case distinction, as we are shifting the branch higher, which may impact reachability. What we have to note is that the moving of a branch upwards only extends the overall reachability relation defined by a set of relational atoms. So we are good. 

%For the non-local propositional rules not creating new labels, we can use the insight above about the extension of the reachability relation.

%For the other non-local propositional rules, then we simply need to apply the IH in the premise (as the fresh label are either moved with the branch, or stay where they are).

%For FO rules the note about reachability condition should work out as well: we only extend reachability and thus availability should be safe too.}
\end{proof}

%\begin{theorem}[Cut-elimination]\label{thm:cut-elim-int}
%The following rule
%\begin{center}
%\AxiomC{$\lseq{\mathcal R}{\mathcal T}{\Gamma}{\Delta,\lform{w}{\phi}}$}
%\AxiomC{$\lseq{\mathcal R}{\mathcal T}{\Gamma,\lform{u}{\phi}}{\Delta}$}
%\RightLabel{$\cut$} %^{\dag}$}
%\BinaryInfC{$\lseq{\mathcal R}{\mathcal T}{\Gamma}{\Delta}$}
%\DisplayProof
%\end{center}
%is admissible in $\nid$, subject to the side condition $\rable{w}{u}$.
%\end{theorem}

% \begin{customlem}{\ref{lem:mrgf-mrgb-admiss}}
% %Let $\modclass \in \{\idclass, \cdclass\}$. 
% The rules $\mrgf$ and $\mrgb$ are hp-admissible in $\calc(\modclass)$.
% \end{customlem}
% \begin{proof}
% Straightforward by induction on derivations. 
% \end{proof}

\begin{lemma}\label{lem:mrgf-mrgb-admiss}
 %Let $\modclass \in \{\idclass, \cdclass\}$. 
The $\mrg$ rule is hp-admissible in $\calc(\modclass)$.
\end{lemma}
\begin{proof}
Straightforward by induction on the height of the given proof. 
\end{proof}

%\begin{customlem}{\ref{lem:ndr-admiss}}
\begin{lemma}\label{lem:ndr-admiss}
%Let $\modclass \in \{\idclass, \cdclass\}$. 
 The $\idr$ rule is hp-admissible in $\calc(\modclass)$.
\end{lemma}
%\end{customlem}

\begin{proof}
%\ian{Note that we can merge this rule with $\ctrv$ if we take $\rable{w}{u}$ as condition for $\idr$.} 
%\tim{Nice observation! I would vote for following your suggestion and remove (cv) along with the associated lemma :)}
%\ian{Sounds good! I will leave the comments here for now as a reminder to perform the change.}
This lemma is straightforwardly proved by induction on the height of proofs. To realize that we can simply apply the induction hypothesis on the premises of the rules, and then reapply the rule, it suffices to note that the deletion of $\lterm{u}{x}$ does not impact reachability, freshness or even availability given that $\rable{w}{u}$. So, any rule is reapplicable once we use the induction hypothesis.
\end{proof}

%\begin{customlem}{\ref{lem:psub-admiss}}
\begin{lemma}\label{lem:psub-admiss}
%Let $\modclass \in \{\idclass, \cdclass\}$. 
 The $\psub$ rule is hp-admissible in $\calc(\modclass)$.
\end{lemma}
%\end{customlem}

\begin{proof}
We reason by induction on the height of proofs, and consider the last rule applied. 

Obviously, the propositional rules are not impacted by the substitution. So, for these rules we simply need to apply the induction hypothesis on the premises of the rule, and then reapply the rule.

The rule $\doms$ is treated straightforwardly: we apply the induction hypothesis on the proof of 
$\lseq{\mathcal R}{\mathcal T, \lterm{w}{\VT{\vec{t'}}}}{\fsa, \lform{w}{p(\vec{t'})}}{\fsb}$ to obtain a proof of $\lseq{\mathcal R}{\mathcal T'}{\fsa(t/x), \lform{w}{p(\vec{t'})(t/x)}}{\fsb(t/x)}$ where $\mathcal T'=(\mathcal T, \lterm{w}{\VT{\vec{t'}}})(t/x)$. Note that $\mathcal T'$ is such that any domain atom $\lterm{u}{x}$ is replaced by $\lterm{u}{x_0},\dots,\lterm{u}{x_m}$ where $x_0,\dots,x_m$ are the variables appearing in $t$. Consequently, all the variables appearing in the terms $\vec{t'}(t/x)$ are labeled by $w$ in $\mathcal T'$. Consequently, we can apply the rule $\doms$ to obtain a proof of our goal, i.e.~$\lseq{\mathcal R}{\mathcal T(t/x)}{\fsa(t/x), w : p(\vec{t'})(t/x)}{\fsb(t/x)}$.

The cases of rules for quantifiers deserve more attention. We show the cases of $\existsl$ and $\allli$, as the other rules can be treated in a similar way.

$\existsl:$ we need to obtain a proof of $\lseq{\mathcal R(t/x)}{\mathcal T(t/x)}{\Gamma(t/x),\lform{w}{(\exists z\phi)(t/x)}}{\Delta(t/x)}$, where we can safely assume that $z$ is different from $x$ and any variable appearing in $t$. First, we apply the induction hypothesis on the proof of $\lseq{\mathcal R}{\mathcal T,\lterm{w}{y}}{\Gamma,\lform{w}{\phi(y/z)}}{\Delta}$, the premise of the rule, to obtain a proof of no greater height of $\lseq{\mathcal R}{\mathcal T(y'/y),\lterm{w}{y(y'/y)}}{\Gamma(y'/y),\lform{w}{\phi(y/z)(y'/y)}}{\Delta(y'/y)}$ with $y'$ fresh and not appearing in $t$ or $x$. Given that $y$ was fresh, we have a proof of no greater height the sequent $\lseq{\mathcal R}{\mathcal T,\lterm{w}{y'}}{\Gamma,\lform{w}{\phi(y'/z)}}{\Delta}$. We can apply the induction hypothesis again here to obtain a proof of no greater height of $\lseq{\mathcal R}{\mathcal T(t/x),\lterm{w}{y'}(t/x)}{\Gamma(t/x),\lform{w}{\phi(y'/z)(t/x)}}{\Delta(t/x)}$. Because of our choice of $y'$, we have that the latter sequent is equal to $\lseq{\mathcal R}{\mathcal T(t/x),\lterm{w}{y'}}{\Gamma(t/x),\lform{w}{\phi(t/x)(y'/z)}}{\Delta(t/x)}$. Thus, we can reapply the rule $\existsl$ on the latter to obtain a proof of $\lseq{\mathcal R(t/x)}{\mathcal T(t/x)}{\Gamma(t/x),\lform{w}{(\exists z\phi)(t/x)}}{\Delta(t/x)}$.

$\allli:$ we need a proof of $\lseq{\mathcal R}{\mathcal T(t/x)}{\Gamma(t/x),\lform{w}{(\forall z\phi)(t/x)}}{\Delta(t/x)}$ of no greater height, where we can safely assume that $z$ is different from $x$ and any variable appearing in $t$. We apply the induction hypothesis on the proof of $\lseq{\mathcal R}{\mathcal T}{\Gamma,\lform{w}{\forall z\phi},\lform{u}{\phi(t'/z)}}{\Delta}$, the premise of the rule, to obtain a proof of no greater height of $\lseq{\mathcal R}{\mathcal T(t/x)}{\Gamma,\lform{w}{(\forall z\phi)(t/x)},\lform{u}{\phi(t'/z)(t/x)}}{\Delta(t/x)}$. Note that the latter is equal to $\lseq{\mathcal R}{\mathcal T(t/x)}{\Gamma,\lform{w}{(\forall z\phi)(t/x)},\lform{u}{\phi(t'(t/x)/z)}}{\Delta(t/x)}$. Clearly, in this sequent we still have that $\rable{w}{u}$. In addition to that, we have that $t'(t/x)$ is available for $u$. We can argue this point by a case distinction on the presence or not of $x$ in $t'$. If $x$ does not appear in $t'$, then $t'(t/x)=t'$, which is available for $u$ as initially given. If $x$ does appear in $t'$, then we have that there must be a $v$ such that $\lterm{v}{x}\in\T$ and $\rable{v}{u}$ else $t'$ would not be available in $u$. Consequently, we have that $\lterm{v}{x}$ is replaced by $\{\lterm{v}{x'}\mid x'\text{ appears in }t\}$ in $\T(t/x)$. This makes $t'(t/x)$ available for $u$. So, as $\rable{w}{u}$ and $t'(t/x)$ is available to $u$ we can apply $\allli$ to obtain a proof of $\lseq{\mathcal R}{\mathcal T(t/x)}{\Gamma(t/x),\lform{w}{(\forall z\phi)(t/x)}}{\Delta(t/x)}$ of no greater height than the proof initially considered.
\end{proof}

%\begin{customlem}{\ref{lem:iwk-admiss}}
\begin{lemma}\label{lem:iwk-admiss}
%Let $\modclass \in \{\idclass, \cdclass\}$. 
 The $\iwk$ rule is hp-admissible in $\calc(\modclass)$.
\end{lemma}
%\end{customlem}

\begin{proof}
With the exception of the $\existsl$ or $\allr$ rules, the $\iwk$ rule permutes above all other rules of $\calc(\modclass)$. In the $\existsl$ and $\allr$ cases, we can apply Lemma \ref{lem:psub-admiss} to the premise of each rule to ensure that the freshness condition holds, then apply $\iwk$, and last, apply the $\existsl$ or $\allr$ rule, respectively. Thus, the $\iwk$ rule can be permuted upward in any $\calc(\modclass)$ proof (in a height-preserving manner) and eliminated at initial sequents, showing the rule hp-admissible.
%can create have to do with the condition of freshness of a variable or a label. To make sure that there is no violation of these conditions, we can use Lemma \ref{lem:psub-admiss} and our view that isomorphic sequents are identical to avoid any overlap with these critical variables or labels (Remark~\ref{rmk:iso-are-equal}). Thus, for all rules we can simply use these lemmas on the premises if needed, apply the induction hypothesis on the potentially modified premises and then reapply the rule.
\end{proof}

%\begin{customlem}{\ref{lem:wkv-admiss}}
\begin{lemma}\label{lem:wkv-admiss}
%Let $\modclass \in \{\idclass, \cdclass\}$.
 The $\wkv$ rule is hp-admissible in $\calc(\modclass)$.
\end{lemma}
%\end{customlem}

\begin{proof}
The addition of a further labeled variable can only affect the freshness condition of an $\existsl$ or $\allr$ rule. In such cases, we apply Lemma \ref{lem:psub-admiss}, then $\wkv$, and last, the $\existsl$ or $\allr$ rule, respectively, to permute $\wkv$ above the $\existsl$ or $\allr$ instance. In all other cases the $\wkv$ rule freely permutes above the rule. Thus, the $\wkv$ rule can be permuted upward in any $\calc(\modclass)$ proof (in a height-preserving manner) and eliminated at initial sequents, showing the rule hp-admissible.
%Thus, we only need to pay attention to the rules involving such conditions and apply Lemma \ref{lem:psub-admiss} before applying the induction hypothesis, if needed, to 
\end{proof}

%\begin{customlem}{\ref{lem:botr-topl-admiss}}
\begin{lemma}\label{lem:botr-topl-admiss}
%Let $\modclass \in \{\idclass, \cdclass\}$. 
 The rules $\botr$ and $\topl$ are hp-admissible in $\calc(\modclass)$.
\end{lemma}
%\end{customlem}

\begin{proof}
We focus on $\botr$ as $\topl$ can be treated dually. By inspecting the rules, one can see that the deletion of $\lform{w}{\bot}$ in the consequent of a sequent impacts the application of no rule. In particular, this deletion does not alter the reachability, availability, or freshness conditions. As a consequence, a straightforward induction on the structure of the proof is sufficient to prove this statement: in the inductive cases, apply the induction hypothesis on the last rule applied, and then the rule.
\end{proof}

%\begin{customlem}{\ref{lem:lwr-admiss}}
\begin{lemma}\label{lem:lwr-admiss}
%Let $\modclass \in \{\idclass, \cdclass\}$. 
 The $\lwr$ rule is hp-admissible in $\calc(\modclass)$.
\end{lemma}
%\end{customlem}

\begin{proof}
We reason by induction on the height of proofs, and consider the last rule applied.

Obviously, the ``local'' propositional rules (which do not involve a condition on reachability) and $\doms$ are not impacted by the modification of the label, whether the principal formula is labeled by one of the labels under focus or not. For $\excl$ we easily reach our goal by applying the induction hypothesis on the premise and then the rule as the principal formula cannot be the one under focus.
However, the other non-local rules require more care.

$\ax:$ then we have a proof of $\lseq{\mathcal R}{\mathcal T}{\Gamma,\lform{w'}{p(\vec{t})}}{\Delta,\lform{u'}{p(\vec{t})}}$ where $\rable{w'}{u'}$. We proceed by case distinction on the equality between $w$ and $u'$. If $w=u'$, then we need to prove $\lseq{\mathcal R}{\mathcal T}{\Gamma,\lform{w'}{p(\vec{t})}}{\Delta,\lform{u}{p(\vec{t})}}$. Now, note that we have $\rable{w'}{u'}$ and $\rable{u'}{u}$ by assumption as $w=u'$. Consequently, we have that $\rable{w'}{u}$ by transitivity. An application of the rule $\ax$ leads us to our goal. If $w\neq u'$, then we can simply reapply the rule as the two labeled formulae are not modified.

$\impr:$ then the last rule has the following form.
\begin{center}
\AxiomC{$\lseq{\mathcal R,\lrel{w'}{u'}}{\mathcal T}{\fsa,\lform{u'}{\phi}}{\fsb,\lform{u'}{\psi}}$}
\RightLabel{$\impr$}
\UnaryInfC{$\lseq{\mathcal R}{\mathcal T}{\fsa}{\fsb,\lform{w'}{\phi\imp\psi}}$}
\DisplayProof
\end{center}
If the principal formula is the labeled formula we intend to modify, then we have that $w'=w$.
As we have $\rable{w}{u}$ and not $\rable{u}{u'}$ because of the freshness of $u'$, we can use Lemma \ref{lem:brf-brb-admiss} to obtain a proof of $\lseq{\mathcal R,\lrel{u}{u'}}{\mathcal T}{\fsa,\lform{u'}{\phi}}{\fsb,\lform{u'}{\psi}}$. Then, it suffices to apply the rule $\impr$ to reach a proof of $\lseq{\mathcal R}{\mathcal T}{\fsa}{\fsb,\lform{u}{\phi\imp\psi}}$ of the adequate height. If the principal formula is not the one we intend to modify, then we can use the induction hypothesis in the premise and reapply the rule.
 
$\excr:$ then the last rule has the following form.
\begin{small}
\begin{center}
\AxiomC{$\lseq{\mathcal R}{\mathcal T}{\fsa}{\fsb,\lform{u'}{\phi\exc\psi},\lform{w'}{\phi}}$}
\AxiomC{$\lseq{\mathcal R}{\mathcal T}{\fsa,\lform{w'}{\psi}}{\fsb,\lform{u'}{\phi\exc\psi}}$}
\RightLabel{$\excr$}
\BinaryInfC{$\lseq{\mathcal R}{\mathcal T}{\fsa}{\fsb,\lform{u'}{\phi\exc\psi}}$}
\DisplayProof
\end{center}
\end{small}
If the principal formula is not the one we intend to modify, then we can use the induction hypothesis in the premise and reapply the rule.
If the principal formula is the labeled formula we intend to modify, then we have that $w=u'$.
As we have $\rable{w}{u}$ and $\rable{w'}{w}$ by the rule application, we get $\rable{w'}{u}$ by transitivity. This allows us to apply the induction hypothesis in both premises to obtain proofs of $\lseq{\mathcal R}{\mathcal T}{\fsa}{\fsb,\lform{u}{\phi\exc\psi},\lform{w'}{\phi}}$ and $\lseq{\mathcal R}{\mathcal T}{\fsa,\lform{w'}{\psi}}{\fsb,\lform{u}{\phi\exc\psi}}$. Then, it suffices to reapply the rule to obtain a proof of adequate height of $\lseq{\mathcal R}{\mathcal T}{\fsa}{\fsb,\lform{u}{\phi\exc\psi}}$.

Let us now turn to the first order rules. For all the rules for quantifiers where the modified labeled formula cannot be principal, i.e.~for $\existsl$ and $\alll$, it suffices to apply the induction hypothesis in the premise, and then reapply the rule. So, we are left with the rules $\allr$ and $\existsr$.

$\allr:$ then the last rule has the following form.
\begin{center}
\AxiomC{$\lseq{\mathcal R,\lrel{w'}{u'}}{\mathcal T,\lterm{u'}{y}}{\fsa}{\fsb,\lform{u'}{\phi(y/x)}}$}
\RightLabel{$\allr$}
\UnaryInfC{$\lseq{\mathcal R}{\mathcal T}{\fsa}{\fsb,\lform{w'}{\forall x\phi}}$}
\DisplayProof
\end{center}
If the principal formula is not the one we intend to modify, then we can use the induction hypothesis in the premise and reapply the rule.
If the principal formula is the labeled formula we intend to modify, then we have that $w'=w$.
As we have $\rable{w}{u}$ and not $\rable{u}{u'}$ because of the freshness of $u'$, we can use Lemma \ref{lem:brf-brb-admiss} to obtain a proof of $\lseq{\mathcal R,\lrel{u}{u'}}{\mathcal T,\lterm{u'}{y}}{\fsa}{\fsb,\lform{u'}{\phi(y/x)}}$. Then, it suffices to apply the rule $\allr$ to reach a proof of $\lseq{\mathcal R}{\mathcal T}{\fsa}{\fsb,\lform{u}{\forall x\phi}}$ of the adequate height. 

$\existsri:$ then the last rule has the following form.
\begin{center}
\AxiomC{$\lseq{\mathcal R}{\mathcal T}{\fsa}{\fsb,\lform{w'}{\exists x\phi},\lform{w'}{\phi(t/x)}}$}
\RightLabel{$\existsri$}
\UnaryInfC{$\lseq{\mathcal R}{\mathcal T}{\fsa}{\fsb,\lform{w'}{\exists x\phi}}$}
\DisplayProof
\end{center}
If the principal formula is not the one we intend to modify, then we can use the induction hypothesis in the premise and reapply the rule.
If the principal formula is the labeled formula we intend to modify, then we have that $w'=w$.
Then, we can apply the induction hypothesis on the premise twice to obtain a proof of $\lseq{\mathcal R}{\mathcal T}{\fsa}{\fsb,\lform{u}{\exists x\phi},\lform{u}{\phi(t/x)}}$. Note that in this case, we have that $t\in \termset(X_{u})$ as $\rable{w}{u}$. So, we can apply the rule $\existsri$ to reach our goal.

%$\existsrii:$ then the last rule has the following form.
%\begin{center}
%\AxiomC{$\lseq{\mathcal R}{\mathcal T,\lterm{v}{y}}{\fsa}{\fsb,\lform{w}{\exists x\phi},\lform{w}{\phi(t/x)}}$}
%\RightLabel{$\existsrii$}
%\UnaryInfC{$\lseq{\mathcal R}{\mathcal T}{\fsa}{\fsb,\lform{w}{\exists x\phi}}$}
%\DisplayProof
%\end{center}
%If the principal formula is not the one we intend to modify, then we can use the induction hypothesis in the premise and reapply the rule.
%If the principal formula is the labeled formula we intend to modify, then we have that $w'=w$.
%Then, we can apply the induction hypothesis on the premise twice to obtain a proof of $\lseq{\mathcal R}{\mathcal T,\lterm{v}{y}}{\fsa}{\fsb,\lform{u}{\exists x\phi},\lform{u}{\phi(t/x)}}$. Note that in this case, we have that $t\in \termset(X_{u})$ \ian{in the premise} as $\rable{w}{u}$ \ian{lemma?}. So, we can apply the rule $\existsrii$ to reach our goal.
\end{proof}

%\begin{customlem}{\ref{lem:lft-admiss}}
\begin{lemma}\label{lem:lft-admiss}
%Let $\modclass \in \{\idclass, \cdclass\}$. 
 The $\lft$ rule is hp-admissible in $\calc(\modclass)$.
\end{lemma}
%\end{customlem}

\begin{proof}
The proof of this lemma is very similar in spirit to the proof of Lemma \ref{lem:lwr-admiss}. Dually, the problematic rules here are $\doms$, $\impl$, $\excl$, $\existsl$ and $\alll$.

$\doms:$ then the last rule has the following form where $u:p(\vec t)$ is the principal formula. 
\begin{small}
\begin{center}
\AxiomC{$\pi_1$}
\UnaryInfC{$\R, \T, \Gamma, u:p(\vec t),  u: \VT{\vec t} \sar \Delta$}
\RightLabel{$\doms$}
\UnaryInfC{$\R, \T, \Gamma, u:p(\vec t) \sar \Delta$}
\DisplayProof
\end{center}
\end{small}
Applying the induction hypothesis to $\pi_1$, we get a proof $\pi_2$ of 
$\R, \T, \Gamma, w:p(\vec t), u:\VT{\vec t} \sar \Delta$ with the same height, from which we can obtain a proof of $\R, \T, \Gamma, w : p(\vec t) \sar \Delta$ as follows:
\begin{small}
\begin{center}
\AxiomC{$\R, \T, \Gamma, w:p(\vec t), u:\VT{\vec t} \sar \Delta$}
\RightLabel{$\wkv$}
\UnaryInfC{$\R, \T, \Gamma, w:p(\vec t), w:\VT{\vec t}, u:\VT{\vec t} \sar \Delta$}
\RightLabel{$\idr$}
\UnaryInfC{$\R, \T, \Gamma, w:p(\vec t), w:\VT{\vec t} \sar \Delta$}
\RightLabel{$\doms$}
\UnaryInfC{$\R, \T, \Gamma, w:p(\vec t) \sar \Delta$}
\DisplayProof
\end{center}
\end{small}
The instance of $\idr$ above is applicable because $\rable{w}{u}$. The hp-admissibility of $\lft$ in this case then follows from Lemma~\ref{lem:ndr-admiss} and Lemma~\ref{lem:wkv-admiss}. 

$\impl:$ then the last rule has the following form.
\begin{small}
\begin{center}
\AxiomC{$\lseq{\mathcal R}{\mathcal T}{\fsa,\lform{w'}{\phi\imp\psi}}{\fsb,\lform{u'}{\phi}}$}
\AxiomC{$\lseq{\mathcal R}{\mathcal T}{\fsa,\lform{w'}{\phi\imp\psi},\lform{u'}{\psi}}{\fsb}$}
\RightLabel{$\impl$}
\BinaryInfC{$\lseq{\mathcal R}{\mathcal T}{\fsa,\lform{w'}{\phi\imp\psi}}{\fsb}$}
\DisplayProof
\end{center}
\end{small}
If the principal formula is not the one we intend to modify, then we can use the induction hypothesis in the premise and reapply the rule.
If the principal formula is the labeled formula we intend to modify, then we have that $u=w'$.
As we have $\rable{w}{u}$ and $\rable{u}{u'}$ by the rule application, we get $\rable{w}{u'}$ by transitivity. This allows us to apply the induction hypothesis in both premises to obtain proofs of $\lseq{\mathcal R}{\mathcal T}{\fsa,\lform{w}{\phi\imp\psi}}{\fsb,\lform{u'}{\phi}}$ and $\lseq{\mathcal R}{\mathcal T}{\fsa,\lform{w}{\phi\imp\psi},\lform{u'}{\psi}}{\fsb}$. Then, it suffices to reapply the rule to obtain a proof of adequate height of $\lseq{\mathcal R}{\mathcal T}{\fsa,\lform{w}{\phi\imp\psi}}{\fsb}$.

$\excl:$ then the last rule has the following form.
\begin{center}
\AxiomC{$\lseq{\mathcal R,\lrel{u'}{w'}}{\mathcal T}{\fsa,\lform{u'}{\phi}}{\fsb,\lform{u'}{\psi}}$}
\RightLabel{$\excl$}
\UnaryInfC{$\lseq{\mathcal R}{\mathcal T}{\fsa,\lform{w'}{\phi\exc\psi}}{\fsb}$}
\DisplayProof
\end{center}
If the principal formula is not the one we intend to modify, then we can use the induction hypothesis in the premise and reapply the rule.
If the principal formula is the labeled formula we intend to modify, then we have that $u=w'$.
As we have $\rable{w}{u}$ and not $\rable{w}{u'}$ because of the freshness of $u'$, we can use Lemma \ref{lem:brf-brb-admiss} to obtain a proof of $\lseq{\mathcal R,\lrel{u'}{w}}{\mathcal T}{\fsa,\lform{u'}{\phi}}{\fsb,\lform{u'}{\psi}}$. Then, it suffices to apply the rule $\excl$ to reach a proof of $\lseq{\mathcal R}{\mathcal T}{\fsa,\lform{w}{\phi\exc\psi}}{\fsb}$ of the adequate height.

$\existsl:$ then the last rule has the following form.
\begin{center}
\AxiomC{$\lseq{\mathcal R}{\mathcal T,\lterm{w'}{y}}{\fsa,\lform{w'}{\phi(y/x)}}{\fsb}$}
\RightLabel{$\existsl$}
\UnaryInfC{$\lseq{\mathcal R}{\mathcal T}{\fsa,\lform{w'}{\exists x\phi}}{\fsb}$}
\DisplayProof
\end{center}
If the principal formula is not the one we intend to modify, then we can use the induction hypothesis in the premise and reapply the rule.
If the principal formula is the labeled formula we intend to modify, then we have that $u=w'$.
First, we can apply Lemma \ref{lem:wkv-admiss} on the premise to obtain a proof of $\lseq{\mathcal R}{\mathcal T,\lterm{w}{y},\lterm{u}{y}}{\fsa,\lform{u}{\phi(y/x)}}{\fsb}$. Then, as $\rable{w}{u}$ we apply Lemma \ref{lem:ndr-admiss} to obtain a proof of $\lseq{\mathcal R}{\mathcal T,\lterm{w}{y}}{\fsa,\lform{u}{\phi(y/x)}}{\fsb}$. Finally, as all previous lemmas preserve height, it suffices to apply the induction hypothesis on this proof to obtain a proof of $\lseq{\mathcal R}{\mathcal T,\lterm{w}{y}}{\fsa,\lform{w}{\phi(y/x)}}{\fsb}$. A simple application of the rule $\existsl$ reaches our goal.

$\allli:$ then the last rule has the following form.
\begin{center}
\AxiomC{$\lseq{\mathcal R}{\mathcal T}{\fsa,\lform{w'}{\forall x\phi},\lform{u'}{\phi(t/x)}}{\fsb}$}
\RightLabel{$\allli$}
\UnaryInfC{$\lseq{\mathcal R}{\mathcal T}{\fsa,\lform{w'}{\forall x\phi}}{\fsb}$}
\DisplayProof
\end{center}
If the principal formula is not the one we intend to modify, then we can use the induction hypothesis in the premise and reapply the rule.
If the principal formula is the labeled formula we intend to modify, then we have that $u=w'$.
Then, we can apply the induction hypothesis on the premise to obtain a proof of $\lseq{\mathcal R}{\mathcal T}{\fsa,\lform{w}{\forall x\phi},\lform{u'}{\phi(t/x)}}{\fsb}$. Note that we have $\rable{w}{u'}$ as $\rable{w}{u}$ and $\rable{u}{u'}$. So, we can apply the rule $\allli$ to reach our goal.

%$\alllii:$ then the last rule has the following form.
%\begin{center}
%\AxiomC{$\lseq{\mathcal R}{\mathcal T,\lterm{v}{y}}{\fsa,\lform{w'}{\forall x\phi},\lform{u'}{\phi(t/x)}}{\fsb}$}
%\RightLabel{$\alllii$}
%\UnaryInfC{$\lseq{\mathcal R}{\mathcal T}{\fsa,\lform{w'}{\forall x\phi}}{\fsb}$}
%\DisplayProof
%\end{center}
%If the principal formula is not the one we intend to modify, then we can use the induction hypothesis in the premise and reapply the rule.
%If the principal formula is the labeled formula we intend to modify, then we have that $u=w'$.
%Then, we can apply the induction hypothesis on the premise to obtain a proof of $\lseq{\mathcal R}{\mathcal T,\lterm{v}{y}}{\fsa,\lform{w}{\forall x\phi},\lform{u'}{\phi(t/x)}}{\fsb}$. Note that we have $\rable{w}{u'}$ as $\rable{w}{u}$ and $\rable{u}{u'}$. So, we can apply the rule $\alllii$ to reach our goal.
\end{proof}

For each of the rules in the following lemma, we obtain a proof via an usual induction on the height of proofs. 

\begin{lemma}\label{lem:con-dis-inv}
%Let $\modclass \in \{\idclass, \cdclass\}$. 
 The rules $\conl,\conr,\disl,\disr$ are hp-invertible in $\calc(\modclass)$.
\end{lemma}

%\begin{customlem}{\ref{lem:impl-excr-existsri-alll-inv}}
\begin{lemma}\label{lem:impl-excr-existsri-alll-inv}
%Let $\modclass \in \{\idclass, \cdclass\}$. 
 The rules $\impl$, $\excr$, $\existsr$, $\alll$ are hp-invertible in $\calc(\modclass)$.
\end{lemma}
%\end{customlem}

\begin{proof} The claim follows immediately by the hp-admissibility of $\iwk$ (Lemma~\ref{lem:iwk-admiss}).
\iffalse
For $\impl$: given a proof of $\lseq{\mathcal R}{\mathcal T}{\Gamma,\lform{w}{\phi\imp\psi}}{\Delta}$, we can simply use height-preserving internal weakening (Lemma \ref{lem:iwk-admiss}) to obtain proofs of $\lseq{\mathcal R}{\mathcal T}{\Gamma,\lform{w}{\phi\imp\psi}}{\Delta,\lform{u}{\phi}}$ and $\lseq{\mathcal R}{\mathcal T}{\Gamma,\lform{w}{\phi\imp\psi},\lform{u}{\psi}}{\Delta}$ of less or equal height.
For $\existsr$: given a proof of $\lseq{\mathcal R}{\mathcal T}{\Gamma}{\Delta,\lform{w}{\exists x\phi}}$, we can simply use height-preserving internal weakening (Lemma \ref{lem:iwk-admiss}) to obtain a proof of $\lseq{\mathcal R}{\mathcal T}{\Gamma}{\Delta,\lform{w}{\exists x\phi},\lform{w}{\phi(t/x)}}$ of less or equal height.
\fi
\end{proof}

%The invertibility of the rules $\impr$ and $\excl$ is proved by induction on the height of proofs. In some places, we need to use isomorphic labeled sequents to ensure the freshness of the label introduced by the rule.

The following two lemmas are shown by induction on the height of the given proof. The proofs are tedious due to the number of cases that need to be considered, but are straightoforward to argue.

\begin{lemma}\label{lem:impr-excl-inv}
%Let $\modclass \in \{\idclass, \cdclass\}$. 
 The rules $\impr$ and $\excl$ are hp-invertible in $\calc(\modclass)$.
\end{lemma}

%In the proof of the lemma below, when a rule crucially involves a freshness condition, we can use Lemma \ref{lem:psub-admiss} to rename variables and ensure that the freshness condition holds. Other cases are resolved in a straightforward manner. %of variables through some modifications. This is what we do in parts of the proofs of the next invertibility lemma.

\begin{lemma}\label{lem:existsl-allr-inv}
%Let $\modclass \in \{\idclass, \cdclass\}$. 
 The rules $\existsl$ and $\allr$ are hp-invertible in $\calc(\modclass)$.
\end{lemma}

%\begin{customlem}{\ref{lem:ctr-admiss}}
\begin{lemma}\label{lem:ctr-admiss}
%Let $\modclass \in \{\idclass, \cdclass\}$. 
 The rules $\ctrl$ and $\ctrr$ are hp-admissible in $\calc(\modclass)$.
\end{lemma}
%\end{customlem}

\begin{proof}
We simultaneously prove the hp-admissibility of the two rules. We reason by primary induction on the structure of $\phi$ (PIH) and secondary induction on the  height of derivations (SIH). In each case, we consider the last rule applied, and whether $\phi$ is principal in that rule. We omit the cases of initial rules as they are straightforward.

\medskip

We start by the rule $\ctrl$. If $\phi$ is not principal in the last rule, then we proceed as usual by applying SIH on $\phi$ in the premises of the rule, and then the rule. Note that all rules which have their principal formula on the right fall in this category.

If $\phi$ is principal, then for the rules $\conl$ and $\disl$ we use the invertibility lemmas proved previously, jointly with PIH. For the rules in which the principal formula is not deleted in the premises, i.e.~$\impl$ and $\alll$, we simply use SIH to contract the principal formula. We give the case of $\impl$ as an example.

$\impl:$ then the last rule has the following form, where $\phi=\chi\imp\psi$.
\begin{small}
\begin{center}
\AxiomC{$\lseq{\mathcal R}{\mathcal T}{\fsa,\lform{w}{\chi\imp\psi},\lform{w}{\chi\imp\psi}}{\fsb,\lform{u}{\chi}}$}
\AxiomC{$\lseq{\mathcal R}{\mathcal T}{\fsa,\lform{w}{\chi\imp\psi},\lform{w}{\chi\imp\psi},\lform{u}{\psi}}{\fsb}$}
\RightLabel{$\impl$}
\BinaryInfC{$\lseq{\mathcal R}{\mathcal T}{\fsa,\lform{w}{\chi\imp\psi},\lform{w}{\chi\imp\psi}}{\fsb}$}
\DisplayProof
\end{center}
\end{small}
We can simply apply SIH on both premises to obtain proofs of $\lseq{\mathcal R}{\mathcal T}{\fsa,\lform{w}{\chi\imp\psi}}{\fsb,\lform{u}{\chi}}$ and $\lseq{\mathcal R}{\mathcal T}{\fsa,\lform{w}{\chi\imp\psi},\lform{u}{\psi}}{\fsb}$. A simple application of $\impr$ gives us the desired result.

Finally, we are left with the rules $\excl$ and $\existsl$ which we treat individually.

$\excl:$ then the last rule has the following form, where $\phi=\chi\exc\psi$.
\begin{center}
\AxiomC{$\lseq{\mathcal R,\lrel{u}{w}}{\mathcal T}{\fsa,\lform{u}{\chi},\lform{w}{\chi\exc\psi}}{\fsb,\lform{u}{\psi}}$}
\RightLabel{$\excl$}
\UnaryInfC{$\lseq{\mathcal R}{\mathcal T}{\fsa,\lform{w}{\chi\exc\psi},\lform{w}{\chi\exc\psi}}{\fsb}$}
\DisplayProof
\end{center}
Then, we proceed as follows.
\begin{center}
\AxiomC{$\lseq{\mathcal R,\lrel{u}{w}}{\mathcal T}{\fsa,\lform{u}{\chi},\lform{w}{\chi\exc\psi}}{\fsb,\lform{u}{\psi}}$}
\RightLabel{Lem.\ref{lem:existsl-allr-inv}}
%\ref{lem:impr-excl-inv}}
\dashedLine
\UnaryInfC{$\lseq{\mathcal R,\lrel{u}{w},\lrel{v}{w}}{\mathcal T}{\fsa,\lform{u}{\chi},\lform{v}{\chi}}{\fsb,\lform{u}{\psi},\lform{v}{\psi}}$}
\RightLabel{Lem.\ref{lem:brf-brb-admiss}}
\dashedLine
\UnaryInfC{$\lseq{\mathcal R,\lrel{u}{w},\lrel{v}{u}}{\mathcal T}{\fsa,\lform{u}{\chi},\lform{v}{\chi}}{\fsb,\lform{u}{\psi},\lform{v}{\psi}}$}
\RightLabel{Lem.\ref{lem:mrgf-mrgb-admiss}}
\dashedLine
\UnaryInfC{$\lseq{\mathcal R,\lrel{u}{w}}{\mathcal T}{\fsa,\lform{u}{\chi},\lform{u}{\chi}}{\fsb,\lform{u}{\psi},\lform{u}{\psi}}$}
\RightLabel{PIH}
\dashedLine
\UnaryInfC{$\lseq{\mathcal R,\lrel{u}{w}}{\mathcal T}{\fsa,\lform{u}{\chi}}{\fsb,\lform{u}{\psi}}$}
\RightLabel{$\excl$}
\UnaryInfC{$\lseq{\mathcal R}{\mathcal T}{\fsa,\lform{w}{\chi\exc\psi}}{\fsb}$}
\DisplayProof
\end{center}
First, we use the invertibility Lemma \ref{lem:impr-excl-inv} on the initial premise. Second, we use Lemma \ref{lem:brf-brb-admiss} to push the branch $\lrel{v}{w}$ to $u$ and obtain $\lrel{v}{u}$. Third, we merge the points $v$ and $u$ using Lemma \ref{lem:mrgf-mrgb-admiss}. Once this is done, we can apply the induction hypothesis PIH and reapply the rule.

$\existsl:$ then the last rule has the following form, where $\phi=\exists x\psi$.
\begin{center}
\AxiomC{$\lseq{\mathcal R}{\mathcal T,\lterm{w}{y}}{\fsa,\lform{w}{\psi(y/x)},\lform{w}{\exists x\psi}}{\fsb}$}
\RightLabel{$\existsl$}
\UnaryInfC{$\lseq{\mathcal R}{\mathcal T}{\fsa,\lform{w}{\exists x\psi},\lform{w}{\exists x\psi}}{\fsb}$}
\DisplayProof
\end{center}
Then, we proceed as follows.
\begin{center}
\AxiomC{$\lseq{\mathcal R}{\mathcal T,\lterm{w}{y}}{\fsa,\lform{w}{\psi(y/x)},\lform{w}{\exists x\psi}}{\fsb}$}
\RightLabel{Lem.\ref{lem:existsl-allr-inv}}
\dashedLine
\UnaryInfC{$\lseq{\mathcal R}{\mathcal T,\lterm{w}{y},\lterm{w}{z}}{\fsa,\lform{w}{\psi(y/x)},\lform{w}{\psi(z/x)}}{\fsb}$}
\RightLabel{Lem.\ref{lem:psub-admiss}}
\dashedLine
\UnaryInfC{$\lseq{\mathcal R}{\mathcal T,\lterm{w}{y},\lterm{w}{y}}{\fsa,\lform{w}{\psi(y/x)},\lform{w}{\psi(y/x)}}{\fsb}$}
\RightLabel{Lem.\ref{lem:ndr-admiss}}
\dashedLine
\UnaryInfC{$\lseq{\mathcal R}{\mathcal T,\lterm{w}{y}}{\fsa,\lform{w}{\psi(y/x)},\lform{w}{\psi(y/x)}}{\fsb}$}
\RightLabel{PIH}
\dashedLine
\UnaryInfC{$\lseq{\mathcal R}{\mathcal T,\lterm{w}{y}}{\fsa,\lform{w}{\psi(y/x)}}{\fsb}$}
\RightLabel{$\existsl$}
\UnaryInfC{$\lseq{\mathcal R}{\mathcal T}{\fsa,\lform{w}{\exists x\psi}}{\fsb}$}
\DisplayProof
\end{center}
First, we use the invertibility Lemma \ref{lem:existsl-allr-inv} on the initial premise. Second, we use Lemma \ref{lem:psub-admiss} to rename the (fresh) variable $z$ to $y$. Third, we contract the labeled variable $\lterm{w}{y}$ using Lemma \ref{lem:ndr-admiss}. Once this is done, we can apply the induction hypothesis PIH and reapply the rule.

\medskip

Now, let us turn to the rule $\ctrr$. If $\phi$ is not principal in the last rule, then we proceed as usual by applying SIH on $\phi$ in the premises of the rule, and then the rule. Note that this time, all rules which have their principal formula on the left fall in this category.

If $\phi$ is principal, then for the rules $\conr$ and $\disr$ we use the invertibility lemmas proved previously, jointly with PIH. For the rules in which the principal formula is not deleted in the premises, i.e.~$\doms$, $\excr$ and $\existsr$, we simply use SIH to contract the principal formula. We give the case of $\excr$ as an example.

$\excr:$ then the last rule has the following form, where $\phi=\chi\exc\psi$.
\begin{small}
\begin{center}
\AxiomC{$\lseq{\mathcal R}{\mathcal T}{\fsa}{\fsb,\lform{w}{\chi\exc\psi},\lform{w}{\chi\exc\psi},\lform{u}{\chi}}$}
\AxiomC{$\lseq{\mathcal R}{\mathcal T}{\fsa,\lform{u}{\psi}}{\fsb,\lform{w}{\chi\exc\psi},\lform{w}{\chi\exc\psi}}$}
\RightLabel{$\excr$}
\BinaryInfC{$\lseq{\mathcal R}{\mathcal T}{\fsa}{\fsb,\lform{w}{\chi\exc\psi},\lform{w}{\chi\exc\psi}}$}
\DisplayProof
\end{center}
\end{small}
We can simply apply SIH on both premises to obtain proofs of $\lseq{\mathcal R}{\mathcal T}{\fsa}{\fsb,\lform{w}{\chi\exc\psi},\lform{u}{\chi}}$ and $\lseq{\mathcal R}{\mathcal T}{\fsa,\lform{u}{\psi}}{\fsb,\lform{w}{\chi\exc\psi}}$. A simple application of the rule gives us the desired result.

Finally, we are left with the rules $\impr$ and $\allr$ which we treat individually.

$\impr:$ then the last rule has the following form, where $\phi=\chi\imp\psi$.
\begin{center}
\AxiomC{$\lseq{\mathcal R,\lrel{w}{u}}{\mathcal T}{\fsa,\lform{u}{\chi}}{\fsb,\lform{u}{\psi},\lform{w}{\chi\imp\psi}}$}
\RightLabel{$\impr$}
\UnaryInfC{$\lseq{\mathcal R}{\mathcal T}{\fsa}{\fsb,\lform{w}{\chi\imp\psi},\lform{w}{\chi\imp\psi}}$}
\DisplayProof
\end{center}
Then, we proceed as follows.
\begin{center}
\AxiomC{$\lseq{\mathcal R,\lrel{w}{u}}{\mathcal T}{\fsa,\lform{u}{\chi}}{\fsb,\lform{u}{\psi},\lform{w}{\chi\imp \psi}}$}
\RightLabel{Lem.\ref{lem:impr-excl-inv}}
\dashedLine
\UnaryInfC{$\lseq{\mathcal R,\lrel{w}{u},\lrel{w}{v}}{\mathcal T}{\fsa,\lform{u}{\chi},\lform{v}{\chi}}{\fsb,\lform{u}{\psi},\lform{v}{\psi}}$}
\RightLabel{Lem.\ref{lem:brf-brb-admiss}}
\dashedLine
\UnaryInfC{$\lseq{\mathcal R,\lrel{w}{u},\lrel{u}{v}}{\mathcal T}{\fsa,\lform{u}{\chi},\lform{v}{\chi}}{\fsb,\lform{u}{\psi},\lform{v}{\psi}}$}
\RightLabel{Lem.\ref{lem:mrgf-mrgb-admiss}}
\dashedLine
\UnaryInfC{$\lseq{\mathcal R,\lrel{w}{u}}{\mathcal T}{\fsa,\lform{u}{\chi},\lform{u}{\chi}}{\fsb,\lform{u}{\psi},\lform{u}{\psi}}$}
\RightLabel{PIH}
\dashedLine
\UnaryInfC{$\lseq{\mathcal R,\lrel{w}{u}}{\mathcal T}{\fsa,\lform{u}{\chi}}{\fsb,\lform{u}{\psi}}$}
\RightLabel{$\impr$}
\UnaryInfC{$\lseq{\mathcal R}{\mathcal T}{\fsa}{\fsb,\lform{w}{\chi\imp\psi}}$}
\DisplayProof
\end{center}
First, we use the invertibility Lemma \ref{lem:impr-excl-inv} on the initial premise. Second, we use Lemma \ref{lem:brf-brb-admiss} to push the branch $\lrel{w}{v}$ to $u$ and obtain $\lrel{u}{v}$. Third, we merge the points $v$ and $u$ using Lemma \ref{lem:mrgf-mrgb-admiss}. Once this is done, we can apply the induction hypothesis PIH and reapply the rule.

$\allr:$ then the last rule has the following form, where $\phi=\forall x\psi$.
\begin{center}
\AxiomC{$\lseq{\mathcal R,\lrel{w}{u}}{\mathcal T,\lterm{u}{y}}{\fsa}{\fsb,\lform{u}{\psi(y/x)},\lform{w}{\forall x\psi}}$}
\RightLabel{$\allr$}
\UnaryInfC{$\lseq{\mathcal R}{\mathcal T}{\fsa}{\fsb,\lform{w}{\forall x\psi},\lform{w}{\forall x\psi}}$}
\DisplayProof
\end{center}
Then, we proceed as follows.
\begin{center}
\AxiomC{$\lseq{\mathcal R,\lrel{w}{u}}{\mathcal T,\lterm{u}{y}}{\fsa}{\fsb,\lform{u}{\psi(y/x)},\lform{w}{\forall x\psi}}$}
\RightLabel{Lem.\ref{lem:existsl-allr-inv}}
\dashedLine
\UnaryInfC{$\lseq{\mathcal R,\lrel{w}{u},\lrel{w}{v}}{\mathcal T,\lterm{u}{y},\lterm{v}{z}}{\fsa}{\fsb,\lform{u}{\psi(y/x)},\lform{v}{\psi(z/x)}}$}
\RightLabel{Lem.\ref{lem:brf-brb-admiss}}
\dashedLine
\UnaryInfC{$\lseq{\mathcal R,\lrel{w}{u},\lrel{u}{v}}{\mathcal T,\lterm{u}{y},\lterm{v}{z}}{\fsa}{\fsb,\lform{u}{\psi(y/x)},\lform{v}{\psi(z/x)}}$}
\RightLabel{Lem.\ref{lem:mrgf-mrgb-admiss}}
\dashedLine
\UnaryInfC{$\lseq{\mathcal R,\lrel{w}{u}}{\mathcal T,\lterm{u}{y},\lterm{u}{z}}{\fsa}{\fsb,\lform{u}{\psi(y/x)},\lform{u}{\psi(z/x)}}$}
\RightLabel{Lem.\ref{lem:psub-admiss}}
\dashedLine
\UnaryInfC{$\lseq{\mathcal R,\lrel{w}{u}}{\mathcal T,\lterm{u}{y},\lterm{u}{y}}{\fsa}{\fsb,\lform{u}{\psi(y/x)},\lform{u}{\psi(y/x)}}$}
\RightLabel{Lem.\ref{lem:ndr-admiss}}
\dashedLine
\UnaryInfC{$\lseq{\mathcal R,\lrel{w}{u}}{\mathcal T,\lterm{u}{y}}{\fsa}{\fsb,\lform{u}{\psi(y/x)},\lform{u}{\psi(y/x)}}$}
\RightLabel{PIH}
\dashedLine
\UnaryInfC{$\lseq{\mathcal R,\lrel{w}{u}}{\mathcal T,\lterm{u}{y}}{\fsa}{\fsb,\lform{u}{\psi(y/x)}}$}
\RightLabel{$\allr$}
\UnaryInfC{$\lseq{\mathcal R}{\mathcal T}{\fsa}{\fsb,\lform{w}{\forall x\psi}}$}
\DisplayProof
\end{center}
First, we use the invertibility Lemma \ref{lem:existsl-allr-inv} on the initial premise. Second, we use Lemma \ref{lem:brf-brb-admiss} to push the branch $\lrel{w}{v}$ to $u$ and obtain $\lrel{u}{v}$. Third, we merge the points $v$ and $u$ using Lemma \ref{lem:mrgf-mrgb-admiss}. Fourth, we use Lemma \ref{lem:psub-admiss} to rename the (fresh) variable $z$ to $y$. Fifth, we contract the labeled variable $\lterm{w}{y}$ using Lemma \ref{lem:ndr-admiss}. Once this is done, we can apply the induction hypothesis PIH and reapply the rule.
\end{proof}

\begin{customthm}{\ref{thm:cut-elim-int}}[Cut-elimination]
The $\cut$ rule is admissible in $\calc(\modclass)$. %$\nid$ and $\ncd$.
\end{customthm}

%UPDATED PROOF

\begin{proof} 
We proceed by a primary induction (PIH) on the complexity of the cut formula, and a secondary induction (SIH) on the sum of the heights of the proofs of the premises of $\cut$. Assume that we have proofs of the following form, with $\rable{w}{u}$.
\begin{center}
\begin{tabular}{c@{\hspace{1cm}} c}
\AxiomC{$\pi_1$}
\RightLabel{$r_1$}
\UnaryInfC{$\lseq{\mathcal R}{\mathcal T}{\fsa}{\fsb,\lform{w}{\phi}}$}
\DisplayProof
&
\AxiomC{$\pi_2$}
\RightLabel{$r_2$}
\UnaryInfC{$\lseq{\mathcal R}{\mathcal T}{\fsa,\lform{u}{\phi}}{\fsb}$}
\DisplayProof \\
\end{tabular}
\end{center}
We prove that there is a proof of $\lseq{\mathcal R}{\mathcal T}{\fsa}{\fsb}$ by a case distinction on $r_1$ and $r_2$, the last rules applied in the above proofs.

\noindent \textbf{(I)} $\mathbf{r_1} = \ax:$ Then $\lseq{\mathcal R}{\mathcal T}{\fsa}{\fsb,\lform{w}{\phi}}$ is of the form $\lseq{\mathcal R}{\mathcal T}{\fsa_0,\lform{v_0}{p(\vec{t})}}{\fsb_0,\lform{v_1}{p(\vec{t})}}$ where $\rable{v_0}{v_1}$. 
If $\lform{v_1}{p(\vec{t})}$ is $\lform{w}{\phi}$, then we have that $\lseq{\mathcal R}{\mathcal T}{\fsa,\lform{u}{\phi}}{\fsb}$ is of the form $\lseq{\mathcal R}{\mathcal T}{\fsa_0,\lform{v_0}{p(\vec{t})},\lform{u}{p(\vec{t})}}{\fsb}$ where $\fsa=\fsa_0,\lform{v_0}{p(\vec{t})}$. Given that $\rable{v_0}{v_1}$ and $\rable{v_1}{u}$, we can apply Lemma \ref{lem:lft-admiss} on the latter to obtain a proof of $\lseq{\mathcal R}{\mathcal T}{\fsa_0,\lform{v_0}{p(\vec{t})},\lform{v_0}{p(\vec{t})}}{\fsb}$. Consequently, we obtain a proof of $\lseq{\mathcal R}{\mathcal T}{\fsa_0,\lform{v_0}{p(\vec{t})}}{\fsb}$, i.e.~of $\lseq{\mathcal R}{\mathcal T}{\fsa}{\fsb}$, using Lemma \ref{lem:ctr-admiss}. 
If $\lform{v_1}{p(\vec{t})}$ is not $\lform{w}{\phi}$, then we have that $\lseq{\mathcal R}{\mathcal T}{\fsa}{\fsb}$ is of the form $\lseq{\mathcal R}{\mathcal T}{\fsa_0,\lform{v_0}{p(\vec{t})}}{\fsb_0,\lform{v_1}{p(\vec{t})}}$ where $\rable{v_0}{v_1}$. The latter is easily provable using the rule $\ax$.

\noindent \textbf{(II)} $\mathbf{r_1=\botl:}$ Then $\lseq{\mathcal R}{\mathcal T}{\fsa}{\fsb,\lform{w}{\phi}}$ is of the form $\lseq{\mathcal R}{\mathcal T}{\fsa_0,\lform{v}{\bot}}{\fsb,\lform{w}{\phi}}$ where $\fsa=\fsa_0,\lform{v}{\bot}$. Consequently, we know that that $\lseq{\mathcal R}{\mathcal T}{\fsa}{\fsb}$ is of the form $\lseq{\mathcal R}{\mathcal T}{\fsa_0,\lform{v}{\bot}}{\fsb}$. We straightforwardly prove the latter via a single application of $\botl$.

\noindent \textbf{(III)} $\mathbf{r_1=\topr:}$ Then $\lseq{\mathcal R}{\mathcal T}{\fsa}{\fsb,\lform{w}{\phi}}$ is of the form $\lseq{\mathcal R}{\mathcal T}{\fsa}{\fsb_0,\lform{v}{\top}}$. If $\lform{w}{\phi}$ is $\lform{v}{\top}$, then we have that $\lseq{\mathcal R}{\mathcal T}{\fsa,\lform{u}{\phi}}{\fsb}$ is of the form $\lseq{\mathcal R}{\mathcal T}{\fsa,\lform{u}{\top}}{\fsb}$. Thus, we obtain a proof of $\lseq{\mathcal R}{\mathcal T}{\fsa}{\fsb}$ using Lemma \ref{lem:botr-topl-admiss}. If $\lform{w}{\phi}$ is not $\lform{v}{\top}$, then we have that $\lseq{\mathcal R}{\mathcal T}{\fsa}{\fsb}$ is of the form $\lseq{\mathcal R}{\mathcal T}{\fsa}{\fsb_0,\lform{v}{\top}}$. A single application of $\topr$ gets us to our goal.

%\tim{I suppose we will have to change the notation in the following case, since we no longer use label terms, but rather, labeled variables.}
%\ian{Well-spotted! I will leave these comments here for now as a reminder.}

\noindent \textbf{(IV)} $\mathbf{r_1} = \doms:$ Then $\lseq{\mathcal R}{\mathcal T}{\fsa}{\fsb,\lform{w}{\phi}}$ is of the form $\lseq{\mathcal R}{\mathcal T}{\fsa_0, \lform{v}{p(\vec{t})}}{\fsb,\lform{w}{\phi}}$ and we have a proof of $\lseq{\mathcal R}{\mathcal T, \lterm{v}{\VT{\vec{t}}}}{\fsa_0, \lform{v}{p(\vec{t})}}{\fsb,\lform{w}{\phi}}$. 
Consequently, we know that $\lseq{\mathcal R}{\mathcal T}{\fsa}{\fsb}$ is of the form $\lseq{\mathcal R}{\mathcal T}{\fsa_0, \lform{v}{p(\vec{t})}}{\fsb}$. We also have that $\lseq{\mathcal R}{\mathcal T}{\fsa,\lform{u}{\phi}}{\fsb}$ is of the form $\lseq{\mathcal R}{\mathcal T}{\fsa_0, \lform{v}{p(\vec{t})},\lform{u}{\phi}}{\fsb}$.
We can apply Lemma \ref{lem:wkv-admiss} repetitively on the proof of the latter to obtain a proof of 
$\lseq{\mathcal R}{\mathcal T, \lterm{v}{\VT{\vec{t}}}}{\fsa_0, \lform{v}{p(\vec{t})},\lform{u}{\phi}}{\fsb}$ which we call $S$, while preserving height. Then, we proceed as follows.

\begin{center}
\AxiomC{$\lseq{\mathcal R}{\mathcal T, \lterm{v}{\VT{\vec{t}}}}{\fsa_0, \lform{v}{p(\vec{t})}}{\fsb,\lform{w}{\phi}}$}
\AxiomC{$S$}
\RightLabel{SIH}
\dashedLine
\BinaryInfC{$\lseq{\mathcal R}{\mathcal T, \lterm{v}{\VT{\vec{t}}}}{\fsa_0, \lform{v}{p(\vec{t})}}{\fsb}$}
\RightLabel{$\doms$}
\UnaryInfC{$\lseq{\mathcal R}{\mathcal T}{\fsa_0, \lform{v}{p(\vec{t})}}{\fsb}$}
\DisplayProof
\end{center}
Note that the instance of SIH is justified because the sum of the heights of the proofs of the premises has decreased.

\noindent \textbf{(V)} $\mathbf{r_1=\conl:}$ Then $\lseq{\mathcal R}{\mathcal T}{\fsa}{\fsb,\lform{w}{\phi}}$ is of the form $\lseq{\mathcal R}{\mathcal T}{\fsa_0,\lform{v}{\psi\land\chi}}{\fsb,\lform{w}{\phi}}$ and we have a proof of $\lseq{\mathcal R}{\mathcal T}{\fsa_0,\lform{v}{\psi},\lform{v}{\chi}}{\fsb,\lform{w}{\phi}}$. Consequently, we know that $\lseq{\mathcal R}{\mathcal T}{\fsa}{\fsb}$ is of the form $\lseq{\mathcal R}{\mathcal T}{\fsa_0,\lform{v}{\psi\land\chi}}{\fsb}$. We also have that $\lseq{\mathcal R}{\mathcal T}{\fsa,\lform{u}{\phi}}{\fsb}$ is of the form $\lseq{\mathcal R}{\mathcal T}{\fsa_0,\lform{v}{\psi\land\chi},\lform{u}{\phi}}{\fsb}$. We apply Lemma \ref{lem:con-dis-inv} on the proof of the latter  sequent to obtain a proof of $\lseq{\mathcal R}{\mathcal T}{\fsa_0,\lform{v}{\psi},\lform{v}{\chi},\lform{u}{\phi}}{\fsb}$ preserving height. Thus, we proceed as follows.
\begin{center}
\AxiomC{$\lseq{\mathcal R}{\mathcal T}{\fsa_0,\lform{v}{\psi},\lform{v}{\chi}}{\fsb,\lform{w}{\phi}}$}
\AxiomC{$\lseq{\mathcal R}{\mathcal T}{\fsa_0,\lform{v}{\psi},\lform{v}{\chi},\lform{u}{\phi}}{\fsb}$}
\RightLabel{SIH}
\dashedLine
\BinaryInfC{$\lseq{\mathcal R}{\mathcal T}{\fsa_0,\lform{v}{\psi},\lform{v}{\chi}}{\fsb}$}
\RightLabel{$\conl$}
\UnaryInfC{$\lseq{\mathcal R}{\mathcal T}{\fsa_0,\lform{v}{\psi\land\chi}}{\fsb}$}
\DisplayProof
\end{center}
Note that the instance of SIH is justified as the sum of the heights of the proofs of its premises is smaller than the one of the initial cut.

\noindent \textbf{(VI)} $\mathbf{r_1=\conr:}$ Then $\lseq{\mathcal R}{\mathcal T}{\fsa}{\fsb,\lform{w}{\phi}}$ is of the form $\lseq{\mathcal R}{\mathcal T}{\fsa}{\fsb_0,\lform{v}{\psi\land\chi}}$ and we have proofs of $\lseq{\mathcal R}{\mathcal T}{\fsa}{\fsb_0,\lform{v}{\psi}}$ and $\lseq{\mathcal R}{\mathcal T}{\fsa}{\fsb_0,\lform{v}{\chi}}$. 
If $\lform{w}{\phi}$ is $\lform{v}{\psi\land\chi}$, then we have proofs of $\lseq{\mathcal R}{\mathcal T}{\fsa}{\fsb,\lform{v}{\psi}}$ and $\lseq{\mathcal R}{\mathcal T}{\fsa}{\fsb,\lform{v}{\chi}}$, and $\lseq{\mathcal R}{\mathcal T}{\fsa,\lform{u}{\phi}}{\fsb}$ is of the form $\lseq{\mathcal R}{\mathcal T}{\fsa,\lform{u}{\psi\land\chi}}{\fsb}$. Then, we proceed as follows where $\pi$ is the first proof.
\begin{small}
\begin{center}
\AxiomC{$\lseq{\mathcal R}{\mathcal T}{\fsa}{\fsb,\lform{v}{\chi}}$}
\RightLabel{Lem.\ref{lem:iwk-admiss}}
\dashedLine
\UnaryInfC{$\lseq{\mathcal R}{\mathcal T}{\fsa,\lform{u}{\psi}}{\fsb,\lform{v}{\chi}}$}
\AxiomC{$\lseq{\mathcal R}{\mathcal T}{\fsa,\lform{u}{\psi\land\chi}}{\fsb}$}
\RightLabel{Lem.\ref{lem:con-dis-inv}}
\dashedLine
\UnaryInfC{$\lseq{\mathcal R}{\mathcal T}{\fsa,\lform{u}{\psi},\lform{u}{\chi}}{\fsb}$}
\RightLabel{PIH}
\dashedLine
\BinaryInfC{$\lseq{\mathcal R}{\mathcal T}{\fsa,\lform{u}{\psi}}{\fsb}$}
\DisplayProof
\end{center}
\end{small}
\begin{center}
\AxiomC{$\lseq{\mathcal R}{\mathcal T}{\fsa}{\fsb,\lform{v}{\psi}}$}

\AxiomC{$\pi$}

\RightLabel{PIH}
\dashedLine
\BinaryInfC{$\lseq{\mathcal R}{\mathcal T}{\fsa}{\fsb}$}
\DisplayProof
\end{center}
If $\lform{w}{\phi}$ is not $\lform{v}{\psi\land\chi}$, then we have proofs of $\lseq{\mathcal R}{\mathcal T}{\fsa}{\fsb_1,\lform{v}{\psi},\lform{w}{\phi}}$ and $\lseq{\mathcal R}{\mathcal T}{\fsa}{\fsb_1,\lform{v}{\chi},\lform{w}{\phi}}$, and $\lseq{\mathcal R}{\mathcal T}{\fsa,\lform{u}{\phi}}{\fsb}$ is of the form $\lseq{\mathcal R}{\mathcal T}{\fsa,\lform{u}{\phi}}{\fsb_1,\lform{v}{\psi\land\chi}}$. Then, we proceed as follows where $\pi$ is the first proof displayed.

\begin{small}
\begin{center}
\AxiomC{$\lseq{\mathcal R}{\mathcal T}{\fsa}{\fsb_1,\lform{v}{\psi},\lform{w}{\phi}}$}
\AxiomC{$\lseq{\mathcal R}{\mathcal T}{\fsa,\lform{u}{\phi}}{\fsb_1,\lform{v}{\psi\land\chi}}$}
\RightLabel{Lem.\ref{lem:con-dis-inv}}
\dashedLine
\UnaryInfC{$\lseq{\mathcal R}{\mathcal T}{\fsa,\lform{u}{\phi}}{\fsb_1,\lform{v}{\psi}}$}
\RightLabel{SIH}
\dashedLine
\BinaryInfC{$\lseq{\mathcal R}{\mathcal T}{\fsa}{\fsb_1,\lform{v}{\psi}}$}
\DisplayProof
\end{center}
\end{small}
\begin{small}
\begin{center}
\AxiomC{$\pi$}

\AxiomC{$\lseq{\mathcal R}{\mathcal T}{\fsa}{\fsb_1,\lform{v}{\chi},\lform{w}{\phi}}$}
\AxiomC{$\lseq{\mathcal R}{\mathcal T}{\fsa,\lform{u}{\phi}}{\fsb_1,\lform{v}{\psi\land\chi}}$}
\RightLabel{Lem.\ref{lem:con-dis-inv}}
\dashedLine
\UnaryInfC{$\lseq{\mathcal R}{\mathcal T}{\fsa,\lform{u}{\phi}}{\fsb_1,\lform{v}{\chi}}$}
\RightLabel{SIH}
\dashedLine
\BinaryInfC{$\lseq{\mathcal R}{\mathcal T}{\fsa}{\fsb_1,\lform{v}{\chi}}$}

\RightLabel{$\conr$}
\BinaryInfC{$\lseq{\mathcal R}{\mathcal T}{\fsa}{\fsb_1,\lform{v}{\psi\land\chi}}$}
\DisplayProof
\end{center}
\end{small}

\noindent \textbf{(VII)} $\mathbf{r_1=\disl:}$ Then $\lseq{\mathcal R}{\mathcal T}{\fsa}{\fsb,\lform{w}{\phi}}$ is of the form $\lseq{\mathcal R}{\mathcal T}{\fsa_0,\lform{v}{\psi\lor\chi}}{\fsb,\lform{w}{\phi}}$ and we have proofs of $\lseq{\mathcal R}{\mathcal T}{\fsa_0,\lform{v}{\psi}}{\fsb,\lform{w}{\phi}}$ and $\lseq{\mathcal R}{\mathcal T}{\fsa_0,\lform{v}{\chi}}{\fsb,\lform{w}{\phi}}$. Consequently, we know that $\lseq{\mathcal R}{\mathcal T}{\fsa}{\fsb}$ is of the form $\lseq{\mathcal R}{\mathcal T}{\fsa_0,\lform{v}{\psi\lor\chi}}{\fsb}$. We also have that $\lseq{\mathcal R}{\mathcal T}{\fsa,\lform{u}{\phi}}{\fsb}$ is of the form $\lseq{\mathcal R}{\mathcal T}{\fsa_0,\lform{v}{\psi\lor\chi},\lform{u}{\phi}}{\fsb}$. We apply Lemma \ref{lem:con-dis-inv} on the proof of the latter  sequent to obtain proofs of $\lseq{\mathcal R}{\mathcal T}{\fsa_0,\lform{v}{\psi},\lform{u}{\phi}}{\fsb}$ and $\lseq{\mathcal R}{\mathcal T}{\fsa_0,\lform{v}{\chi},\lform{u}{\phi}}{\fsb}$ preserving height. Thus, we proceed as follows where $\pi$ is the first proof.
\begin{center}
\AxiomC{$\lseq{\mathcal R}{\mathcal T}{\fsa_0,\lform{v}{\chi}}{\fsb,\lform{w}{\phi}}$}
\AxiomC{$\lseq{\mathcal R}{\mathcal T}{\fsa_0,\lform{v}{\chi},\lform{u}{\phi}}{\fsb}$}
\RightLabel{SIH}
\dashedLine
\BinaryInfC{$\lseq{\mathcal R}{\mathcal T}{\fsa_0,\lform{v}{\chi}}{\fsb}$}
\DisplayProof
\end{center}
\begin{small}
\begin{center}
\AxiomC{$\lseq{\mathcal R}{\mathcal T}{\fsa_0,\lform{v}{\psi}}{\fsb,\lform{w}{\phi}}$}
\AxiomC{$\lseq{\mathcal R}{\mathcal T}{\fsa_0,\lform{v}{\psi},\lform{u}{\phi}}{\fsb}$}
\RightLabel{SIH}
\dashedLine
\BinaryInfC{$\lseq{\mathcal R}{\mathcal T}{\fsa_0,\lform{v}{\psi}}{\fsb}$}

\AxiomC{$\pi$}

\RightLabel{$\disl$}
\BinaryInfC{$\lseq{\mathcal R}{\mathcal T}{\fsa_0,\lform{v}{\psi\lor\chi}}{\fsb}$}
\DisplayProof
\end{center}
\end{small}
Note that both instances of SIH are justified as the sum of the heights of the proofs of their premises is smaller than the one of the initial cut.

\noindent \textbf{(VIII)} $\mathbf{r_1=\disr:}$ Then $\lseq{\mathcal R}{\mathcal T}{\fsa}{\fsb,\lform{w}{\phi}}$ is of the form $\lseq{\mathcal R}{\mathcal T}{\fsa}{\fsb_0,\lform{v}{\psi\lor\chi}}$ and we have a proof of $\lseq{\mathcal R}{\mathcal T}{\fsa}{\fsb_0,\lform{v}{\psi},\lform{v}{\chi}}$. 
If $\lform{w}{\phi}$ is $\lform{v}{\psi\lor\chi}$, then we have a proof of $\lseq{\mathcal R}{\mathcal T}{\fsa}{\fsb,\lform{v}{\psi},\lform{v}{\chi}}$, and $\lseq{\mathcal R}{\mathcal T}{\fsa,\lform{u}{\phi}}{\fsb}$ is of the form $\lseq{\mathcal R}{\mathcal T}{\fsa,\lform{u}{\psi\lor\chi}}{\fsb}$. Then, we proceed as follows where $\pi$ is the first proof displayed.
\begin{center}
\AxiomC{$\lseq{\mathcal R}{\mathcal T}{\fsa}{\fsb,\lform{v}{\psi},\lform{v}{\chi}}$}

\AxiomC{$\lseq{\mathcal R}{\mathcal T}{\fsa,\lform{u}{\psi\lor\chi}}{\fsb}$}
\RightLabel{Lem.\ref{lem:con-dis-inv}}
\dashedLine
\UnaryInfC{$\lseq{\mathcal R}{\mathcal T}{\fsa,\lform{u}{\chi}}{\fsb}$}
\RightLabel{Lem.\ref{lem:iwk-admiss}}
\dashedLine
\UnaryInfC{$\lseq{\mathcal R}{\mathcal T}{\fsa,\lform{u}{\chi}}{\fsb,\lform{v}{\psi}}$}

\RightLabel{PIH}
\dashedLine
\BinaryInfC{$\lseq{\mathcal R}{\mathcal T}{\fsa}{\fsb,\lform{v}{\psi}}$}
\DisplayProof
\end{center}
\begin{center}
\AxiomC{$\pi$}

\AxiomC{$\lseq{\mathcal R}{\mathcal T}{\fsa,\lform{u}{\psi\lor\chi}}{\fsb}$}
\RightLabel{Lem.\ref{lem:con-dis-inv}}
\dashedLine
\UnaryInfC{$\lseq{\mathcal R}{\mathcal T}{\fsa,\lform{u}{\psi}}{\fsb}$}

\RightLabel{PIH}
\dashedLine
\BinaryInfC{$\lseq{\mathcal R}{\mathcal T}{\fsa}{\fsb}$}
\DisplayProof
\end{center}
If $\lform{w}{\phi}$ is not $\lform{v}{\psi\lor\chi}$, then we have a proof of $\lseq{\mathcal R}{\mathcal T}{\fsa}{\fsb_1,\lform{v}{\psi},\lform{v}{\chi},\lform{w}{\phi}}$, and $\lseq{\mathcal R}{\mathcal T}{\fsa,\lform{u}{\phi}}{\fsb}$ is of the form $\lseq{\mathcal R}{\mathcal T}{\fsa,\lform{u}{\phi}}{\fsb_1,\lform{v}{\psi\lor\chi}}$. Then, we proceed as follows.
\begin{center}
\AxiomC{$\lseq{\mathcal R}{\mathcal T}{\fsa}{\fsb_1,\lform{v}{\psi},\lform{v}{\chi},\lform{w}{\phi}}$}

\AxiomC{$\lseq{\mathcal R}{\mathcal T}{\fsa,\lform{u}{\phi}}{\fsb_1,\lform{v}{\psi\lor\chi}}$}
\RightLabel{Lem.\ref{lem:con-dis-inv}}
\dashedLine
\UnaryInfC{$\lseq{\mathcal R}{\mathcal T}{\fsa,\lform{u}{\phi}}{\fsb_1,\lform{v}{\psi},\lform{v}{\chi}}$}

\RightLabel{SIH}
\dashedLine
\BinaryInfC{$\lseq{\mathcal R}{\mathcal T}{\fsa}{\fsb_1,\lform{v}{\psi},\lform{v}{\chi}}$}
\RightLabel{$\disr$}
\UnaryInfC{$\lseq{\mathcal R}{\mathcal T}{\fsa}{\fsb_1,\lform{v}{\psi\lor\chi}}$}
\DisplayProof
\end{center}

\noindent \textbf{(IX)} $\mathbf{r_1=\impl:}$ Then $\lseq{\mathcal R}{\mathcal T}{\fsa}{\fsb,\lform{w}{\phi}}$ is of the form $\lseq{\mathcal R}{\mathcal T}{\fsa_0,\lform{v}{\psi\imp\chi}}{\fsb,\lform{w}{\phi}}$ and we have proofs of $\lseq{\mathcal R}{\mathcal T}{\fsa_0,\lform{v}{\psi\imp\chi}}{\fsb,\lform{w}{\phi},\lform{v_0}{\psi}}$ and $\lseq{\mathcal R}{\mathcal T}{\fsa_0,\lform{v}{\psi\imp\chi},\lform{v_0}{\chi}}{\fsb,\lform{w}{\phi}}$ such that $\rable{v}{v_0}$. Consequently, we know that $\lseq{\mathcal R}{\mathcal T}{\fsa}{\fsb}$ is of the form $\lseq{\mathcal R}{\mathcal T}{\fsa_0,\lform{v}{\psi\imp\chi}}{\fsb}$. We also have that $\lseq{\mathcal R}{\mathcal T}{\fsa,\lform{u}{\phi}}{\fsb}$ is of the form $\lseq{\mathcal R}{\mathcal T}{\fsa_0,\lform{v}{\psi\imp\chi},\lform{u}{\phi}}{\fsb}$. Given that $\rable{v}{v}$, we apply Lemma \ref{lem:impl-excr-existsri-alll-inv} on the proof of the latter  sequent to obtain proofs of $\lseq{\mathcal R}{\mathcal T}{\fsa_0,\lform{v}{\psi\imp\chi},\lform{u}{\phi}}{\fsb,\lform{v}{\psi}}$, which we call $S_0$, and $\lseq{\mathcal R}{\mathcal T}{\fsa_0,\lform{v}{\psi\imp\chi},\lform{v}{\chi},\lform{u}{\phi}}{\fsb}$, which we call $S_1$, preserving height. Thus, we proceed as follows where $\pi$ is the first proof displayed.
\begin{center}
\AxiomC{$\lseq{\mathcal R}{\mathcal T}{\fsa_0,\lform{v}{\psi\imp\chi}}{\fsb,\lform{w}{\phi},\lform{v}{\psi}}$}
\AxiomC{$S_0$}
\RightLabel{SIH}
\dashedLine
\BinaryInfC{$\lseq{\mathcal R}{\mathcal T}{\fsa_0,\lform{v}{\psi\imp\chi}}{\fsb,\lform{v}{\psi}}$}
\DisplayProof
\end{center}

\begin{center}
\AxiomC{$\pi$}
\AxiomC{$\lseq{\mathcal R}{\mathcal T}{\fsa_0,\lform{v}{\psi\imp\chi},\lform{v}{\chi}}{\fsb,\lform{w}{\phi}}$}
\AxiomC{$S_1$}
\RightLabel{SIH}
\dashedLine
\BinaryInfC{$\lseq{\mathcal R}{\mathcal T}{\fsa_0,\lform{v}{\psi\imp\chi},\lform{v}{\chi}}{\fsb}$}

\RightLabel{$\impl$}
\BinaryInfC{$\lseq{\mathcal R}{\mathcal T}{\fsa_0,\lform{v}{\psi\imp\chi}}{\fsb}$}
\DisplayProof
\end{center}
Note that both instances of SIH are justified as the sum of the heights of the proofs of their premises is smaller than the one of the initial cut.

\noindent \textbf{(X)} $\mathbf{r_1=\impr:}$ Then $\lseq{\mathcal R}{\mathcal T}{\fsa}{\fsb,\lform{w}{\phi}}$ is of the form $\lseq{\mathcal R}{\mathcal T}{\fsa}{\fsb_0,\lform{v}{\psi\imp\chi}}$ and we have a proof of $\lseq{\mathcal R,\lrel{v}{v_0}}{\mathcal T}{\fsa,\lform{v_0}{\psi}}{\fsb_0,\lform{v_0}{\chi}}$. If $\lform{w}{\phi}$ is not $\lform{v}{\psi\imp\chi}$, then we have a proof of $\lseq{\mathcal R,\lrel{v}{v_0}}{\mathcal T}{\fsa,\lform{v_0}{\psi}}{\fsb_1,\lform{v_0}{\chi},\lform{w}{\phi}}$, which we call $S$, and $\lseq{\mathcal R}{\mathcal T}{\fsa,\lform{u}{\phi}}{\fsb}$ is of the form $\lseq{\mathcal R}{\mathcal T}{\fsa,\lform{u}{\phi}}{\fsb_1,\lform{v}{\psi\imp\chi}}$. Then, we proceed as follows.

\begin{center}
\AxiomC{$S$}

\AxiomC{$\lseq{\mathcal R}{\mathcal T}{\fsa,\lform{u}{\phi}}{\fsb_1,\lform{v}{\psi\imp\chi}}$}
\RightLabel{Lem.\ref{lem:impr-excl-inv}}
\dashedLine
\UnaryInfC{$\lseq{\mathcal R,\lrel{v}{v_0}}{\mathcal T}{\fsa,\lform{u}{\phi},\lform{v_0}{\psi}}{\fsb_1,\lform{v_0}{\chi}}$}

\RightLabel{SIH}
\dashedLine
\BinaryInfC{$\lseq{\mathcal R,\lrel{v}{v_0}}{\mathcal T}{\fsa,\lform{v_0}{\psi}}{\fsb_1,\lform{v_0}{\chi}}$}
\RightLabel{$\impr$}
\UnaryInfC{$\lseq{\mathcal R}{\mathcal T}{\fsa}{\fsb_1,\lform{v}{\psi\imp\chi}}$}
\DisplayProof
\end{center}

If $\lform{w}{\phi}$ is $\lform{v}{\psi\imp\chi}$, then we have a proof of $\lseq{\mathcal R,\lrel{v}{v_0}}{\mathcal T}{\fsa,\lform{v_0}{\psi}}{\fsb,\lform{v_0}{\chi}}$, and $\lseq{\mathcal R}{\mathcal T}{\fsa,\lform{u}{\phi}}{\fsb}$ is of the form $\lseq{\mathcal R}{\mathcal T}{\fsa,\lform{u}{\psi\imp\chi}}{\fsb}$. In this case, we need to consider the shape of $r_2$. If $\lform{u}{\psi\imp\chi}$ is not principal in $r_2$, then we apply the hp-invertibility of $r_2$ to the proof of $\lseq{\mathcal R}{\mathcal T}{\fsa}{\fsb,\lform{v}{\psi\imp\chi}}$ and invoke SIH with the resulting proof to cut $\lform{u}{\psi\imp\chi}$ from the premises of $r_2$, and then reapply $r_2$ to reach our goal. If $\lform{u}{\psi\imp\chi}$ is principal in $r_2$, then the premises of $r_2$ are of the shape $\lseq{\mathcal R}{\mathcal T}{\fsa,\lform{u}{\psi\imp\chi}}{\fsb,\lform{v_1}{\psi}}$ and $\lseq{\mathcal R}{\mathcal T}{\fsa,\lform{u}{\psi\imp\chi},\lform{v_1}{\chi}}{\fsb}$ where $\rable{u}{v_1}$. Note that by the shape of the original cut we know that $\rable{w}{u}$ and by our assumption that $\lform{w}{\phi}$ is $\lform{v}{\psi\imp\chi}$, we have that $w = v$. Therefore, $\rable{v}{v_1}$ since $\rable{v}{u}$ and $\rable{u}{v_1}$, meaning we can proceed as follows where $\pi_0$ and $\pi_1$ are (in this order) the first proofs given.

\begin{center}
\AxiomC{$\lseq{\mathcal R}{\mathcal T}{\fsa}{\fsb,\lform{v}{\psi\imp\chi}}$}
\RightLabel{Lem.\ref{lem:iwk-admiss}}
\dashedLine
\UnaryInfC{$\lseq{\mathcal R}{\mathcal T}{\fsa}{\fsb,\lform{v_1}{\psi},\lform{v}{\psi\imp\chi}}$}
\AxiomC{$\lseq{\mathcal R}{\mathcal T}{\fsa,\lform{u}{\psi\imp\chi}}{\fsb,\lform{v_1}{\psi}}$}
\RightLabel{SIH}
\dashedLine
\BinaryInfC{$\lseq{\mathcal R}{\mathcal T}{\fsa}{\fsb,\lform{v_1}{\psi}}$}
\DisplayProof
\end{center}
\begin{center}
\AxiomC{$\lseq{\mathcal R}{\mathcal T}{\fsa}{\fsb,\lform{v}{\psi\imp\chi}}$}
\RightLabel{Lem.\ref{lem:iwk-admiss}}
\dashedLine
\UnaryInfC{$\lseq{\mathcal R}{\mathcal T}{\fsa,\lform{v_1}{\chi}}{\fsb,\lform{v}{\psi\imp\chi}}$}
\AxiomC{$\lseq{\mathcal R}{\mathcal T}{\fsa,\lform{u}{\psi\imp\chi},\lform{v_1}{\chi}}{\fsb}$}
\RightLabel{SIH}
\dashedLine
\BinaryInfC{$\lseq{\mathcal R}{\mathcal T}{\fsa,\lform{v_1}{\chi}}{\fsb}$}
\DisplayProof
\end{center}
\begin{center}
\AxiomC{$\pi_0$}

\AxiomC{$\lseq{\mathcal R,\lrel{v}{v_0}}{\mathcal T}{\fsa,\lform{v_0}{\psi}}{\fsb,\lform{v_0}{\chi}}$}
\RightLabel{Lem.\ref{lem:brf-brb-admiss}}
\dashedLine
\UnaryInfC{$\lseq{\mathcal R,\lrel{v_1}{v_0}}{\mathcal T}{\fsa,\lform{v_0}{\psi}}{\fsb,\lform{v_0}{\chi}}$}
\RightLabel{Lem.\ref{lem:mrgf-mrgb-admiss}}
\dashedLine
\UnaryInfC{$\lseq{\mathcal R}{\mathcal T}{\fsa,\lform{v_1}{\psi}}{\fsb,\lform{v_1}{\chi}}$}
\AxiomC{$\pi_1$}
\RightLabel{PIH}
\dashedLine
\BinaryInfC{$\lseq{\mathcal R}{\mathcal T}{\fsa,\lform{v_1}{\psi}}{\fsb}$}
\RightLabel{PIH}
\dashedLine
\BinaryInfC{$\lseq{\mathcal R}{\mathcal T}{\fsa}{\fsb}$}
\DisplayProof
\end{center}

\noindent \textbf{(XI)} $\mathbf{r_1=\excl:}$ Then $\lseq{\mathcal R}{\mathcal T}{\fsa}{\fsb,\lform{w}{\phi}}$ is of the form $\lseq{\mathcal R}{\mathcal T}{\fsa_0,\lform{v}{\psi\exc\chi}}{\fsb,\lform{w}{\phi}}$ and we a have proof of $\lseq{\mathcal R,\lrel{v_0}{v}}{\mathcal T}{\fsa_0,\lform{v_0}{\psi}}{\fsb,\lform{w}{\phi},\lform{v_0}{\chi}}$. Consequently, we know that $\lseq{\mathcal R}{\mathcal T}{\fsa}{\fsb}$ is of the form $\lseq{\mathcal R}{\mathcal T}{\fsa_0,\lform{v}{\psi\exc\chi}}{\fsb}$. We also have that $\lseq{\mathcal R}{\mathcal T}{\fsa,\lform{u}{\phi}}{\fsb}$ is of the form $\lseq{\mathcal R}{\mathcal T}{\fsa_0,\lform{v}{\psi\exc\chi},\lform{u}{\phi}}{\fsb}$. 
We apply Lemma \ref{lem:impr-excl-inv} on the proof of the latter  sequent to obtain a proof of $\lseq{\mathcal R,\lrel{v_0}{v}}{\mathcal T}{\fsa_0,\lform{v_0}{\psi},\lform{u}{\phi}}{\fsb,\lform{v_0}{\chi}}$, which we call $S$. Thus, we proceed as follows.
\begin{center}
\AxiomC{$\lseq{\mathcal R,\lrel{v_0}{v}}{\mathcal T}{\fsa_0,\lform{v_0}{\psi}}{\fsb,\lform{w}{\phi},\lform{v_0}{\chi}}$}

\AxiomC{$S$}

\RightLabel{SIH}
\dashedLine
\BinaryInfC{$\lseq{\mathcal R,\lrel{v_0}{v}}{\mathcal T}{\fsa_0,\lform{v_0}{\psi}}{\fsb,\lform{v_0}{\chi}}$}
\RightLabel{$\excl$}
\UnaryInfC{$\lseq{\mathcal R}{\mathcal T}{\fsa_0,\lform{v}{\psi\exc\chi}}{\fsb}$}
\DisplayProof
\end{center}
Note that the instance of SIH is justified as the sum of the heights of the proofs of the premises is smaller than the one of the initial cut.

\noindent \textbf{(XII)} $\mathbf{r_1=\excr:}$ Then $\lseq{\mathcal R}{\mathcal T}{\fsa}{\fsb,\lform{w}{\phi}}$ is of the form $\lseq{\mathcal R}{\mathcal T}{\fsa}{\fsb_0,\lform{v}{\psi\exc\chi}}$ and we have proofs of $\lseq{\mathcal R}{\mathcal T}{\fsa}{\fsb_0,\lform{v}{\psi\exc\chi},\lform{v_0}{\psi}}$ and $\lseq{\mathcal R}{\mathcal T}{\fsa,\lform{v_0}{\chi}}{\fsb_0,\lform{v}{\psi\exc\chi}}$ where $\rable{v_0}{v}$. If $\lform{w}{\phi}$ is not $\lform{v}{\psi\exc\chi}$, then we have proofs of $\lseq{\mathcal R}{\mathcal T}{\fsa}{\fsb_1,\lform{v}{\psi\exc\chi},\lform{v_0}{\psi},\lform{w}{\phi}}$, which we call $S_0$, and $\lseq{\mathcal R}{\mathcal T}{\fsa,\lform{v_0}{\chi}}{\fsb_1,\lform{v}{\psi\exc\chi},\lform{w}{\phi}}$, which we call $S_1$, and $\lseq{\mathcal R}{\mathcal T}{\fsa,\lform{u}{\phi}}{\fsb}$ is of the form $\lseq{\mathcal R}{\mathcal T}{\fsa,\lform{u}{\phi}}{\fsb_1,\lform{v}{\psi\exc\chi}}$. Then, we proceed as follows where $\pi$ is the first proof displayed.

\begin{center}
\AxiomC{$S_0$}
\AxiomC{$\lseq{\mathcal R}{\mathcal T}{\fsa,\lform{u}{\phi}}{\fsb_1,\lform{v}{\psi\exc\chi}}$}
\RightLabel{Lem.\ref{lem:iwk-admiss}}
\dashedLine
\UnaryInfC{$\lseq{\mathcal R}{\mathcal T}{\fsa,\lform{u}{\phi}}{\fsb_1,\lform{v}{\psi\exc\chi},\lform{v_0}{\psi}}$}
\RightLabel{SIH}
\dashedLine
\BinaryInfC{$\lseq{\mathcal R}{\mathcal T}{\fsa}{\fsb_1,\lform{v}{\psi\exc\chi},\lform{v_0}{\psi}}$}
\DisplayProof
\end{center}
\begin{center}
\AxiomC{$\pi$}

\AxiomC{$S_1$}
\AxiomC{$\lseq{\mathcal R}{\mathcal T}{\fsa,\lform{u}{\phi}}{\fsb_1,\lform{v}{\psi\exc\chi}}$}
\RightLabel{Lem.\ref{lem:iwk-admiss}}
\dashedLine
\UnaryInfC{$\lseq{\mathcal R}{\mathcal T}{\fsa,\lform{v_0}{\chi},\lform{u}{\phi}}{\fsb_1,\lform{v}{\psi\exc\chi}}$}
\RightLabel{SIH}
\dashedLine
\BinaryInfC{$\lseq{\mathcal R}{\mathcal T}{\fsa,\lform{v_0}{\chi}}{\fsb_1,\lform{v}{\psi\exc\chi}}$}

\RightLabel{$\excr$}
\BinaryInfC{$\lseq{\mathcal R}{\mathcal T}{\fsa}{\fsb_1,\lform{v}{\psi\exc\chi}}$}
\DisplayProof
\end{center}

If $\lform{w}{\phi}$ is $\lform{v}{\psi\exc\chi}$, then we have proofs of $\lseq{\mathcal R}{\mathcal T}{\fsa}{\fsb,\lform{v}{\psi\exc\chi},\lform{v_0}{\psi}}$ and $\lseq{\mathcal R}{\mathcal T}{\fsa,\lform{v_0}{\chi}}{\fsb,\lform{v}{\psi\exc\chi}}$, and $\lseq{\mathcal R}{\mathcal T}{\fsa,\lform{u}{\phi}}{\fsb}$ is of the form $\lseq{\mathcal R}{\mathcal T}{\fsa,\lform{u}{\psi\exc\chi}}{\fsb}$. In this case, we need to consider the shape of $r_2$. If $\lform{u}{\psi\exc\chi}$ is not principal in $r_2$, then we apply the hp-invertibility of $r_2$ to the proof of $\lseq{\mathcal R}{\mathcal T}{\fsa}{\fsb,\lform{v}{\psi\exc\chi}}$ and then invoke SIH with the resulting proof to cut $\lform{u}{\psi\exc\chi}$ from the premises of $r_2$, and then reapply $r_2$ to reach our goal. If $\lform{u}{\psi\exc\chi}$ is principal in $r_2$, then the premise of $r_2$ is of the shape $\lseq{\mathcal R,\lrel{v_1}{u}}{\mathcal T}{\fsa,\lform{v_1}{\psi}}{\fsb,\lform{v_1}{\chi}}$. Note that by shape of the original cut we have that $\rable{w}{u}$, by the fact that $\lform{w}{\phi}$ is $\lform{v}{\psi\exc\chi}$, we know that $w = v$, and by what was said above, $\rable{v_0}{v}$. Therefore, $\rable{v_0}{u}$, meaning we can proceed as follows where $\pi_0$ and $\pi_1$ are (in this order) the first proofs given.
\begin{center}
\AxiomC{$\lseq{\mathcal R}{\mathcal T}{\fsa}{\fsb,\lform{v}{\psi\exc\chi}}$}
\RightLabel{Lem.\ref{lem:iwk-admiss}}
\dashedLine
\UnaryInfC{$\lseq{\mathcal R}{\mathcal T}{\fsa}{\fsb,\lform{v_0}{\psi},\lform{v}{\psi\exc\chi}}$}
\AxiomC{$\lseq{\mathcal R}{\mathcal T}{\fsa,\lform{u}{\psi\exc\chi}}{\fsb,\lform{v_0}{\psi}}$}
\RightLabel{SIH}
\dashedLine
\BinaryInfC{$\lseq{\mathcal R}{\mathcal T}{\fsa}{\fsb,\lform{v_0}{\psi}}$}
\DisplayProof
\end{center}
\begin{center}
\AxiomC{$\lseq{\mathcal R}{\mathcal T}{\fsa}{\fsb,\lform{v}{\psi\exc\chi}}$}
\RightLabel{Lem.\ref{lem:iwk-admiss}}
\dashedLine
\UnaryInfC{$\lseq{\mathcal R}{\mathcal T}{\fsa,\lform{v_0}{\chi}}{\fsb,\lform{v}{\psi\exc\chi}}$}
\AxiomC{$\lseq{\mathcal R}{\mathcal T}{\fsa,\lform{u}{\psi\exc\chi},\lform{v_0}{\chi}}{\fsb}$}
\RightLabel{SIH}
\dashedLine
\BinaryInfC{$\lseq{\mathcal R}{\mathcal T}{\fsa,\lform{v_0}{\chi}}{\fsb}$}
\RightLabel{Lem.\ref{lem:iwk-admiss}}
\dashedLine
\UnaryInfC{$\lseq{\mathcal R}{\mathcal T}{\fsa,\lform{v_0}{\psi},\lform{v_0}{\chi}}{\fsb}$}
\DisplayProof
\end{center}
\begin{center}
\AxiomC{$\pi_0$}

\AxiomC{$\lseq{\mathcal R,\lrel{v_{1}}{u}}{\mathcal T}{\fsa,\lform{v_1}{\psi}}{\fsb,\lform{v_1}{\chi}}$}
\RightLabel{Lem.\ref{lem:brf-brb-admiss}}
\dashedLine
\UnaryInfC{$\lseq{\mathcal R,\lrel{v_1}{v_0}}{\mathcal T}{\fsa,\lform{v_1}{\psi}}{\fsb,\lform{v_1}{\chi}}$}
\RightLabel{Lem.\ref{lem:mrgf-mrgb-admiss}}
\dashedLine
\UnaryInfC{$\lseq{\mathcal R}{\mathcal T}{\fsa,\lform{v_0}{\psi}}{\fsb,\lform{v_0}{\chi}}$}
\AxiomC{$\pi_1$}
\RightLabel{PIH}
\dashedLine
\BinaryInfC{$\lseq{\mathcal R}{\mathcal T}{\fsa,\lform{v_0}{\psi}}{\fsb}$}

\RightLabel{PIH}
\dashedLine
\BinaryInfC{$\lseq{\mathcal R}{\mathcal T}{\fsa}{\fsb}$}
\DisplayProof
\end{center}

\noindent \textbf{(XIII)} $\mathbf{r_1=\existsl:}$ Then $\lseq{\mathcal R}{\mathcal T}{\fsa}{\fsb,\lform{w}{\phi}}$ is of the form $\lseq{\mathcal R}{\mathcal T}{\fsa_0,\lform{v}{\exists x\psi}}{\fsb,\lform{w}{\phi}}$ and we a have proof of $\lseq{\mathcal R}{\mathcal T,\lterm{v}{x}}{\fsa_0,\lform{v}{\psi(y/x)}}{\fsb,\lform{w}{\phi}}$, which we call $S$. Consequently, we know that $\lseq{\mathcal R}{\mathcal T}{\fsa}{\fsb}$ is of the form $\lseq{\mathcal R}{\mathcal T}{\fsa_0,\lform{v}{\exists\psi}}{\fsb}$. We also have that $\lseq{\mathcal R}{\mathcal T}{\fsa,\lform{u}{\phi}}{\fsb}$ is of the form $\lseq{\mathcal R}{\mathcal T}{\fsa_0,\lform{v}{\exists x\psi},\lform{u}{\phi}}{\fsb}$. 
We apply Lemma \ref{lem:existsl-allr-inv} on the proof of the latter  sequent to obtain a proof of $\lseq{\mathcal R}{\mathcal T,\lterm{v}{z}}{\fsa_0,\lform{v}{\psi(z/x)},\lform{u}{\phi}}{\fsb}$. Thus, we proceed as follows.
\begin{center}
\AxiomC{$S$}

\AxiomC{$\lseq{\mathcal R}{\mathcal T,\lterm{v}{z}}{\fsa_0,\lform{v}{\psi(z/x)},\lform{u}{\phi}}{\fsb}$}
\dashedLine
\RightLabel{Lem.\ref{lem:psub-admiss}}
\UnaryInfC{$\lseq{\mathcal R}{\mathcal T,\lterm{v}{y}}{\fsa_0,\lform{v}{\psi(y/x)},\lform{u}{\phi}}{\fsb}$}

\RightLabel{SIH}
\dashedLine
\BinaryInfC{$\lseq{\mathcal R}{\mathcal T,\lterm{v}{y}}{\fsa_0,\lform{v}{\psi(y/x)}}{\fsb}$}
\RightLabel{$\existsl$}
\UnaryInfC{$\lseq{\mathcal R}{\mathcal T}{\fsa_0,\lform{v}{\exists\psi}}{\fsb}$}
\DisplayProof
\end{center}
Note that the instance of SIH is justified as the sum of the heights of the proofs of the premises is smaller than the one of the initial cut.

\noindent \textbf{(XIV)} $\mathbf{r_1=\existsr:}$ Then $\lseq{\mathcal R}{\mathcal T}{\fsa}{\fsb,\lform{w}{\phi}}$ is of the form $\lseq{\mathcal R}{\mathcal T}{\fsa}{\fsb_0,\lform{v}{\exists x\psi}}$ and we have a proof of $\lseq{\mathcal R}{\mathcal T}{\fsa}{\fsb_0,\lform{v}{\exists x\psi},\lform{v}{\psi(t/x)}}$ where $t$ is available for $v$.

If $\lform{w}{\phi}$ is not $\lform{v}{\exists x\psi}$, then we have a proof of $\lseq{\mathcal R}{\mathcal T}{\fsa}{\fsb_1,\lform{v}{\exists x\psi},\lform{v}{\psi(t/x)},\lform{w}{\phi}}$, which we call $S$, and $\lseq{\mathcal R}{\mathcal T}{\fsa,\lform{u}{\phi}}{\fsb}$ is of the form $\lseq{\mathcal R}{\mathcal T}{\fsa,\lform{u}{\phi}}{\fsb_1,\lform{v}{\exists x\psi}}$. Then, we proceed as follows.

\begin{center}
\AxiomC{$S$}

\AxiomC{$\lseq{\mathcal R}{\mathcal T}{\fsa,\lform{u}{\phi}}{\fsb_1,\lform{v}{\exists x\psi}}$}
\RightLabel{Lem.\ref{lem:impl-excr-existsri-alll-inv}}
\dashedLine
\UnaryInfC{$\lseq{\mathcal R}{\mathcal T}{\fsa,\lform{u}{\phi}}{\fsb_1,\lform{v}{\exists x\psi},\lform{v}{\psi(t/x)}}$}

\RightLabel{SIH}
\dashedLine
\BinaryInfC{$\lseq{\mathcal R}{\mathcal T}{\fsa}{\fsb_1,\lform{v}{\exists x\psi},\lform{v}{\psi(t/x)}}$}
\RightLabel{$\existsr$}
\UnaryInfC{$\lseq{\mathcal R}{\mathcal T}{\fsa}{\fsb_1,\lform{v}{\exists x\psi}}$}
\DisplayProof
\end{center}

If $\lform{w}{\phi}$ is $\lform{v}{\exists x\psi}$, then we have proof a of $\lseq{\mathcal R}{\mathcal T}{\fsa}{\fsb,\lform{v}{\exists x\psi},\lform{v}{\psi(t/x)}}$, and $\lseq{\mathcal R}{\mathcal T}{\fsa,\lform{u}{\phi}}{\fsb}$ is of the form $\lseq{\mathcal R}{\mathcal T}{\fsa,\lform{u}{\exists x\psi}}{\fsb}$. In this case, we need to consider the shape of $r_2$. If $\lform{u}{\exists x\psi}$ is not principal in $r_2$, then we apply the hp-invertibility of $r_2$ to the proof of $\lseq{\mathcal R}{\mathcal T}{\fsa}{\fsb,\lform{v}{\exists x\psi}}$ and then invoke SIH with the resulting proof to cut $\lform{u}{\exists x\psi}$ with the premises of $r_2$, and then reapply $r_2$ to reach our goal. If $\lform{u}{\exists x\psi}$ is principal in $r_2$, then the premise of $r_2$ is of the shape $\lseq{\mathcal R}{\mathcal T,\lterm{v}{y}}{\fsa,\lform{v}{\psi(y/x)}}{\fsb}$ where $y$ is fresh. Then, we proceed as follows where $\pi$ is the first proof given and $x_0,\dots,x_n$ are all the variables appearing in $t$.
\begin{center}
\AxiomC{$\lseq{\mathcal R}{\mathcal T}{\fsa}{\fsb,\lform{v}{\exists x\psi},\lform{v}{\psi(t/x)}}$}
\AxiomC{$\lseq{\mathcal R}{\mathcal T}{\fsa,\lform{u}{\exists x\psi}}{\fsb}$}
\RightLabel{Lem.\ref{lem:iwk-admiss}}
\dashedLine
\UnaryInfC{$\lseq{\mathcal R}{\mathcal T}{\fsa,\lform{u}{\exists x\psi}}{\fsb,\lform{v}{\psi(t/x)}}$}
\RightLabel{SIH}
\dashedLine
\BinaryInfC{$\lseq{\mathcal R}{\mathcal T}{\fsa}{\fsb,\lform{v}{\psi(t/x)}}$}
\DisplayProof
\end{center}

\begin{center}
\AxiomC{$\pi$}

\AxiomC{$\lseq{\mathcal R}{\mathcal T,\lterm{v}{y}}{\fsa,\lform{v}{\psi(y/x)}}{\fsb}$}
\RightLabel{Lem.\ref{lem:psub-admiss}}
\dashedLine
\UnaryInfC{$\lseq{\mathcal R}{\mathcal T,\lterm{v}{x_0},\dots,\lterm{v}{x_n}}{\fsa,\lform{v}{\psi(t/x)}}{\fsb}$}
\RightLabel{Lem.\ref{lem:ndr-admiss}}
\dashedLine
\UnaryInfC{$\lseq{\mathcal R}{\mathcal T}{\fsa,\lform{v}{\psi(t/x)}}{\fsb}$}

\RightLabel{PIH}
\dashedLine
\BinaryInfC{$\lseq{\mathcal R}{\mathcal T}{\fsa}{\fsb}$}
\DisplayProof
\end{center}
Note that the step involving Lemma \ref{lem:ndr-admiss} is justified as $t$ is available for $v$, which implies that we can push all its variables to the original labels making $t$ available for $v$.

\noindent \textbf{(XV)} $\mathbf{r_1=\alll:}$ Then $\lseq{\mathcal R}{\mathcal T}{\fsa}{\fsb,\lform{w}{\phi}}$ is of the form $\lseq{\mathcal R}{\mathcal T}{\fsa_0,\lform{v}{\forall x\psi}}{\fsb,\lform{w}{\phi}}$ and we a have proof of $\lseq{\mathcal R}{\mathcal T}{\fsa_0,\lform{v}{\forall x\psi},\lform{v_0}{\psi(t/x)}}{\fsb,\lform{w}{\phi}}$ where $\rable{v}{v_0}$ and $t$ is available for $v_0$. Consequently, we know that $\lseq{\mathcal R}{\mathcal T}{\fsa}{\fsb}$ is of the form $\lseq{\mathcal R}{\mathcal T}{\fsa_0,\lform{v}{\forall x\psi}}{\fsb}$. We also have that $\lseq{\mathcal R}{\mathcal T}{\fsa,\lform{u}{\phi}}{\fsb}$ is of the form $\lseq{\mathcal R}{\mathcal T}{\fsa_0,\lform{v}{\forall x\psi},\lform{u}{\phi}}{\fsb}$. 
We apply Lemma \ref{lem:impl-excr-existsri-alll-inv} on the proof of the latter sequent to obtain a proof of $\lseq{\mathcal R}{\mathcal T}{\fsa_0,\lform{v}{\forall x\psi},\lform{v_0}{\psi(t/x)},\lform{u}{\phi}}{\fsb}$, which we call $S$. Thus, we proceed as follows.
\begin{center}
\AxiomC{$\lseq{\mathcal R}{\mathcal T}{\fsa_0,\lform{v}{\forall x\psi},\lform{v_0}{\psi(t/x)}}{\fsb,\lform{w}{\phi}}$}

\AxiomC{$S$}

\RightLabel{SIH}
\dashedLine
\BinaryInfC{$\lseq{\mathcal R}{\mathcal T}{\fsa_0,\lform{v}{\forall x\psi},\lform{v_0}{\psi(t/x)}}{\fsb}$}
\RightLabel{$\alll$}
\UnaryInfC{$\lseq{\mathcal R}{\mathcal T}{\fsa_0,\lform{v}{\forall\psi}}{\fsb}$}
\DisplayProof
\end{center}
Note that the instance of SIH is justified as the sum of the heights of the proofs of the premises is smaller than the one of the initial cut.

\noindent \textbf{(XVI)} $\mathbf{r_1=\allr:}$ Then $\lseq{\mathcal R}{\mathcal T}{\fsa}{\fsb,\lform{w}{\phi}}$ is of the form $\lseq{\mathcal R}{\mathcal T}{\fsa}{\fsb_0,\lform{v}{\forall x\psi}}$ and we have a proof of $\lseq{\mathcal R,\lrel{v}{v_0}}{\mathcal T,\lterm{v_0}{y}}{\fsa}{\fsb_0,\lform{v_0}{\psi(y/x)}}$, which we call $S$, where $y$ is fresh.

If $\lform{w}{\phi}$ is not $\lform{v}{\forall x\psi}$, then we have a proof of $\lseq{\mathcal R,\lrel{v}{v_0}}{\mathcal T,\lterm{v_0}{y}}{\fsa}{\fsb_1,\lform{v_0}{\psi(y/x)},\lform{w}{\phi}}$, and $\lseq{\mathcal R}{\mathcal T}{\fsa,\lform{u}{\phi}}{\fsb}$ is of the form $\lseq{\mathcal R}{\mathcal T}{\fsa,\lform{u}{\phi}}{\fsb_1,\lform{v}{\forall x\psi}}$. Then, we proceed as follows where $\pi$ is the first proof displayed.

\begin{center}
\AxiomC{$S$}

\AxiomC{$\lseq{\mathcal R}{\mathcal T}{\fsa,\lform{u}{\phi}}{\fsb_1,\lform{v}{\forall x\psi}}$}
\RightLabel{Lem.\ref{lem:existsl-allr-inv}}
\dashedLine
\UnaryInfC{$\lseq{\mathcal R,\lrel{v}{v_0}}{\mathcal T,\lterm{v_0}{y}}{\fsa,\lform{u}{\phi}}{\fsb_1,\lform{v_0}{\psi(y/x)}}$}

\RightLabel{SIH}
\dashedLine
\BinaryInfC{$\lseq{\mathcal R,\lrel{v}{v_0}}{\mathcal T,\lterm{v_0}{y}}{\fsa}{\fsb_1,\lform{v_0}{\psi(y/x)}}$}
\RightLabel{$\allr$}
\UnaryInfC{$\lseq{\mathcal R}{\mathcal T}{\fsa}{\fsb_1,\lform{v}{\forall x\psi}}$}
\DisplayProof
\end{center}

If $\lform{w}{\phi}$ is $\lform{v}{\forall x\psi}$, then we have proof a of $\lseq{\mathcal R,\lrel{v}{v_0}}{\mathcal T,\lterm{v_0}{y}}{\fsa}{\fsb,\lform{v_0}{\psi(y/x)}}$, and $\lseq{\mathcal R}{\mathcal T}{\fsa,\lform{u}{\phi}}{\fsb}$ is of the form $\lseq{\mathcal R}{\mathcal T}{\fsa,\lform{u}{\forall x\psi}}{\fsb}$. In this case, we need to consider the shape of $r_2$. If $\lform{u}{\forall x\psi}$ is not principal in $r_2$, then we apply the hp-invertibility of $r_2$ to the proof of $\lseq{\mathcal R}{\mathcal T}{\fsa}{\fsb,\lform{v}{\forall x\psi}}$ and then invoke SIH with the resulting proof to cut $\lform{u}{\forall x\psi}$ with the premises of $r_2$, and then reapply $r_2$ to reach our goal. If $\lform{u}{\forall x\psi}$ is principal in $r_2$, then the premise of $r_2$ is of the shape $\lseq{\mathcal R}{\mathcal T}{\fsa,\lform{u}{\forall x\psi},\lform{v_1}{\psi(t/x)}}{\fsb}$, which we call $S$, where $\rable{u}{v_1}$ and $t$ is available for $v_1$. Then, we proceed as follows where $\pi$ is the first proof given and $x_0,\dots,x_n$ are all the variables appearing in $t$.
\begin{center}
\AxiomC{$\lseq{\mathcal R,\lrel{v}{v_0}}{\mathcal T,\lterm{v_0}{y}}{\fsa}{\fsb,\lform{v_0}{\psi(y/x)}}$}
\RightLabel{Lem.\ref{lem:mrgf-mrgb-admiss}}
\dashedLine
\UnaryInfC{$\lseq{\mathcal R}{\mathcal T,\lterm{v}{y}}{\fsa}{\fsb,\lform{v}{\psi(y/x)}}$}
\RightLabel{Lem.\ref{lem:psub-admiss}}
\dashedLine
\UnaryInfC{$\lseq{\mathcal R}{\mathcal T,\lterm{v}{x_0},\dots,\lterm{v}{x_n}}{\fsa}{\fsb,\lform{v}{\psi(t/x)}}$}
\RightLabel{Lem.\ref{lem:ndr-admiss}}
\dashedLine
\UnaryInfC{$\lseq{\mathcal R}{\mathcal T}{\fsa}{\fsb,\lform{v}{\psi(t/x)}}$}
\DisplayProof
\end{center}

\begin{center}
\AxiomC{$\pi$}

\AxiomC{$\lseq{\mathcal R}{\mathcal T}{\fsa}{\fsb,\lform{v}{\forall x\psi}}$}
\RightLabel{Lem.\ref{lem:iwk-admiss}}
\dashedLine
\UnaryInfC{$\lseq{\mathcal R}{\mathcal T}{\fsa,\lform{v_1}{\psi(t/x)}}{\fsb,\lform{v}{\forall x\psi}}$}
\AxiomC{$S$}
\RightLabel{SIH}
\dashedLine
\BinaryInfC{$\lseq{\mathcal R}{\mathcal T}{\fsa,\lform{v_1}{\psi(t/x)}}{\fsb}$}

\RightLabel{PIH}
\dashedLine
\BinaryInfC{$\lseq{\mathcal R}{\mathcal T}{\fsa}{\fsb}$}
\DisplayProof
\end{center}
Note that the step involving Lemma \ref{lem:ndr-admiss} is justified as $t$ is available for $v$, which implies that we can push all its variables to the original labels making $t$ available for $v$. In addition to that, the use of PIH is justified by the holding of $\rable{v}{v_1}$ which we infer from $\rable{v}{u}$ and $\rable{u}{v_1}$.
\end{proof}

\end{document}